\documentclass[11pt]{article} 

\usepackage[top=1in,bottom=1in,left=1in,right=1in]{geometry} 
\RequirePackage{amsthm,amsmath,amsfonts,amssymb}
\RequirePackage[numbers]{natbib}
\RequirePackage[colorlinks,allcolors=blue,urlcolor=blue]{hyperref}
\RequirePackage{graphicx}

\usepackage{color, amsmath, amssymb, amsbsy, amsthm, graphicx, bbm, amsfonts, bm, dsfont, courier, subfig}
\usepackage[toc,page]{appendix}
\usepackage{booktabs}
\usepackage[flushleft]{threeparttable}
\usepackage{latexsym}
\usepackage[dvipsnames]{xcolor}
\usepackage{comment}
\usepackage{wrapfig}
\usepackage{graphicx}
\usepackage{graphics}
\usepackage{algorithm}
\usepackage{algorithmic}
\usepackage{setspace}

\usepackage{soul}
\usepackage{threeparttable,booktabs}
\usepackage{etoolbox}
\appto\TPTnoteSettings{\footnotesize}
\newtheoremstyle{theoremstyle}
{\topsep} 
{\topsep} 
{\itshape} 
{} 
{} 
{} 
{.5em} 
{\color{black}\ifthenelse{\equal{#3}{}}{{\bfseries #1 #2}}{{\bfseries #1 #2 (#3)}}}
\newtheoremstyle{theoremstylealt}
{\topsep} 
{\topsep} 
{\itshape} 
{} 
{} 
{} 
{.5em} 
{\color{black}\ifthenelse{\equal{#3}{}}{{\bfseries #1 #2$^\prime$}}{{\bfseries #1 #2$^\prime$ (#3)}}}
\newtheoremstyle{examplestyle}
{\topsep} 
{\topsep} 
{} 
{} 
{} 
{} 
{.5em} 
{\color{black}\ifthenelse{\equal{#3}{}}{{\bfseries #1 #2}}{{\bfseries #1 #2 (#3)}}}
\theoremstyle{theoremstyle}\newtheorem{thm}{Theorem}
\theoremstyle{theoremstylealt}
\theoremstyle{theoremstyle}     
\theoremstyle{theoremstyle}\newtheorem{lem}{Lemma}  
\theoremstyle{theoremstyle}\newtheorem{coro}{Corollary}        
\theoremstyle{theoremstyle}\newtheorem{prop}{Proposition}
\theoremstyle{theoremstyle}\newtheorem{assumption}{Assumption}
\theoremstyle{theoremstylealt}
\theoremstyle{theoremstyle}

\theoremstyle{theoremstyle}

\theoremstyle{theoremstyle}

\theoremstyle{examplestyle}
\theoremstyle{examplestyle}\newtheorem{remark}{Remark}
\theoremstyle{examplestyle}

\usepackage{xr}
\renewcommand{\epsilon}{\varepsilon}

\def \hat{\widehat}

\usepackage{titletoc}
\usepackage{enumitem}

\usepackage{array}
\newcolumntype{H}{>{\setbox0=\hbox\bgroup}c<{\egroup}@{}}

\usepackage{minitoc}
\renewcommand \thepart{}

\begin{document}

\doparttoc 
\faketableofcontents 

\title{Efficient Targeted Learning of Heterogeneous Treatment Effects for Multiple Subgroups
}
\author{
Waverly Wei\thanks{Division of Biostatistics, University of California, Berkeley.} \and
Maya Petersen\footnotemark[1] \and
Mark J van der Laan\footnotemark[1]\and Zeyu Zheng \thanks{Department of Industrial Engineering and Operations Research, University of California, Berkeley}\and Chong Wu  \thanks{Department of Biostatistics, The University of Texas MD Anderson Cancer Center.} \footnotemark[4] \and Jingshen Wang \footnotemark[1] \thanks{Correspondence: jingshenwang@berkeley.edu, cwu18@mdanderson.org}
} 

\date{}

\maketitle

		\begin{abstract}
	In biomedical science, analyzing treatment effect heterogeneity plays an essential role in assisting personalized medicine. The main goals of analyzing treatment effect heterogeneity include estimating treatment effects in clinically relevant subgroups and predicting whether a patient subpopulation might benefit from a particular treatment. Conventional approaches often evaluate the subgroup treatment effects via parametric modeling and can thus be susceptible to model mis-specifications. In this manuscript, we take a model-free semiparametric perspective and aim to efficiently evaluate the heterogeneous treatment effects of multiple subgroups simultaneously under the one-step targeted maximum-likelihood estimation (TMLE) framework. When the number of subgroups is large, we further expand this path of research by looking at a variation of the one-step TMLE that 
is robust to the presence of small estimated propensity scores in finite samples. From our simulations, our method demonstrates substantial finite sample improvements compared to conventional methods. In a case study, our method unveils the potential treatment effect heterogeneity of rs12916-T allele (a proxy for statin usage) in decreasing Alzheimer's disease risk.
 	 \\ \bigskip 
		
		\end{abstract}
		
		\vskip 2cm
		\begin{center}\bfseries
		\end{center}

\clearpage
\doublespacing

\section{Introduction}
\subsection{Motivation and our contribution}

In biomedical studies with observational data, investigators often aim to assess the heterogeneity of treatment effects in subpopulations of patients. Such analyses may provide useful information for patient care and for future medical research. For example, existing studies suggest that statins--a class of commonly prescribed  coronary artery disease (CAD) drugs for lowering low-density lipoprotein cholesterol concentration--may reduce Alzheimer's disease (AD) risk in some, but not all population (\citealp{zissimopoulos2017sex}). Understanding the heterogeneous treatment effects of statin usage may provide new insights for personalizing drug prescriptions to prevent AD.

In this manuscript, we aim to make valid inference on heterogeneous treatment effects in a user-supplied family of subgroups after adjusting for potential confounding factors with state-of-the-art machine learning algorithms. Motivated by our case study (Section \ref{sec: real-data}), we work under the setting that the treatment and outcome variables are binary. The extension of our method to continuous outcomes is discussed in Web Appendix \ref{appendix:continuous-outcome}. Our parameter of interest includes relative risk under a treatment versus a control in $d$ pre-specified subgroups of interest: 
$\bm{\alpha}_{\textbf{RR}} = (\alpha_{\text{RR},1},\dots \alpha_{\text{RR},d})^\intercal,\   \alpha_{\text{RR},j} =  \frac{{P}\big( Y(1)=1|X\in\mathcal{A}_j\big)}{{P}\big( Y(0)=1|X\in\mathcal{A}_j\big)}, \  j=1, \ldots, d, $
where $ {P}\big( Y(1)=1|X\in\mathcal{A}_j\big)$ (or $ {P}\big( Y(0)=1|X\in\mathcal{A}_j\big)$) is the conditional expectations of the potential outcome under treatment (or control) evaluated in the subgroup $\mathcal{A}_j$. We denote $X\in\mathbb{R}^p$ as the potential confounders, and denote  $\{\mathcal{A}_j\}_{j=1}^d$ as pre-specified possibly overlapped subgroups. We work under the classical semi-parametric inference framework, in which we aim to make inference on the low-dimensional target parameter $\bm{\alpha}_{\textbf{RR}} $ in the presence of high-dimensional nuisance parameters (see Section \ref{sec:conditions-notations} for rigorous statements). 

In this context, two potential issues emerge when one evaluates the treatment effects for multiple subgroups. On the one hand, while a commonly used method is to serially divide individuals into subgroups based on relevant pre-treatment characteristics and then estimate the treatment effect in each subgroup with either the (augmented) inverse propensity score weighting (\citealp{rosen1983overlap}) or the targeted maximum likelihood estimator (TMLE) (\citealp{laantmle2006}), this ``one-group-at-a-time" approach can be computationally costly (see Section \ref{sec:limitation-of-one-step-TMLE} for a concrete example). On the other hand,
when the estimated propensity scores or subgroup proportions are close to zero or one in finite
samples (a phenomenon refereed to as ``practical positivity violation" in \citep{petersen2012diagnosing}), such approaches can be numerically unstable due to the inverse propensity score or inverse subgroup proportion weights tending to infinity.

To address such potential issues, we work with a one-step targeted maximum likelihood estimator that ``targets" multiple subgroup treatment effects simultaneously.
The so-called ``targeting" step here involves fluctuating the initial plug-in estimator of the nuisance parameters in semiparametric models in directions which maximally adjust those initial estimates per change in the log-likelihood. Furthermore, we propose a variation of the one-step TMLE that not only targets multiple subgroups simultaneously but is also robust to the presence of small estimated propensity scores in finite samples. Deviating from the mainstream literature on the targeted learning,  we also look into the problem from an optimization point of view, where we further demonstrate that such a variation of the one-step TMLE can be viewed as a reparametrized dual formulation of the primal optimization problem. 

From our theoretical investigations, we show that the proposed estimator for multiple subgroup treatment effects attains the semiparametric efficiency bound, and it converges in distribution to a multivariate Gaussian distribution when the sample size becomes large. This result thus allows us to construct valid simultaneous confidence intervals and develop powerful multiple testing procedures fully utilizing the joint dependence among the subgroup specific test statistics. In addition to these large sample guarantees, through simulation studies, we demonstrate that the proposed estimator has substantial finite sample improvements relative to either applying the classical targeted learning approach \citep{van2011targeted} or the ``double machine learning" frequently adopted in the econometrics literature \citep{chernozhukov2017double}. From an application point of view, leveraging the observational data collected from the UK Biobank study, we analyze the differential effects of inheriting rs12916-T allele (a proxy for statin usage) in decreasing AD risk across multiple subgroups.

\subsection{Related literature}

The proposed method builds on the foundation of the targeted learning framework which is, broadly speaking, a meta-learning framework allowing various machine learning algorithms to enter the process of estimating desired target parameters (\citealp{van2011targeted}). \citet{laantmle2006} propose the original version of TMLE, which uses maximum likelihood in a least favourable direction and then performs $k$‐step updates using the estimated scores, in an effort to better estimate the target parameter. \citet{zheng2010asymptotic} introduce the cross-validated TMLE, which relaxes the stringent Donsker condition via sample splitting for the initial estimation of the nuisance parameters.
\citet{laan2010ctmle} further advances the original TMLE by designing different sets of candidate scores. A recent advancement in the targeted learning framework is the one-step TMLE (\citealp{van2016one}), which adopts a ``universal least favorable submodel" to avoid excessive data fitting in the locally least favorable submodel.
In terms of estimating a vector of multi-dimensional parameters with TMLE, seminal works by \cite{van2011targeted} and  \cite{van2016one} develop a universal canonical one-dimensional submodel such that the one-step TMLE, only maximizing the log-likelihood over a univariate parameter, solves the multivariate efficient influence curve equation. A recent work (\citealp{levy2021fundamental}) adopts this general TMLE approach for estimating the variance of the stratum-specific treatment effect functions.  We also note that the general strategy of TMLE that targets multi-dimensional parameters have also been discussed for estimating survival curves (see, \citet{van2018targeted}, Chapter 5 for example).

Our proposal contributes to the semiparametric statistics literature. Early work on semiparametric statistics (\citealp{newey1990semiparametric}) provides general efficiency results for the development of semiparametric estimators. Based on these efficiency results,  \citet{robins1992estequation} propose a general estimating equation approach that solves for the parameter of interest by setting the efficient score equations to zero. The estimating equation approach is further discussed in \citet{laan2003estequation}. \citet{bickel1993onestep} develop a one-step estimator that adds the empirical average of the efficient influence function to an initial estimator.   \citet{newey1994asymptotic} advances the semiparametric efficiency results by accounting for the nonparametric estimation of nuisance parameters. \citet{vaart2000asymptotic} discusses the use of maximum likelihood estimator and parametric submodel in semiparametric estimation.

Our work is also tied to the literature on heterogeneous treatment effect estimation in causal inference.
Different from our parameter of interest, \citet{Chernozhukov2018Simultaneous}, building on the debiased double machine learning framework (\citealp{chernozhukov2017double}), propose to estimate the average treatment effect conditional on a small subset of the potential confounders.
 \citet{KunzelCATE} and \citet{abrevaya2015estimating} propose meta-learning frameworks that estimates the average treatment conditional on all possible confounders. 
Unlike our approach, which efficiently evaluates the treatment effects in pre-specified subgroups, \citet{imai2013estimating} formulate the problem on heterogeneous treatment effect identification from a variable selection perspective. In this thread on heterogeneity identification, \citet{su2009subgroup} propose a recursive partitioning tree approach to identify treatment heterogeneity across subgroups. \cite{vanderweele2019selecting} provide a nice overview of subgroup selection problems encountered in practice.

\section{Causal Framework and Identification}\label{sec:framework}

Let $\{O_i\}_{i=1}^n =\{ (Y_i, T_i, X_i)\}_{i=1}^n $ be an independent and identically distributed (i.i.d.) random sample of the observed binary response variable $Y$, the treatment indicator variable $T$, and potential confounders $X\in\mathds{R}^p$. In accordance with the Neyman-Rubin causal model (\citealp{neyman1923application, rubin1974estimating}), we define the potential outcome $Y(T)$ as the outcome we would have observed under the treatment assignment $T$. The observed outcome is thus the potential outcome variable corresponding to the received treatment, i.e., $Y = T Y(1) + (1-T) Y(0)$. This framework allows us to characterize the multi-subgroup disease risk under different treatment arms as:
$\bm{\alpha}_{t} = (\alpha_{t,1},\dots \alpha_{t,d})^\intercal$, $\alpha_{t,j} =  {P}\big( Y(t)=1|X\in\mathcal{A}_j\big)$,  $t\in\{0,1\},\quad j=1,\dots, d$,
where $\mathcal{A}_j$ denotes a pre-specified subgroup $j$. Here, we allow different subgroup to overlaps, and we assume that the variables used to define the subgroups of interest are based on $X$. When comparing disease risks between two treatment arms, our framework allows practitioners to estimate three popular causal effect measures: relative risk, odds ratio, and absolute risk difference, across different subgroups, defined as 
$\bm{\alpha}_{\textbf{RR}} = (\alpha_{\text{RR},1},\dots \alpha_{\text{RR},d})^\intercal$,   $\alpha_{\text{RR},j} = \alpha_{1,j}/\alpha_{0,j}, 
\bm{\alpha}_{\textbf{OR}} = (\alpha_{\text{OR},1},\dots \alpha_{\text{OR},d})^\intercal$,   $\alpha_{\text{OR},j} = \big(\alpha_{1,j}/(1-\alpha_{1,j})\big)/\big(\alpha_{0,j}/(1-\alpha_{0,j} )\big)$, and $\bm{\alpha_{\textbf{ARD}}} = (\alpha_{1,1}-\alpha_{0,1},\dots, \alpha_{1,d}-\alpha_{0,d})$ (Section \ref{sec:Estimate-RR-OR}). 

The three causal quantities described above are not observable because the potential outcomes are subject to missingness, meaning that for each individual we observe either the potential outcome under the control, $Y(0)$, or the potential outcome under the treatment, $Y(1)$, but never both. Following the mainstream literature in causal inference, we impose the unconfoundedness, positivity, and stable unit treatment value assumptions (SUTVA) below to identify our causal parameters of interest:
\begin{assumption}[Unconfoundedness]\label{assump:unconfound}
	Conditional on $X$, the treatment assignment is as good as random, that is $T \perp Y(1), Y(0) | X$. 
\end{assumption}

\begin{assumption}[Positivity]
\label{assumption:overlap} 
 For any $x\in X$, $t\in \{0,1\}$, there exists a constant $c\in(0,1)$ such that $ c <  {P}(T = t|X = x, X\in\mathcal{A}_j) < 1-c$ and $c < {P}(\mathcal{A}_j)<1-c$, for $j= 1, \dots, d$. 
\end{assumption}

\begin{assumption}[SUTVA]\label{assumption:sutva}
If unit $i$ receives treatment $T_i$, the observed outcome $Y_i$ equals the potential outcome $Y_i(T_i)$, meaning that the potential outcome for unit $i$ under treatment $T_i$ is unrelated to the treatment received by other units. 
\end{assumption}

Under Assumption \ref{assump:unconfound}-\ref{assumption:sutva}, we are able to identify $\alpha_{t,j}$ as
$\alpha_{t,j} =  P\big( Y(1)=1|X\in\mathcal{A}_j\big) = E_{X}\big[ P(Y=1|T= t, X\in\mathcal{A}_j) \big].$
Here, by ``identify" we mean that under Assumption \ref{assump:unconfound}, the causal effect involving unobserved potential outcomes can be first written as a function of observed data. Then, within an i.i.d. sample $\{ (Y_i, T_i, X_i)\}_{i=1}^n$, under Assumption \ref{assumption:overlap} and \ref{assumption:sutva}, the causal parameter can be estimated (or point identified) at a regular parametric root-$n$ rate \citep{khan2010irregular}.

\noindent\textit{Notation}
We use $P$ to denote the probability operator and $E$ to denote the expectation operator. We use capitalized letters to denote random variables, e.g. $T$, and lower case letters to denote the realizations of random variables, e.g. $t$. For $t\in\{0,1\}$, we denote $ p_t( X)={P}(Y=1|T=t, X)$ as the conditional probability of $Y = 1$ given $T = t$ and $X$. $e_t(X) = {P}(T=t|X)$ denotes the conditional probability of $T=t$ given $X$. Lastly, we define $\text{expit}(x) = \frac{1}{1+e^{-x}}$ and $\text{logit}(x)= \text{log}(\frac{x}{1-x})$.

\section{Multiple Subgroup Targeted Learning}\label{sec:method}

In this section, to simplify presentation, we first introduce our method on estimating the conditional average risk $\bm{\alpha}_t$ for group $t\in\{0,1\}$ and defer the estimation for other causal parameters to Section \ref{sec:Estimate-RR-OR} and Web Appendix \ref{supp:simul}. We shall review the classical one-step targeted maximum likelihood estimator (TMLE) (\citealp{van2016one}) in a single subgroup case, followed by discussing its limitations when naively generalizing it to the multi-subgroup case. We then introduce the one-step TMLE that directly targets the multi-subgroup treatment effects simultaneously. 

\subsection{Limitation of the classical one-step TMLE}\label{sec:limitation-of-one-step-TMLE}

To estimate $\bm{\alpha}_t$, a natural choice is to apply the one-step TMLE in each subgroup separately. For a subgroup $j$, one-step TMLE starts with some initial estimates of $p_t(X)$ and $e_t(X)$ using the observations in the subgroup $\mathcal{A}_j$, denoted as  $\hat{p}_{tj}^{\mathtt{Init}}(X)$ and $\hat{e}_{tj}(X)$. These initial estimates can be obtained from any state-of-art machine learning methods--such as random forest, gradient boosting (\citealp{breiman2001random}), or Highly Adaptive Lasso (HAL) (\citealp{benkeser2016highly})--as long as they are not too far away from the target estimands (see Assumption \ref{assumption:regularity-nuisance} in Section \ref{sec:conditions-notations} for rigorous specifications). Within a random sample, because $\hat{p}_{tj}^{\text{Init}}(X)$ and $\hat{e}_{tj}(X)$ may substantially deviate from the truth, the targeted learning approach identifies a correction term, $\hat{\varepsilon}\cdot \hat{S}_{tj}(X)$, that pushes the initial estimates to ``concentrate/target" on the estimand:
$
 \hat{p}_{tj}(X_i) = \text{expit}\Big(\text{logit}\big(\hat{p}^{\text{Init}}_{tj}(X_i)\big) + \hat{\varepsilon}\cdot \hat{S}_{tj}(X_i)\Big), \  \hat{S}_{tj}(X_i) = 
   \frac{\mathds{1}(X_i\in\mathcal{A}_j)}{\hat{P}(\mathcal{A}_j)}\frac{\mathds{1}(T_i=t)}{\hat{e}_{tj}(X_i)}.
$
Here, $\hat{P}(\mathcal{A}_j) = \frac{\sum_{i=1}^n\mathds{1}(X_i\in\mathcal{A}_j)}{n}$, $\hat{\varepsilon}$ captures the magnitude of the correction $\hat{S}_{tj}(X_i)$ (so called ``clever covariate" in \citet{laantmle2006}), and it is the estimated coefficient of $\hat{S}_{tj}(X_i)$ in the logistic regression: 
\begin{align}\label{eq:simple-epsilon-update}
Y_i \sim \text{logit}\big(\hat{p}^{\text{Init}}_{tj}(X_i)\big) + \varepsilon \hat{S}_{tj}(X_i),\quad i \in \mathcal{A}_{tj},
\end{align}
that regresses $Y_i$ on $\text{logit}\big(\hat{p}^{\text{Init}}_{tj}(X_i)\big)$ and $\hat{S}_{tj}(X_i)$ with a fixed coefficient 1 for $\text{logit}\big(\hat{p}^{\text{Init}}_{tj}(X_i)\big)$. Here $\mathcal{A}_{tj} = \mathcal{A}_j \cap \{i: T_i=t\}$ contains the subjects with $T_i =t$ in the subgroup $\mathcal{A}_j$.
After this one-step correction, the final estimate $\hat{\alpha}^{\text{one-step}}_{t,j}$ takes the empirical average of $\hat{p}_{tj}(X_i)$:
$\hat{\alpha}^{\text{one-step}}_{t,j} = \frac{1}{n_{tj}}\sum_{i=1}^n \hat{p}_{tj}(X_i),$ where $n_{tj}$ is the cardinality of the set $\mathcal{A}_{tj}$. 

The regression problem defined in Eq~\eqref{eq:simple-epsilon-update} is the essence of the one-step TMLE. Such a regression problem adaptively learns the difference between $\hat{p}^{\text{Init}}_{tj}(\cdot)$ and $p_{tj}(\cdot)$ from the data, aiming to find an $\hat{\varepsilon}$ that locally improves the empirical fit of the initial estimator $\hat{p}^{\text{Init}}_{tj}(\cdot)$. We choose $\hat{\varepsilon}$ in a data adaptive fashion because when the initial estimate of the conditional probability is identical to the true conditional probability, we hope to set $\hat{\varepsilon}=0$.
It is only when the initial estimate $\hat{p}^{\text{Init}}_{tj}(\cdot)$ drifts away from $p_{tj}(\cdot)$, $\hat{\varepsilon}$ accounts for their difference and updates $\hat{p}^{\text{Init}}_{tj}(\cdot)$ accordingly. Furthermore, because our goal is to estimate $\alpha_{t,j}$, the clever covariate $S_{tj}(X_i)$ specifies the updating direction of the initial estimator that yields a maximal change (or maximal information gain) in the target parameter. Benefiting from such an update, the final estimator $\hat{\alpha}^{\text{one-step}}_{t,j}$ attains the semiparametric efficiency bound under the regularity conditions in Section \ref{sec:conditions-notations}.
\begin{wraptable}{r}{6.5cm}
    \centering
     \caption{Computational time (in seconds) of the conventional TMLE and the proposed method with sample size $n=228,466$ on a Lenovo NeXtScale nx360m5 node (24 cores per node) equipped with Intel Xeon Haswell processor. The core frequency is 2.3 Ghz and supports 16 floating-point operations per clock period.}
       \label{table:computation-time}
    \begin{tabular}{ccc}
         &  Classical one-step TMLE & iTMLE  \\\hline
         & 1441.36 &  924.51
    \end{tabular}
\end{wraptable}
In addition, because the one-step TMLE applies an ``expit" transformation on the sum of $\text{logit}\big(\hat{p}^{\text{Init}}_{tj}(X_i)\big)$ and the inverse propensity score, the estimated conditional risk $\hat{\alpha}^{\text{one-step}}_{t,j}$ never falls out of the range between 0 and 1 regardless of how small $\hat{e}_{tj}(\cdot)$ is (see Section \ref{Section:Simulation-comparison-with-others} for numerical verification).

Nevertheless, naively carrying out the above procedure one subgroup at a time can be computationally inefficient in the presence of many subgroups.  In a simple comparison provided in Table \ref{table:computation-time}, our proposed estimator directly targeting the multi-subgroup parameter $\bm{\alpha}_{t}$ as a whole improves the computational speed by about 35\% compared to this one-group-at-a-time approach, 
when the initial estimator $\hat{p}^{\text{Init}}_{tj}(\cdot)$ and the estimated propensity scores $\hat{e}_{tj}(\cdot)$ are obtained via GLMs.

\subsection{One-step TMLE targeting multiple subgroups}\label{subsec:proposal}

\subsubsection{Procedure overview}
To avoid the discussed potential problems of the conventional one-step TMLE, we amend the one-step TMLE estimator so that it directly targets $\bm{\alpha}_t$. A natural idea is to replace the univariate clever covariate with a multi-dimensional vector of clever covariates 
$\big( \hat{S}_{t1}(X_i), \ldots, \hat{S}_{td}(X_i) \big)^\intercal$ in the logistic regression
\begin{align}\label{eq:tmle-constraint}
Y_i \sim \text{logit}\big(\hat{p}^{\text{Init}}_t(X_i)\big) +  \sum_{j=1}^d{\varepsilon}_{t,j}\cdot \hat{S}_{tj}(X_i), \quad i\in \{i: T_i = t\},
\end{align}
where $\hat{S}_{tj}(X_i) = \frac{\mathds{1}(X_i\in\mathcal{A}_j)}{\hat{P}(\mathcal{A}_j)}\frac{\mathds{1}(T_i=t)}{\hat{e}_{t}(X_i)}$. Note that here we generate the initial estimates $\hat{p}_t^{\text{Init}}(X_i)$ and $\hat{e}_{t}(X_i)$ with the entire available sample. We then construct the estimator for $\bm{\alpha}_t$ with
\begin{align}\label{eq:one-step-alpha}
\bm{\hat{\alpha}}^{\text{one-step}}_t = \big(\frac{1}{n_{t1}}\sum_{i=1}^n \hat{p}_{t1}(X_i), \ldots, \frac{1}{n_{td}}\sum_{i=1}^n \hat{p}_{td}(X_i) \big)^\intercal,
\end{align}
where $\hat{p}_{tj}(X_i) = \text{expit}\Big(\text{logit}\big(\hat{p}^{\text{Init}}_{t}(X_i)\big) + \hat{\varepsilon}_{t,j}\cdot \hat{S}_{t,j}(X_i)\Big)$.

In the presence of multiple subgroups with large $d$, we may observe small $\hat{P}(A_j)$ or $\hat{e}_t(X_i)$ within a random sample. In this situation, given that  $\hat{P}(A_j)$ and $\hat{e}_t(X_i)$ enter the regression problem in Eq \eqref{eq:tmle-constraint} as denominators, the above procedure can potentially produce numerically unstable estimates, which may inflate the variance of $\bm{\hat{\alpha}}^{\text{one-step}}_{t}$. We hope to further robustify the above procedure by considering a simple variation, where we shall also demonstrate that the algorithm proposed below is a reparametrized dual problem of the above (primal) problem defined in Eq \eqref{eq:tmle-constraint}.

Our proposed procedure operates as follows, for each iteration $k$, 
\begin{align}\label{eq:TMLE-variant}
& Y_i \sim \text{logit}\big(\hat{p}^{(k-1)}_t(X_i)\big) + \gamma \tilde{S}^{(k-1)}_t(X_i) , \\
& \nonumber\hat{p}^{(k)}_t(X_i) = \text{expit}\Big(\text{logit}\big(\hat{p}^{(k-1)}_t(X_i)\big) + \hat{\gamma}^{(k)}\cdot \tilde{S}^{(k-1)}_t(X_i)\Big), \quad i\in\{i:T_i=t\}, \quad k=1, \ldots, K, 
\end{align}
where $\hat{\gamma}^{(k)}$ is the estimated regression coefficient obtained in the logistic regression \eqref{eq:TMLE-variant}.  $\hat{p}_t^{(1)}(X_i)$ denotes the initial estimate. $\hat{p}_t^{(k-1)}(X_i)$ denotes the estimate from the previous iteration, and $\tilde{S}^{(k-1)}_t(X_i) $ is the customized ``clever covariate" that directly targets $\bm{\alpha}_t$:
\begin{align}\label{eq:clever-covariate}
\tilde{S}^{(k-1)}_t(X_i) = 
 \frac{ \sum_{j=1}^d \frac{\mathds{1}(X_i\in\mathcal{A}_j)}{\hat{P}(\mathcal{A}_j)} \frac{\mathds{1}( T_i=t)}{\hat{e}_t(X_i)}\cdot\Big(\sum_{l=1}^n \hat{\phi}^{(k-1)}_j(Y_l, T_l,X_l)\Big)}{\sqrt{ \sum_{j=1}^d \left(  \sum_{l=1}^n \hat{\phi}^{(k-1)}_j(Y_l,T_l,X_l) \right)^2}},
\end{align}
where
$\hat{\phi}^{(k-1)}_j(Y_i, T_i,X_i) = \frac{\mathds{1}(X_i\in\mathcal{A}_j)}{\hat{P}(\mathcal{A}_j)}\frac{\mathds{1}( T_i=t)}{\hat{e}_t(X_i)}(Y_i - \hat{p}_t^{(k-1)}(X_i)).$
The intuition of $\tilde{S}^{(k-1)}_t(X_i)$ shall be explained in the next section. When the maximum number of iterations $K$ is reached or when $\hat{\gamma}$ is sufficiently close to 0, we take the final estimate $\hat{p}_t(X_i) = \hat{p}^{(K)}_t(X_i)$ and estimate $\bm{\alpha}_t$ again with: 
\begin{align}\label{eq:alpha-proposed}
    \hat{\bm{\alpha}}_t = \big(\frac{\sum_{i\in \mathcal{A}_1} \hat{p}_t(X_i) }{n_{t1}}, \ldots, \frac{\sum_{i\in \mathcal{A}_d} \hat{p}_t(X_i) }{n_{td}}\big)^\intercal,
\end{align}
where $n_{tj} =\sum_{i=1}^n \mathds{1}( T_i=t) \mathds{1}(X_i\in\mathcal{A}_j) $ denotes the subgroup $j$'s sample size in the arm $t$.

We refer to the estimator in Eq \eqref{eq:alpha-proposed}, which is obtained from Eq \eqref{eq:TMLE-variant},
as the iterative version of the one-step TMLE (iTMLE) targeting multiple subgroups of interest.

\subsubsection{Intuitive explanation of our proposal}

Note that although the proposed estimators in Eq \eqref{eq:one-step-alpha} and Eq \eqref{eq:alpha-proposed} are asymptotically equivalent as $n\rightarrow \infty$, we provide some heuristic explanations of the benefits of adopting our procedure defined in Eq \eqref{eq:TMLE-variant} compared to the procedure defined in Eq \eqref{eq:tmle-constraint} in finite samples.

First, given that the performance of the one-step TMLE defined by Eq \eqref{eq:tmle-constraint} depends on the initial estimator $\hat{p}_t^{\text{Init}}(X_i)$, our revised procedure in Eq \eqref{eq:TMLE-variant} works with an improved initial estimator in each iteration. Concretely, in Eq \eqref{eq:TMLE-variant}, the initial estimator entering each iteration is constantly being updated, leading to increased estimation efficiency and reduced estimation bias compared to the procedure defined in Eq \eqref{eq:tmle-constraint}. Such improvements can be rather prominent in finite samples (See Web Appendix \ref{appendix:sim-comparison} for simulation comparisons). 

Second, the form of the clever covariate $\tilde{S}_t(X_i)$ in Eq \eqref{eq:TMLE-variant} may have the added benefit of being robust to the presence of small estimated propensity scores, because the estimated propensity scores only enter the estimation process after being self-normalized in $\tilde{S}_t(X_i)$.
Small propensity scores are often encountered in datasets with unbalanced covariate distribution across the treatment and control groups. Such an imbalance can lead to conventional estimators having substantial biases and large variances \citep{crump2009dealing, petersen2012diagnosing}. Many numerical studies have found that similar self-normalization of propensity scores provides much more stable estimates of the treatment effects in finite samples (\citealp{hajek1971comment}). While the original formulation of the primal problem in Eq \eqref{eq:tmle-constraint} involves a sum over $d$ inverse propensity score weighted clever covariates, its performance can be sensitive to the presence of small propensity scores in finite samples. Even though the estimator obtained by Eq \eqref{eq:TMLE-variant} and the estimator obtained by Eq \eqref{eq:tmle-constraint} are asymptotically equivalent, the estimator obtained by Eq \eqref{eq:TMLE-variant} may have finite sample improvements when the estimated propensity scores are small.

Third, the estimator obtained from Eq \eqref{eq:TMLE-variant} not only remains semi-parametric efficient and ``doubly robust," but also solves the direct sample analogue of the efficient influence function. To see why it is semiparametric efficient, we set the derivative of the objective function of the logistic regression in \eqref{eq:tmle-constraint} with respect to $\varepsilon$ to zero, which reduces to (see Web Appendix \ref{appendix:multi-score-proof}
for detailed derivations)
\begin{align}\label{eq:estimating-equations}
   & \sum_{j=1}^d \left(  \frac{1}{n}\sum_{i=1}^{n} \frac{\mathds{1}(X_i\in\mathcal{A}_j)}{\hat{P}(\mathcal{A}_j)}\frac{T_i}{\hat{e}_t(X_i)}(Y_i - \hat{p}_t(X_i)) \right)^2 = 0.
\end{align}
 This indicates that our estimator $\hat{\bm{\alpha}}_t =(\hat{\alpha}_{t,1},\ldots,  \hat{\alpha}_{t,d})^\intercal$ solves the direct sample analogue of the efficient influence function:
$\frac{1}{n}\sum_{i=1}^n \frac{\mathds{1}(X_i\in\mathcal{A}_j)}{\hat{P}(\mathcal{A}_j)}\left\{ \frac{T_i}{\hat{e}_t(X_i)}(Y_i - \hat{p}_t(X_i))  + \hat{p}_t(X_i) \right\} - \hat{\alpha}_{t,j} =0 , \ j=1,\ldots, d.$
Therefore, it attains the semiparametric efficiency bound (\citealp{bickel1993onestep}) under appropriate conditions imposed on the nuisance parameter estimators (Theorem \ref{thm:multi-efficiency}). Regarding the ``doubly robustness," for any model-based estimators $\hat{e}_t(\cdot)$ and $\hat{p}_t(\cdot)$, our estimator combines regression imputation and inverse propensity score weighting, and remains consistent if either the model $e_t(\cdot)$ or $p_t(\cdot)$ is misspecified (see Section \ref{Section:Simulation-comparison-with-others} for simulation results). In Section \ref{subsec:heuristics}, we shall provide further heuristic explanations of the targeted maximum likelihood estimator from a semiparametric inference point of view.

\subsection{Extension to relative risk, odds ratio, and  absolute risk difference estimations}\label{sec:Estimate-RR-OR}
Given that $\bm{\alpha_1}$ and  $\bm{\alpha_0}$ are the building blocks of the multi-subgroup relative risk and odds ratio, estimation for these two parameters of interest largely follows our proposal in Section \ref{subsec:proposal}. The iterative version of the one-step TMLE needs a slight modification in that at each iteration $k$, we adopt the following logistic regression problem: $
 Y_i \sim \text{logit}\big(\hat{p}^{(k-1)}(T_i,X_i)\big) + \gamma_1 \tilde{S}^{(k-1)}_1(X_i) + \gamma_0 \tilde{S}^{(k-1)}_0(X_i) , \ k=1, \ldots, K, $
and perform the updating as $
 \hat{p}^{(k)}(T_i,X_i) = \text{expit}\Big(\text{logit}\big(\hat{p}^{(k-1)}(T_i,X_i)\big) + \hat{\gamma}_1^{(k)}\cdot \tilde{S}^{(k-1)}_1(X_i)+\hat{\gamma}_0^{(k)}\cdot \tilde{S}^{(k-1)}_0(X_i\Big).$
Then we estimate $\bm{{\alpha}_{\textbf{RR}}} $, $\bm{{\alpha}_{\textbf{OR}}} $, and $\bm{{\alpha}_{\textbf{ARD}}}$  with $
\bm{\hat{\alpha}_{\textbf{RR}}} = \Big(\frac{\hat{\alpha}_{1,1}}{\hat{\alpha}_{0,1}},\dots, \frac{\hat{\alpha}_{1,d}}{\hat{\alpha}_{0,d}}\Big),\ \bm{\hat{\alpha}_{\textbf{OR}}} = \Big(\frac{\hat{\alpha}_{1,1}}{1-\hat{\alpha}_{1,1}}\Big/\frac{\hat{\alpha}_{0,1}}{1-\hat{\alpha}_{0,1}},\dots, \frac{\hat{\alpha}_{1,d}}{1-\hat{\alpha}_{1,d}}\Big/\frac{\hat{\alpha}_{0,d}}{1-\hat{\alpha}_{0,d}}\Big),$ and $\bm{\hat{\alpha}_{\textbf{ARD}}} = \Big(\hat{\alpha}_{1,1}-\hat{\alpha}_{0,1},\dots, \hat{\alpha}_{1,d}-\hat{\alpha}_{0,d}\Big)$

As for constructing simultaneous confidence intervals, we apply the Delta method on $(\bm{\alpha_1},\bm{\alpha_0})$ to estimate the sample covariance matrices of the relative risk and the odds ratio estimators following a recipe similar to Section \ref{subsec:confidence-interval}. To avoid redundancy, we leave the detailed descriptions to Web Appendix \ref{supp:simul}.

\section{Theoretical Investigations}\label{sec:theory}

\subsection{Properties of the proposed estimator}
In this section, we introduce the main theoretical results and some necessary notation. We defer additional notation and regularity conditions to Section \ref{sec:conditions-notations}.
Recall that  $\{O_i\}_{i=1}^n: =\{ (Y_i, T_i, X_i)\}_{i=1}^n $ is an i.i.d. random sample defined on the space $\mathcal{O}$ with respect to a probability measure $P$. Denote $o=(y,t,x)$ as a realized data point, $o \in\mathcal{O}$. 

\begin{thm}\label{thm:multi-efficiency}
Under Assumptions \ref{assump:unconfound}-\ref{assumption:regularity-nuisance}, we define the vector of the efficient influence function $\bm{\varphi}_t = (\varphi_{t,1},\ldots,\varphi_{t,d})^{\intercal}$, where $\varphi_{t,j}$ is the efficient influence function (as given in Eq (8)) measured at a realized data point $o = (y,t,x) $ for the subgroup $j$. The proposed conditional risk estimator $ \hat{\bm{\alpha}}_{t} = ( \alpha_{t,1}, \ldots, \alpha_{t,d})^{\intercal}\in\mathbb{R}^d $, after scaling by $\sqrt{n}$, converges to a multivariate Gaussian random variable with mean 0 and covariance matrix $P[\bm{\varphi}_t\bm{\varphi}_t^{\intercal}]$ when $n\rightarrow \infty$, that is $\sqrt{n}\big(  \hat{\bm{\alpha}}_{t} -  {\bm{\alpha}}_{t} \big) \leadsto \mathcal{N}\Big(0, P[\bm{\varphi}_t\bm{\varphi}_t^{\intercal}] \Big)$. (See the precise definition of $\varphi_{t,j}$ in Section \ref{sec:conditions-notations}).
\end{thm}

Theorem \ref{thm:multi-efficiency} says that our conditional risk estimator converges in distribution to a multivariate Gaussian distribution. For any subgroups under consideration, the variance of our conditional risk estimator attains the semiparametric efficiency bound. Theorem \ref{thm:multi-efficiency} also justifies the validity of the simultaneous confidence interval provided in Eq to be presented \eqref{eq:multi-subgroup-CI} in Section \ref{subsec:confidence-interval}. 

Derivations of the efficient influence functions for relative risk, odds ratio and absolute risk difference estimators are provided in Web Appendix \ref{appendix:rr-or-semiparametric}. We summarize the large sample properties of $\bm{{\alpha}_{\textbf{RR}}} $, $\bm{{\alpha}_{\textbf{OR}}} $, and $\bm{{\alpha}_{\textbf{ARD}}} $ in the following Proposition \ref{thm:RR-OR}, which demonstrates that the variance of the proposed causal effect estimators attains the semiparametric efficiency bound. The proof of the proposition below can be found in Web Appendix \ref{appendix:proof-theorem-1}.

\begin{prop}\label{thm:RR-OR}
Under Assumptions \ref{assump:unconfound} - \ref{assumption:regularity-nuisance},
define the vector of the efficient influence function $\bm{\varphi}_{\text{RR}} =(\varphi_{\text{RR},1},\ldots,\varphi_{\text{RR},d})^{\intercal}$, the vector of the efficient influence function $\bm{\varphi}_{\text{OR}} = (\varphi_{\text{OR},1},\ldots,\varphi_{\text{OR},d})^{\intercal}$, and the vector of the efficient influence function $\bm{\varphi}_{\text{ARD}} = (\varphi_{\text{ARD},1},\ldots,\varphi_{\text{ARD},d})^{\intercal}$, where $\varphi_{\text{RR},j}$, $\varphi_{\text{OR},j}$,  and $\varphi_{\text{ARD},j}$ are the efficient influence functions (as given in Eq (9)-(11)) measured at a realized data point $o = (y,t,x) $.
The proposed causal effect estimators satisfy that as $n\rightarrow \infty$, $ \sqrt{n}\big(  \hat{\bm{\alpha}}_{\text{RR}} -  {\bm{\alpha}}_{\text{RR}} \big) \leadsto \mathcal{N}\Big(0, P[\bm{\varphi}_{\text{RR}}\bm{\varphi}_{\text{RR}}^{\intercal}] \Big)$,  $ \sqrt{n}\big(  \hat{\bm{\alpha}}_{\text{OR}} -  {\bm{\alpha}}_{\text{OR}} \big) \leadsto \mathcal{N}\Big(0, P[\bm{\varphi}_{\text{OR}}\bm{\varphi}_{\text{OR}}^{\intercal}] \Big)$ and $ \sqrt{n}\big(  \hat{\bm{\alpha}}_{\text{ARD}} -  {\bm{\alpha}}_{\text{ARD}} \big) \leadsto \mathcal{N}\Big(0, P[\bm{\varphi}_{\text{ARD}}\bm{\varphi}_{\text{ARD}}^{\intercal}] \Big)$ (See the precise definitions of $\varphi_{\text{RR},j}$, $\varphi_{\text{OR},j}$, and $\varphi_{\text{ARD},j}$ in Section \ref{sec:conditions-notations}).

\end{prop}


\subsection{ Regularity conditions}\label{sec:conditions-notations}

In this section, we introduce additional notation and assumptions adopted in the theoretical results. 
Recall that $\{O_i\}_{i=1}^n: =\{ (Y_i, T_i, X_i)\}_{i=1}^n $ are i.i.d. random variables defined on the space $\mathcal{O}$ with respect to a probability measure $P$. 
If $\mathcal{F}$ is a collection of real-valued functions defined on $\mathcal{O}$, we assume that $Pf = \int f \mbox{d} P$ exists for each $f\in\mathcal{F}$. Note that such a notation can be more helpful as it allows us to conveniently work with random functions. We use $E_X[f(X)]$ to denote the expectation taken with respect to the random variable $X$ when it is more convenient to simplify notation. Given the probability measure $P$, our target parameter $\bm{\alpha}_t$ can also be written as a statistical function of $P$, denoted as $\bm{\alpha}_t(P)$. Let $\mathcal{H}$ be a convex set of functions such that the true nuisance parameter 
$\bm{\eta}_0\triangleq({e}(x), {p}_1(x), {p}_0(x), {{P}}(\mathcal{A}_1), \ldots, {{P}}(\mathcal{A}_d)) \in\mathcal{H}.$
Let $\mathcal{H}_n \subset\mathcal{H} $ denote the nuisance estimator realization set, i.e., the estimator of the nuisance parameters satisfy $\hat{\bm{\eta}} =(\hat{e}_t(x), \hat{p}_1(x), \hat{p}_0(x), \hat{P}(\mathcal{A}_1), \ldots, \hat{P}(\mathcal{A}_d))\in \mathcal{H}_n$.

Let $c$, $q$, and $C$ be fixed strictly positive constants, where $q>2$. Let $(\xi_n)_{n=1}^{\infty}$ and  $(\Delta_n)_{n=1}^{\infty}$ be sequences of positive constants approaching $0$. Denote the $l_q$-norm with respect to a probability measure $P$ as $||\cdot||_{P,q}$, e.g. $||f(X)||_{P,q} := (\int |f(x)|^q dP(x))^{1/q}$. For $o\in\mathcal{O}$, we define $\bm{\varphi}_t(o;\bm{\alpha}_t, \bm{\eta}_0) \triangleq \big(\varphi_{t,1}, \ldots, \varphi_{t,d}  \big)^{\intercal}$ as the vector of the efficient influence function for estimating $\bm{\alpha}_t$, $\bm{\varphi}_\text{RR}(o;\bm{\alpha}_\text{RR}, \bm{\eta}_0) \triangleq \big(\varphi_{\text{RR},1}, \ldots, \varphi_{\text{RR},d}  \big)^{\intercal}$ as the vector of the efficient influence function for estimating $\bm{\alpha}_\text{RR}$,  $\bm{\varphi}_{\text{OR}}(o;\bm{\alpha}_\text{OR}, \bm{\eta}_0) \triangleq \big(\varphi_{\text{OR},1}, \ldots, \varphi_{\text{OR},d}  \big)^{\intercal}$ as the vector of the efficient influence function for estimating $\bm{\alpha}_\text{OR}$, 
and $\bm{\varphi}_{\text{ARD}}(o;\bm{\alpha}_\text{ARD}, \bm{\eta}_0) \triangleq \big(\varphi_{\text{ARD},1}, \ldots, \varphi_{\text{ARD},d}  \big)^{\intercal}$ as the vector of the efficient influence function for estimating $\bm{\alpha}_\text{ARD}$, 
where for $j=1, \ldots, d$,
\begin{align}
\varphi_{t,j} \triangleq \varphi_{t,j}(o;\bm{\alpha}_t, \bm{\eta}_0) & = \frac{ \mathds{1}(x \in \mathcal{A}_j) }{P(\mathcal{A}_j)}\Big[ \Big(y - p_t(x)\Big)\frac{\mathds{1}(T= t) }{e_t(x)} + p_t(x) - \alpha_{t,j}\Big], \\
    \varphi_{\text{RR},j} \triangleq \varphi_{\text{RR},j}(o;\bm{\alpha}_{\text{RR}}, \bm{\eta}_0)    & \ = \frac{ \mathds{1}(x\in \mathcal{A}_j) }{P(\mathcal{A}_j)}\Big[ \frac{1}{\alpha_{0,j}}\Big(\big(y - p_1(x)\big)\frac{t}{e_1(x)}+ p_1(x)-\alpha_{1,j}\Big)\\
     &\quad\quad\quad\quad\quad\quad\quad +\frac{\alpha_{1,j}}{\alpha_{0,j}^2}\Big(\frac{1-t}{e_0(x)}\big(y-p_0(x)\big)+p_0(x)-\alpha_{0,j}\Big)\Big],\nonumber\\
    \varphi_{\text{OR},j} \triangleq \varphi_{\text{OR},j}(o;\bm{\alpha}_{\text{OR}}, \bm{\eta}_0) & \ =  \frac{ \mathds{1}(x \in \mathcal{A}_j) }{P(\mathcal{A}_j)}\Big[ \frac{1-\alpha_{0,j}}{\alpha_{0,j}(1-\alpha_{1,j})^2}\Big(\big(y - p_1(x)\big)\frac{t}{e_1(x)}+ p_1(x)-\alpha_{1,j}\Big)\\
&\quad\quad\quad\quad\quad\quad\quad - \frac{\alpha_{1,j}}{\alpha_{0,j}^2(1-\alpha_{1,j})}\Big(\frac{1-t}{e_0(x)}\big(y-p_0(x)\big)+p_0(x)-\alpha_{0,j}\Big)\Big]. \nonumber \\
\varphi_{\text{ARD},j} \triangleq \varphi_{\text{ARD},j}(o;\bm{\alpha}_{\text{ARD}}, \bm{\eta}_0) & \ =  \frac{ \mathds{1}(x \in \mathcal{A}_j) }{P(\mathcal{A}_j)}\Big[\Big(\big(y - p_1(x)\big)\frac{t}{e_1(x)}+ p_1(x)-\alpha_{1,j}\Big)\\
&\quad\quad\quad\quad\quad\quad\quad - \Big(\frac{1-t}{e_0(x)}\big(y-p_0(x)\big)+p_0(x)-\alpha_{0,j}\Big)\Big].\nonumber
\end{align}

\begin{assumption}\label{assumption:regularity-score}
The function class $\{ \bm{\varphi}(o;\bm{\alpha}_t, \bm{\eta}), \bm{\eta}\in \mathcal{H}\}$ is a Donsker class.
\end{assumption}

\begin{assumption}\label{assumption:regularity-nuisance}
The nuisance parameter estimator $\hat{\bm{\eta}}$  satisfies that $\sup_{\bm{\eta}\in\mathcal{H}_n }|| \bm{\eta} - \bm{\eta}_0||_2 = o_P(1)$ and $||\hat{e}(X) - e(X)||_{P,2} \times  ||\hat{p}_t(X) - p_t(X)||_{P,2}\leq \xi_n n^{-1/2}$ holds with probability 1 when $n$ tends to infinity. 
\end{assumption}

Assumption \ref{assumption:overlap} assumes that all units have non-zero probabilities of being assigned to the treatment or the control arm. Such an assumption has been frequently considered in the causal inference literature. In addition, because we estimate the treatment effects across multiple subgroups, we require each subgroup to satisfy the positivity condition as well. 
Assumption \ref{assumption:regularity-score} assumes the Donsker class condition for the class of efficient influence functions. This Donsker class condition can be weakened by conducting cross-fitting (see Web Appendix \ref{appendix:cv-TMLE} for implementation details) and at the expense of more complicated proofs (see \citet{zheng2010asymptotic}, for example). Additionally, \citet{benkeser2016highly} propose the highly adaptive lasso (HAL) estimator which guarantees $\sqrt{n}$-rate of convergence in the initial estimation step. Assumption \ref{assumption:regularity-nuisance} imposes regularity conditions on the nuisance parameter estimator.  The second part in Assumption \ref{assumption:regularity-nuisance}  bounds the product of errors of the nuisance parameter estimators $\hat{p}_t(X)$ and $\hat{e}(X)$.

\subsection{Duality theory}
In this section, we provide an alternative explanation of the iterative one-step TMLE from a duality theory perspective. Recall that in the discussed method, we replace the univariate clever covariate $\hat{S}_{tj}(X_i)$ with a multi-dimensional vector of clever covariate in the logistic regression in Eq \eqref{eq:tmle-constraint}.
Because we hope to limit the sum of the squared influences of the clever covariates on updating $p_t(\cdot)$ when $d$ is a large number, we impose a constraint that $\Vert\bm{\varepsilon}\Vert_2 \leq \delta$, where $\bm{\varepsilon} = (\varepsilon_1, \dots, \varepsilon_d)^\intercal$. Thus, solving for $\bm{\varepsilon}$ can be equivalently reformulated as iteratively solving the constraint optimization problem below 
\begin{align}\label{eq:primal-problem}
 \hat{\bm{\varepsilon}}^{(k)}  = \underset{ ||\bm{\varepsilon}||\leq \delta }{\arg\min} 
-\frac{1}{n}\sum_{i:T_i=t}  \Big[ Y_i & \Big( \text{logit}\big(\hat{p}^{(k-1)}_t(X_i)\big)+ \sum_{j=1}^d\varepsilon_j \cdot  \hat{S}_{t,j}(X_i) \Big) \\
&\nonumber - \log \Big( 1 + \exp^{ \text{logit}\big(\hat{p}^{(k-1)}_t(X_i)\big)+ \sum_{j=1}^d\varepsilon_j \cdot \hat{S}_{t,j}(X_i) } \Big) \Big]. 
 \end{align}
 where $\hat{S}_{t,j}(X_i) = \frac{\mathds{1}(X_i\in\mathcal{A}_j)}{\mathds{P}(\mathcal{A}_j)}\frac{\mathds{1}(T_i=t)}{\hat{e}_t(X_i)}$. We refer to the above problem as the primal problem. We show that the above optimization problem has the following dual: 
\begin{align}\label{eqn:dual-problem}
 \hat{\lambda}^{(k)}&= \underset{\lambda\geq 0}{\arg\max} -\frac{1}{n}\sum_{i:T_i=t}   \Big[ Y_i \Big( \text{logit}\big(\hat{p}^{(k-1)}_t(X_i)\big)+ \tilde{S}_t(X_i) \cdot \frac{||\bm{\hat{\phi}}(X_i)||_2}{\lambda} \Big)\\
  &\nonumber\quad\quad\quad\quad\quad\quad\quad\quad\quad\quad\quad- \log \Big( 1 + \text{exp}^{ \text{logit}\big(\hat{p}^{(k-1)}_t(X_i)\big)+  \tilde{S}_t(X_i) \cdot \frac{||\bm{\hat{\phi}}(X_i)||_2}{\lambda} } \Big) \Big] - \lambda\delta,
 \end{align} 
 where $\tilde{S}_t(X_i)$ is defined in Section \ref{subsec:proposal}. Providing that the strong duality holds, this primal-dual relationship says that we can estimate the regression coefficient $\hat{\bm{\varepsilon}}^{(k)}$ in the primal problem by either solving the primal problem \eqref{eq:primal-problem} or by solving the dual problem \eqref{eqn:dual-problem} and then exploiting the primal-dual relationship. After reparametrizing  \eqref{eqn:dual-problem} with $\gamma = \frac{||\bm{\hat{\phi}}(X_i)||_2}{\lambda}$, in Lemma \ref{lemma:lagrangian}, we demonstrate that the  unconstrained optimization problem:
\begin{small}
\begin{align}\label{eq:optimization-dual-problem}
\hat{{\gamma}}^{(k)}  = \underset{\gamma\geq 0}{\arg\min}\  -\frac{1}{n}\sum_{i: T_i=t}   \Big[ Y_i & \Big( \text{logit}\big(\hat{p}^{(k-1)}_t(X_i)\big)+ \tilde{S}_t(X_i) \cdot \gamma \Big)- \log \Big( 1 + \text{exp}^{ \text{logit}\big(\hat{p}^{(k-1)}_t(X_i)\big)+  \tilde{S}_t(X_i) \cdot \gamma } \Big) \Big],
\end{align}
\end{small}
is the reparametrized dual problem of the primal problem, and the updated estimate obtained from primal problem with large $d$ yields the same estimate as what we propose whenever $\delta$ is sufficiently close to zero. 

 \begin{lem}[Primal and dual relationship]\label{lemma:lagrangian} The optimization problems \eqref{eq:primal-problem} and \eqref{eqn:dual-problem} form a primal-dual pair, and their solutions satisfy 
$ \varepsilon_j^{(k)} = \frac{\hat{\phi}^{(k-1)}_j(Y_i, T_i, X_i)}{\lambda} $, $j=1, \ldots,d$, where recall $\hat{\phi}^{(k-1)}_j(Y_i, T_i, X_i)$ is defined as 
$$\hat{\phi}^{(k-1)}_j(Y_i, T_i,X_i) = \frac{\mathds{1}(X_i\in\mathcal{A}_j)}{\hat{\mathds{P}}(\mathcal{A}_j)}\frac{\mathds{1}1(T_i=t)}{\hat{e}_t(X_i)}(Y_i - \hat{p}_1^{(k-1)}(X_i)).$$
Reparametrizing $\gamma = \frac{||\bm{\hat{\phi}}(X_i)||_2}{\lambda}$ in the dual problem \eqref{eqn:dual-problem}. Whenever $||\hat{\phi}(X_i)||_2\cdot\delta/\gamma \rightarrow 0$, the dual problem can be represented as
\begin{align*}
  &\arg\min_{\gamma> 0} -\frac{1}{n}\sum_{i=1}^n   \Big[ Y_i \Big( \text{logit}\big(\hat{p}^{(k-1)}_1(X_i)\big)+ \tilde{S}_1(X_i) \cdot \gamma \Big) - \log \Big( 1 + \text{exp}^{ \text{logit}\big(\hat{p}^{(k-1)}_1(X_i)\big)+  \tilde{S}_1(X_i) \cdot \gamma } \Big) \Big],\tag{Reparametrized dual problem}
 \end{align*}%
which yields the same solution as the primal problem. 
\end{lem}

The dual formulation in the discussed method has several benefits. In the reparametrized dual problem \eqref{eq:optimization-dual-problem}, because the estimated propensity score only enters the estimation process after being self-normalized in $\tilde{S}_t(X_i)$,  the discussed approach has the added benefit of being robust to the presence of small estimated propensity scores. Small propensity scores are often encountered in datasets with unbalanced covariate distribution across the treatment and control groups, and such unbalancedness can lead to conventional estimators having substantial bias and large variances. Many numerical studies have found that similar self-normalization of propensity scores provides much more stable estimates of the treatment effects in finite samples. While the original formulation of the primal problem involves a sum over $d$ inverse propensity score weighted clever covariates, its performance can be sensitive to the presence of small propensity score in finite samples.

\subsection{Heuristics of TMLE from a semiparametric inference perspective}\label{subsec:heuristics}

To fully appreciate the one-step targeted maximum likelihood estimator introduced in Section \ref{sec:limitation-of-one-step-TMLE} from a semiparametric perspective, we start with introducing some basics of the classical semiparametric inference framework. To facilitate the discussion, we slightly abuse notations in using $\alpha_0 = E_{X}[ E_{Y|X, T=t}[Y| X , T=t] ]$ (the mean outcome under the treatment arm $t$) to denote the target parameter and in using $P_0$ to denote the true probability measure. Specifically, in a semiparametric model, we observe an i.i.d. sample $ \{O_i\}_{i=1}^n =\{ (Y_i, T_i, X_i)\}_{i=1}^n $ defined on the space $\mathcal{O}$ with a probability measure $P_0$ that possesses a density. The density belongs to the class $\mathcal{M}=\Big\{ p\big(o; \alpha, \eta\big), \ {\alpha}\in \mathcal{A}, \ {\eta}\in\mathcal{H}\Big\}$, 
with respect to some dominating measure $\nu$, where $\mathcal{A}\subset\mathbb{R}$, and $\mathcal{H}$ is an infinite-dimensional set. We denote the true density that generates the data by $p_0\big(o; \alpha_0, \eta_0\big) \in  \mathcal{M}$. Then, 
the parameter of interest is the finite-dimensional parameter $\alpha_0 $, and the nuisance parameter is the infinite-dimensional parameter $\eta_0 = \big(p_t(x), e(x) , f(x) \big)$, where $p_t(x)$ is the conditional density function of $Y$ given $T =t $ and $X= x$, $e(x)$ is the propensity score, and $f(x)$ is the marginal density of $X$ evaluated at $x$. Given the probability measure $P_0$, we can also write the target parameter ${\alpha}_0$ as a statistical function of $P_0$, denoted as ${\alpha}_0(P_0)$.

Under the above semiparametric statistics framework, a natural question raised here: how to evaluate the statistical efficiency of estimators for ${\alpha}_0$ in a semiparametric model? As is often the case in semiparametric statistics, infinite-dimensional problems are tackled by first working with a finite-dimensional problem as an approximation and then taking limit to infinity (\citealp{vaart2000asymptotic}). Therefore, the first step in a semiparametric model is to consider a simpler finite-dimensional ``parameteric submodel" contained in $\mathcal{M}$, and use the theory and methods developed in classical parametric models to obtain an efficient estimator (that typically attains the Cram\`er-Rao lower bound) of $\alpha_0$. Similar to the Cram\`er-Rao lower bound in parametric models, we use ``semiparametric efficiency lower bound" as a metric for evaluating the asymptotic behavior of the semiparametric estimators. Heuristically, the semiparametric efficiency lower bound is simply the supremum of the Cram\`er–Rao bounds for all ``parametric submodels" for estimating $\alpha_0$ (\citealp{tsiatis2006semi}). 

Formally, we define a parametric submodel as 
$\mathcal{M}_{\alpha, \eta_{\varepsilon, S_h}}=\Big\{ p\big(o; \alpha,  \eta_{\varepsilon, S_h} \big), \ {\alpha}\in \mathcal{A}, \  \varepsilon\in \mathcal{E}\subset \mathbb{R} \Big\},$
where $S_h$ is the score function indexed by $h$ which is, intuitively, the direction we perturb the nuisance parameter, and $\varepsilon$ is the perturbing magnitude. In this parametric submodel, the true density is obtained by setting $\alpha = \alpha_0$ and $\varepsilon = 0$.  
Following the above definition, \citet{van2016one} have shown that in a parametric submodel with sufficient smoothness, the Cram\`er-Rao lower bound $\text{CR}(\alpha_0, S_h ) $ for estimating $\alpha_0$ in the submodel $\mathcal{M}_{\alpha, \eta_{\varepsilon, S_h}}$ satisfies 
\begin{align}\label{eq:CR-lower-bound}
\text{CR}(\alpha_0, S_h ) =\lim_{\varepsilon\rightarrow 0} \frac{ \big(\alpha(P_{\varepsilon, S_h}) - \alpha(P_0)\big)^2   }{  -2 {E}_O \big[ \log p(O; \alpha_0,  \eta_{\varepsilon,S_h} ) - \log p (O; \alpha_0, \eta_0)\big] }.
\end{align}
This suggests that $\text{CR}(\alpha_0, S_h ) $ captures the square change in the target parameter divided by the change in the log-likelihood at an infinitesimal $\varepsilon$. 

Because the semiparametric efficiency bound (SPEB) is the supremum of the Cram\`er-Rao lower bound for all parametric submodels, in order to find an estimator of $\alpha_0$ that attains the SPEB, the targeted learning approach looks for a submodel in which a small change in the log-likelihood yields the maximal change in target parameter (\citealp{laantmle2006}). Such a parametric submodel, which maximizes the Cram\`er-Rao lower bound defined in Eq (\ref{eq:CR-lower-bound}), is known as the ``least favorable submodel" in the semiparametric literature (\citealp{vaart2000asymptotic,van2011targeted}). One can thus view the classical targeted learning approach as a principled approach of constructing the least favorable submodel. The one-dimensional least favorable submodel constructed by TMLE allows us to directly work with the conditional likelihood related to the target parameter and perform the usual updating step as in MLE. While the classical MLE method cannot be extended to semiparametric models due to the infinite-dimensional nuisance parameter component, TMLE makes the MLE method feasible in semiparametric models. 

Nevertheless, the least favorable submodel satisfying Eq (\ref{eq:CR-lower-bound}) is only ``locally least favorable" because the density $p(O;\alpha_0,\eta_{\varepsilon,S_h})$ maximizes the Cram\`er-Rao lower bound locally at $\varepsilon=0$. This suggests a practical drawback presents if an initial estimate $\hat{p}_0^{\text{Init}}(o)$ is far away from $p_0(o)$, yielding a large $\hat{\varepsilon}$ as more calibration is needed to push $\hat{p}_0^{\text{Init}}(o)$ towards the truth. A larger calibration yields a larger denominator in Eq (\ref{eq:CR-lower-bound}) and thus smaller change in the target parameter. The consequence is that even though we can iteratively update the initial estimate until $\hat{\varepsilon}\approx 0$, the sample variance has been inflated because each updating step fails to maximize the change in the target parameter while maintaining minimal change in the log-likelihood. 

To resolve the issue that the maximal change in the target parameter is only attained at $\varepsilon=0$, \citet{van2016one} introduce the concept of ``the universal least favorable submodel,"
denoted as   $\mathcal{M}_{\alpha, \eta_{\varepsilon, S_h}}=\Big\{ p\big(o; \alpha,  \eta_{\varepsilon,S_h} \big), \ {\alpha}\in \mathcal{A}, \  \varepsilon\in \mathcal{E}\subset \mathbb{R} \Big\}$.  $\mathcal{M}_{\alpha, \eta_{\varepsilon}}$ is the universal least favorable submodel if
\begin{align*}
      \frac{\partial}{\partial \varepsilon } \Big[  \log p(o; \alpha_0,  \eta_{\varepsilon,S_h} ) - \log p (o; \alpha_0, \eta_0) \Big] = \varphi(o;{\alpha}_0, \eta_{\varepsilon, S_h}),\quad \forall\varepsilon\in \mathcal{E}\subset \mathbb{R}.
\end{align*}
The universal least favorable submodel defined above achieves maximal change in the target parameter as long as $\varepsilon \in \mathcal{E}\subset \mathbb{R}$. The practical benefit of adopting the universal least favorable submodel is that we can avoid inflating the sample variance while performing TMLE updates, especially when the sample size is small. 

 In our example for estimating $\alpha_0$, the universal least favorable submodel takes the form $\mathcal{M}_{\alpha, \eta_{\varepsilon, S_h}}=\Big\{ p\big(o; \alpha,  \eta_{\varepsilon, S_h} \big), \ {\alpha}\in \mathcal{A}, \  \varepsilon\in \mathcal{E}\subset \mathbb{R} \Big\},$
where the nuisance parameter is
\begin{align}\label{eq:nuisance}
\eta_{\varepsilon,S_h}  =  \big(p_t^{\varepsilon,S_h}(x), e(x) , f(x) \big), \quad p_t^{\varepsilon,S_h}(x) =\text{expit}\Big( \text{logit}\big( p_t(x) \big)+ 
\varepsilon\cdot S_h\Big), \quad S_h = t/e_t(x).
\end{align}
Therefore, provided with an initial estimate of the nuisance parameter  $\hat{\eta}_{\varepsilon,S_h}  = (\hat{p}^{\text{Init}}_t(x), \hat{e}_t(x), \hat{f}(x))$ obtained from an i.i.d. sample, the one-step TMLE tries to find $\hat{\varepsilon}$ that minimizes the denominator of the ratio defined in \eqref{eq:CR-lower-bound}, or equivalently, maximizes the likelihood based on the observed data: 
\begin{align*}
\hat{\varepsilon} =&\ \underset{\varepsilon}{\arg\min}\ - \mathbb{E}_n \big[ \log p(O; \alpha_0,  \hat{\eta}_{\varepsilon,S_h} ) - \log p (O; \alpha_0, \eta_0)\big] \\
 =&\ \underset{\varepsilon}{\arg\min}\  -\frac{1}{n}\sum_{i:T_i=t}  \Big[ Y_i \Big( \text{logit}\big(\hat{p}_t(X_i)\big)+ \varepsilon \cdot  \hat{S}_{t}(X_i) \Big)- \log \Big( 1 + \exp^{ \text{logit}\big(\hat{p}_t(X_i)\big)+ \varepsilon \cdot \hat{S}_{t}(X_i) } \Big) \Big],
\end{align*}
where $\hat{S}_t(X_i) = \mathds{1}(T_i=t)/\hat{e}_t(X_i)$. Then, we see that $\hat{\varepsilon}$ is equivalent to the estimated coefficient from running the logistic regression discussed in Section \ref{sec:limitation-of-one-step-TMLE}.

\section{Simultaneous Confidence Intervals }\label{subsec:confidence-interval}

To construct a level-$q$ confidence interval for a single subgroup $j$, we work with
$\hat{\alpha}_{t,j} \pm \Phi^{-1}(1-q/2) \cdot\Big(\frac{\hat{\Sigma}_{t,jj}}{n}\Big)^{1/2}$, 
where $\hat{\Sigma}_t$ is the estimated covariance matrix with
\begin{align}\label{eq:estimated-covariance-matrix}
& \nonumber\hat{\Sigma}_t = \big( \hat{\Sigma}_{t, jk}\big)_{j,k=1}^d = \frac{1}{n}\sum_{i=1}^n \bm{\hat{\varphi}}_{t,i}\bm{\hat{\varphi}}_{t,i}^\intercal , \quad \bm{\hat{\varphi}}_{t,i} = \big(\hat{\varphi}_{t,1}(Y_i,T_i,X_i), \ldots, \hat{\varphi}_{t,d}(Y_i,T_i,X_i)\big)^\intercal,\\
& \hat{\varphi}_{t,j}(O_i) = \frac{1}{n}\sum_{i=1}^n \frac{ \mathds{1}(X_i \in \mathcal{A}_j)}{\hat{P}(\mathcal{A}_j)}\Big[\Big(\frac{T_i}{\hat{e}_t(X_i)}\big(Y_i-\hat{p}_t(X_i)\big)+\hat{p}_t(X_i)-\alpha_{t,j}\Big).
\end{align}

To construct a simultaneous level-$q$ confidence interval though, let $\hat{\kappa}(q,  \tilde{\Sigma}_{t} )$ be a consistent estimate of the $(1-q)$-th quantile of  $\max_{j\in 1, \ldots, d} |Z_j|$, where $(Z_1, \ldots, Z_d)^\intercal \sim N\big(0,\tilde{\Sigma}_t \big) $ with $\tilde{\Sigma}_t = \big( \tilde{\Sigma}_{t, jk}\big)_{j,k=1}^d  $ and $ \tilde{\Sigma}_{t, jk} = \frac{ \hat{\Sigma}_{t, jk} }{\sqrt{ \hat{\Sigma}_{t, jj}\hat{\Sigma}_{t, kk} }}$. Then, the constructed simultaneous confidence interval satisfies 
\begin{align}\label{eq:multi-subgroup-CI}
\lim_{n\rightarrow \infty} {P}\left(\hat{\alpha}_{t,j} \pm\hat{\kappa}(q,  \tilde{\Sigma}_{t} ) \cdot\Big(\frac{\hat{\Sigma}_{t,jj}}{n}\Big)^{1/2}, j =1, \ldots, d\right) = 1-q. 
\end{align}
Such a simultaneous confidence interval ensures that all the confidence intervals cover the corresponding true subgroup parameter at the same time.

\section{Simulation Studies}\label{sec:simulation}

To demonstrate the merit of the proposed method (iTMLE), we compare it with some conventional estimators under overlapping and non-overlapping subgroups cases. We compare the proposed method with a doubly robust (DR) estimator and a generalized linear model estimator (GLM), and we compare the cross-fitted version of iTMLE with the double machine learning (DML) method, since DML also utilizes cross-fitting. Before we present our simulation results, we summarize two main takeaways from the simulation studies for our readers: (1) The proposed method has smaller bias, smaller variance, and lower family-wise error rate (FWER) compared to the considered estimators in finite samples. Recall that FWER refers to the probability of at least one constructed simultaneous confidence interval excluding the truth; (2) With cross-fitting, the proposed method shows enhanced finite sample performance in terms of smaller bias than the implementation without cross-fitting. 

We measure the performance of various estimators according to their  $\sqrt{n}$-scaled biases (computed as the root-$n$ sum of mean differences between the Monte Carlo estimates and the true parameter across multiple subgroups), standard deviations (computed as the root-$n$ sum of standard deviations of the Monte Carlo estimates across multiple subgroup), and FWER (computed as the proportion of Monte Carlo samples in which at least one constructed confidence interval for multiple subgroups excluding the truth). 
We scale the bias and variance by the sample size as they converge to zero as  $n$ goes to infinity. 

\subsection{Simulation design}
Our simulation design mimics observational studies where treatments are assigned based on covariates. 
We simulate 1000 random Monte Carlo samples from: 
 $X = (X_1, \ldots, X_5)^{\intercal} \sim N(0, \Sigma), \Sigma_{ij} = 0.5^{|i-j|}, T \sim \text{Bernoulli}\Big(\text{expit}(X_{1} - 0.5\cdot X_{2} + 0.25 \cdot X_{3} + 0.1 \cdot X_{4})\Big)$, and $Y | T, X \sim \text{Bernoulli}\Big( \text{expit} (21 + T + 27.4\cdot X_{1} + 13.7\cdot X_{2} + 13.7\cdot X_{3} + 13.7\cdot X_{4}) \Big)$. We consider this specific simulation design because the design has been frequently adopted in the causal inference literature \citep[see][for example]{kang2007demystifying, imai2014covariate}. This enables us to better compare our approach with existing methods. Kindly pointed out by an anonymous reviewer, the above simulation design produces rather deterministic outcomes, and we thus provide additional simulation results under an alternative simulation design in Web Appendix \ref{appendix:alternative-sim}.

We consider two types of subgroups: overlapping subgroups and non-overlapping subgroups. Overlapping subgroups with moderate $d$, $d=4$, are generated by
$\mathcal{A}_1 = \{X_1> \Phi^{-1}(0.1)\},  \mathcal{A}_2 = \{ \Phi^{-1}(0.1)< X_2< \Phi^{-1}(0.9)\},  \mathcal{A}_3 = \{X_3 + X_4 >-2 \} , \mathcal{A}_4 = \{ \mathds{1}_{X_4>0.5}>-1\}.$
Non-overlapping subgroups with large $d$, $d=10$, are generated by $\mathcal{A}_{j} = \big\{Q_{X_1}(j/10) <X_1 < Q_{X_1}((j+1)/10 ) \big\}, j = 1, \ldots, 10. $
For simplicity, in the following simulation studies, the considered parameter is  $\bm{\alpha}_1 = (\alpha_{1,1},\ldots,\alpha_{1,d})^\intercal$.

\subsection{Comparison with conventional estimators}\label{Section:Simulation-comparison-with-others}

We generate initial estimates of $e_t(\cdot)$ and $p_t(\cdot)$ through logistic regression, random forest, or gradient boosting, implemented in \texttt{R} packages \texttt{stats}, \texttt{ranger} (\citealp{ranger2020}), and \texttt{xgboost} (\citealp{xgboost2020}). We compare the proposed iterative one-step TMLE method (iTMLE) with the doubly robust estimator, a simple regression adjustment estimator, and the inverse propensity score estimator, which are defined as 
$ \hat{\alpha}_{t,j}^{\text{DR}}  = \frac{1}{n_j}\sum_{i\in\mathcal{A}_j} \Big[ \frac{T_i}{\hat{e}_t(X_i)}(Y_i-\hat{p}^{\text{Init}}_t(X_i)) + \hat{p}^{\text{Init}}_t(X_i)\Big], \ 
     \hat{\alpha}_{t,j}^{\text{GLM}} = \frac{1}{n_j}\sum_{i\in\mathcal{A}_j}  \hat{p}_t^{\text{Init}}(X_i), \ \hat{\alpha}_{t,j}^{\text{IPW}} = \frac{1}{n_j}\sum_{i\in\mathcal{A}_j}  \frac{T_i}{\hat{e}_t(X_i)}Y_i.
$
Simultaneous confidence intervals for these estimators are constructed using standard large sample theory adopted in the literature (see 
\citet{hahn1998propensity} for the doubly robust estimator and  \citet{van2011ipw} for the IPW estimator). We provide finite-sample comparisons through Figure \ref{fig:sim1-comparison}(A) -- (C) for overlapping subgroups, and Figure \ref{fig:sim1-comparison}(D) -- (E) for non-overlapping subgroups. As the IPW estimator has much larger variance than the other estimators, we exclude its results from these figures. From Figure \ref{fig:sim1-comparison}, we observe that the iTMLE estimator outperforms the others for bias, standard deviation, and FWER, regardless of how $e_1(\cdot)$ and $p_1(\cdot)$ are estimated in the first stage. This is in-line with our theoretical results because the proposed estimator consists of a data-adaptive bias correction term (Section \ref{sec:method}), which largely improves its finite sample performance. In addition, among all three initial estimators, random forest overall seems to be a winner.

\begin{figure}[!p]
\centering
\includegraphics[width=\textwidth]{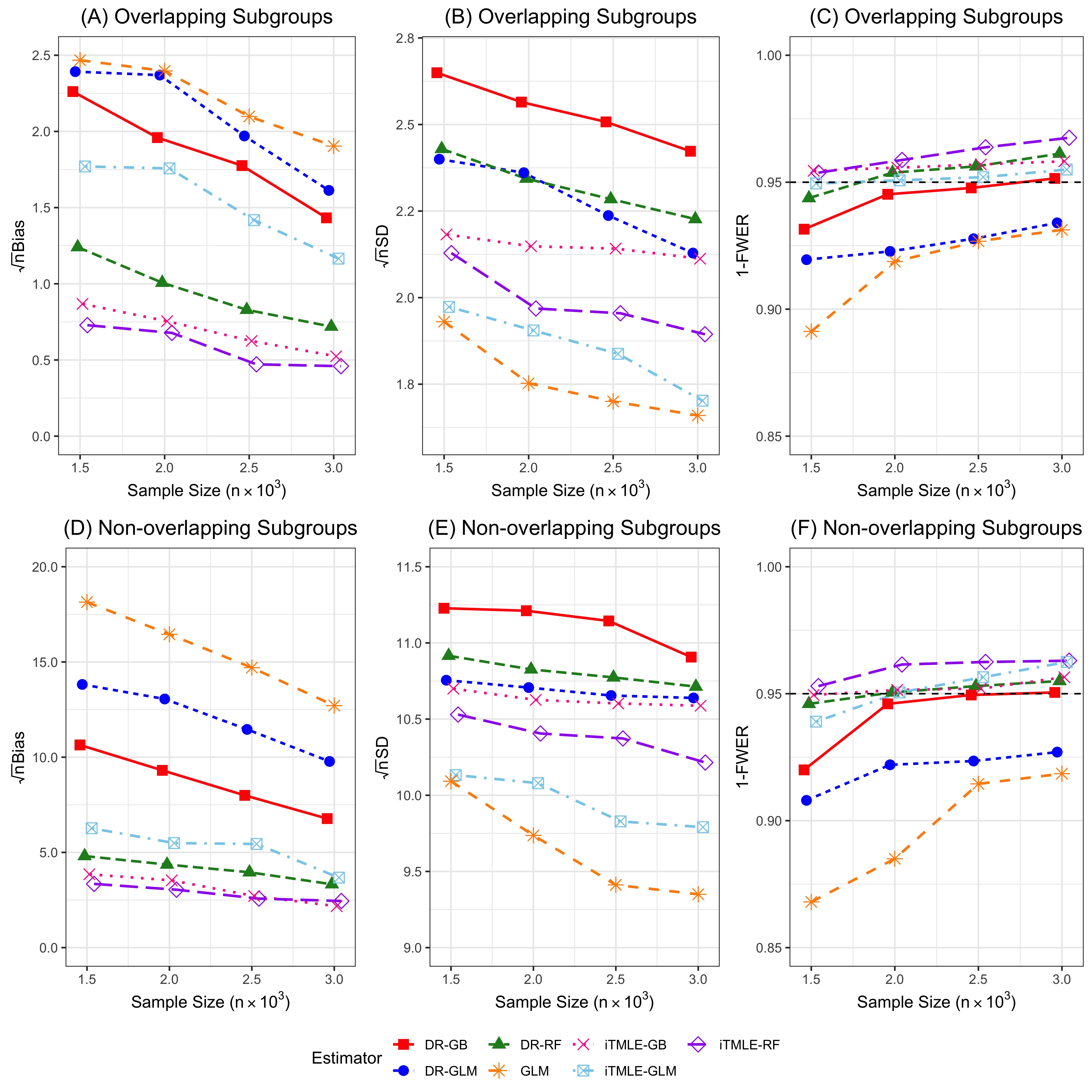}
\caption{Comparison of bias, standard deviation (scaled by root-$n$), and (1-FWER) in overlapping and non-overlapping subgroups. ``iTMLE" denotes the proposed estimator. ``DR" denotes the doubly robust estimator. ``GLM" denotes the generalized linear models. The maximum Monte Carlo standard error of (1-FWER) is 0.026 for iTMLE, 0.028 for DR, and 0.022 for GLM. ``The maximum Monte Carlo standard error of (1-FWER)" refers to the largest standard error of (1-FWER) (out of all three considered estimators for the propensity score and the conditional expectation of the outcome based on logistic regression, random forest, and gradient boosting) computed from Monte Carlo samples.}
\label{fig:sim1-comparison}
\end{figure}




\subsection{Comparison with the double machine learning}

In this part of the simulation study, we compare the performance of the cross-validated version of iterated one-step TMLE for multiple parameters with the double machine learning (DML) method (\citealp{chernozhukov2017double}). DML also involves the estimations of the propensity score model and the conditional mean model, and it is a meta-learning method that relies on Neyman orthogonal score and cross-fitting to generate debiased estimates for the causal estimands. The simulation results of the three-fold cross-validated iTMLE and DML (implemented with the \texttt{R} package \texttt{DoubleML} (\citealp{DoubleML2021R}))
 are presented in Figure \ref{fig:sim2-cv-comparison}. There are two takeaways from the summarized results in Figure \ref{fig:sim2-cv-comparison}. First, the performance of CV-iTMLE surpasses DML. Although DML is rather robust compared to the doubly robust estimator, it still yields larger bias and variance than CV-iTMLE. Second, compared to the iTMLE implementation without cross-fitting (Figure \ref{fig:sim1-comparison}), CV-iTMLE shows a faster convergence rate. We conjecture that the sample splitting step allows the non-parametric estimators in the initial stage to converge faster and thus shows more robust performance (smaller bias, smaller standard deviation, and smaller FWER).

\begin{figure}[!p]
\centering
\includegraphics[width=\textwidth]{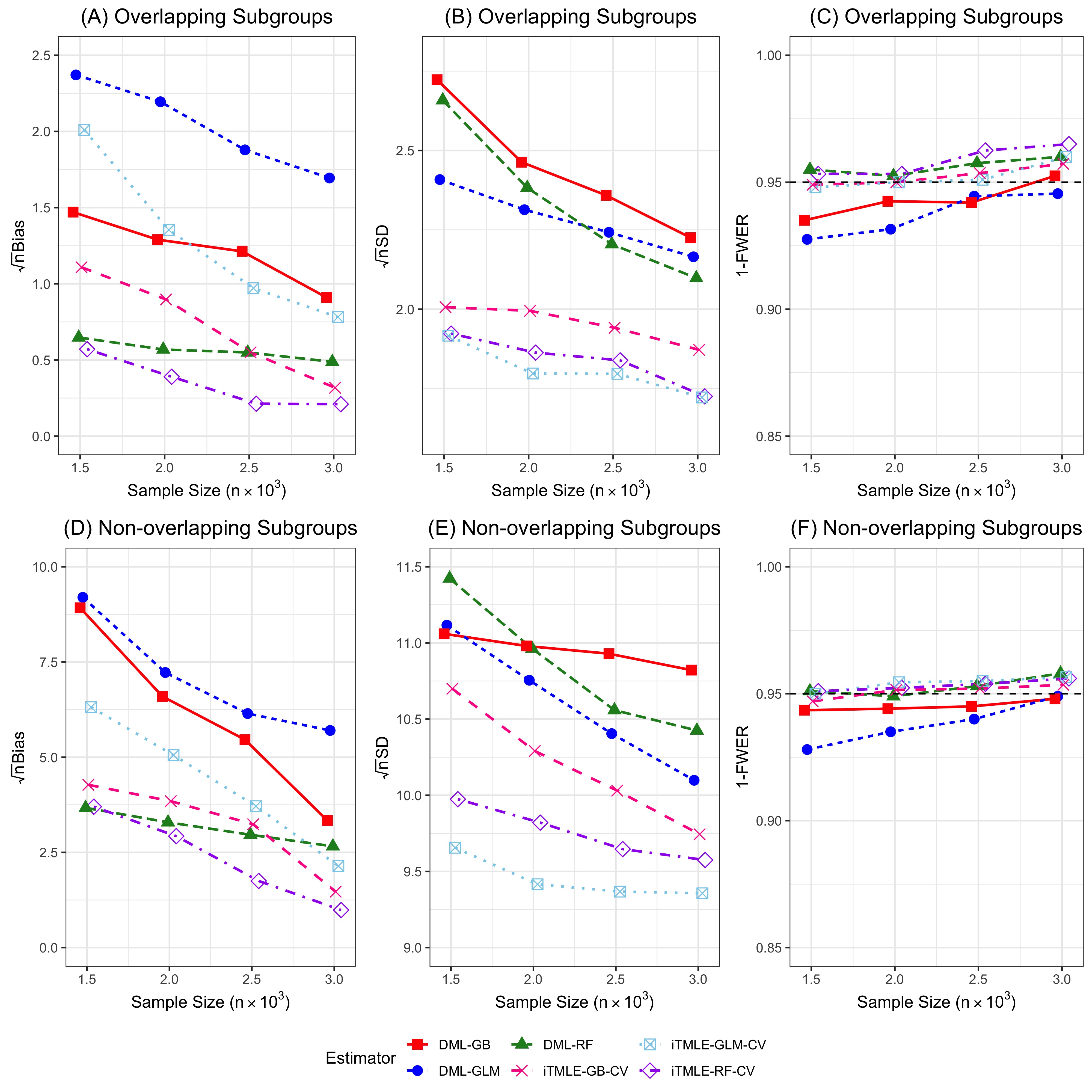}
\caption{Comparison of the cross-validated iTMLE implementation and the double machine learning method. ``iTMLE-CV" denotes the proposed method with cross-fitting. ``DML" denotes the double machine learning method. The maximum Monte Carlo standard error of (1-FWER) is 0.024 for CV-iTMLE and 0.026 for DML. ``The maximum Monte Carlo standard error of (1-FWER)" refers to the largest standard error of (1-FWER) (out of all three considered estimators for the propensity score and the conditional expectation of the outcome based on logistic regression, random forest, and gradient boosting) computed from Monte Carlo samples.}
\label{fig:sim2-cv-comparison}
\end{figure}

\section{Case Study in UK Biobank Data}\label{sec: real-data}


Statins are the most commonly prescribed  cholesterol-lowering medications in the United States. Cholesterol's role in $\beta$-amyloid processing and the potential link between serum cholesterol levels and AD pathology (\citealp{reed2014associations}) have led to the argument that cholesterol-moderating drugs such as statins could reduce the risk of AD onset and progression. However, this argument is controversial by current evidence. Several cohort studies found a negative association between statin usage and AD (\citealp{zissimopoulos2017sex}), while others have failed to replicate those findings. These inconsistent findings might be due to the effect of statins on AD varying across gender, age, and other subgroups (\citealp{zissimopoulos2017sex}). Thus, we hypothesize that statin usage has significant benefits of reducing AD risk in some (but not all) subgroups. To test this hypothesis, we analyzed data in the UK Biobank to investigate the heterogeneous treatment effect of inheriting rs12916-T allele, a proxy for statin usage, on AD risk in the White British subpopulations. We considered a cross-sectional study design by looking at the disease prevalence at the end of year 2021. 

\subsection{Study design}
 The UK Biobank study recruited 502,536 participants aged from 40 to 69 in the United Kingdom from 2006 to 2010.
We defined AD status by integrating information provided by Hospital Episode Statistics, death registries, and self-reported diagnoses (see details in Web Appendix \ref{appendix:data-details}). We restricted our study to 293,929 White British individuals. These individuals are unrelated and had passed standard quality control steps.



Instead of directly adopting statin usage as a treatment variable, we adopted a genetic variant rs12916-T as a surrogate treatment variable. This means that if the subject carries the variant rs12916-T, the treatment indicator variable is set to be $T = 1$;  otherwise, $T$ is set to be zero. We adopted this genetic surrogate biomarker as the treatment variable for two reasons. On the one hand, the rs12916-T allele only affects the LDL cholesterol concentration through HMGCR inhibition, and it is thus functionally equivalent to statin usage (\citealp{swerdlow2015hmg,guo2022assessing}). More specifically, the decreased LDL cholesterol level associated with statin usage is similar to the association pattern with \textit{\textit{rs12916-T}} ($R^2 = 0.94$) \citep{wurtz2016metabolomic}, thus rs12916-T is a sensible surrogate treatment variable for statin usage. On the other hand, given that genetic variants are randomly inherited from parents, our treatment variable (whether or not the individual carries rs12916-T) is thus independent of unmeasured confounding factors such as lifestyle modifications after statin usage, potentially making Assumption \ref{assump:unconfound} more plausible.



To account for genetic pleiotropy, we adjusted for SNPs that are associated with LDL. Briefly, we selected $385$ independent genome-wide significant SNPs (with $p$-values less than $5\times 10^{-8}$ and $R^2 < 0.01$) associated with LDL according to the published genome-wide association study (GWAS) results harmonized in GWAS Catalogue (\citealp{gwascatalogue2016}). We further adjusted for age and gender variables, which may improve estimation efficiency given their associations with the outcome. It was our hope that this study design could increase the plausibility of Assumption \ref{assump:unconfound}. 

We investigated the effect of inheriting rs12916-T allele on AD risk in the following subgroups: (1) males, (2) females, (3) age $<65$, (4)age $\geq 65$, (5) individuals with high AD genetic risk, and (6) individuals with low AD genetic risk. Notably, ``high AD genetic risk" was defined as either a subject's parents or siblings being diagnosed with AD, while ``Low AD genetic risk" was defined as neither a subject's parents nor siblings being diagnosed with AD. We compared the performance of the proposed method (CV-iTMLE) with the double machine learning (DML) method and the widely used generalized linear models (GLM). We used the random forest as our first stage estimator as it provides the most robust results in our simulation studies.

Because statin usage may increase the risk of T2D (\citealp{swerdlow2015hmg}), as a secondary analysis, we further investigated the  effect  of  inheriting rs12916-T allele  on  T2D to evaluate the potential heterogeneous side effects. The study design and study results of this secondary analysis can be found in Web Appendix \ref{appendix:case-study-t2d}.


\subsection{Results}

\begin{figure}[!p]
\centering
\includegraphics[width=\textwidth]{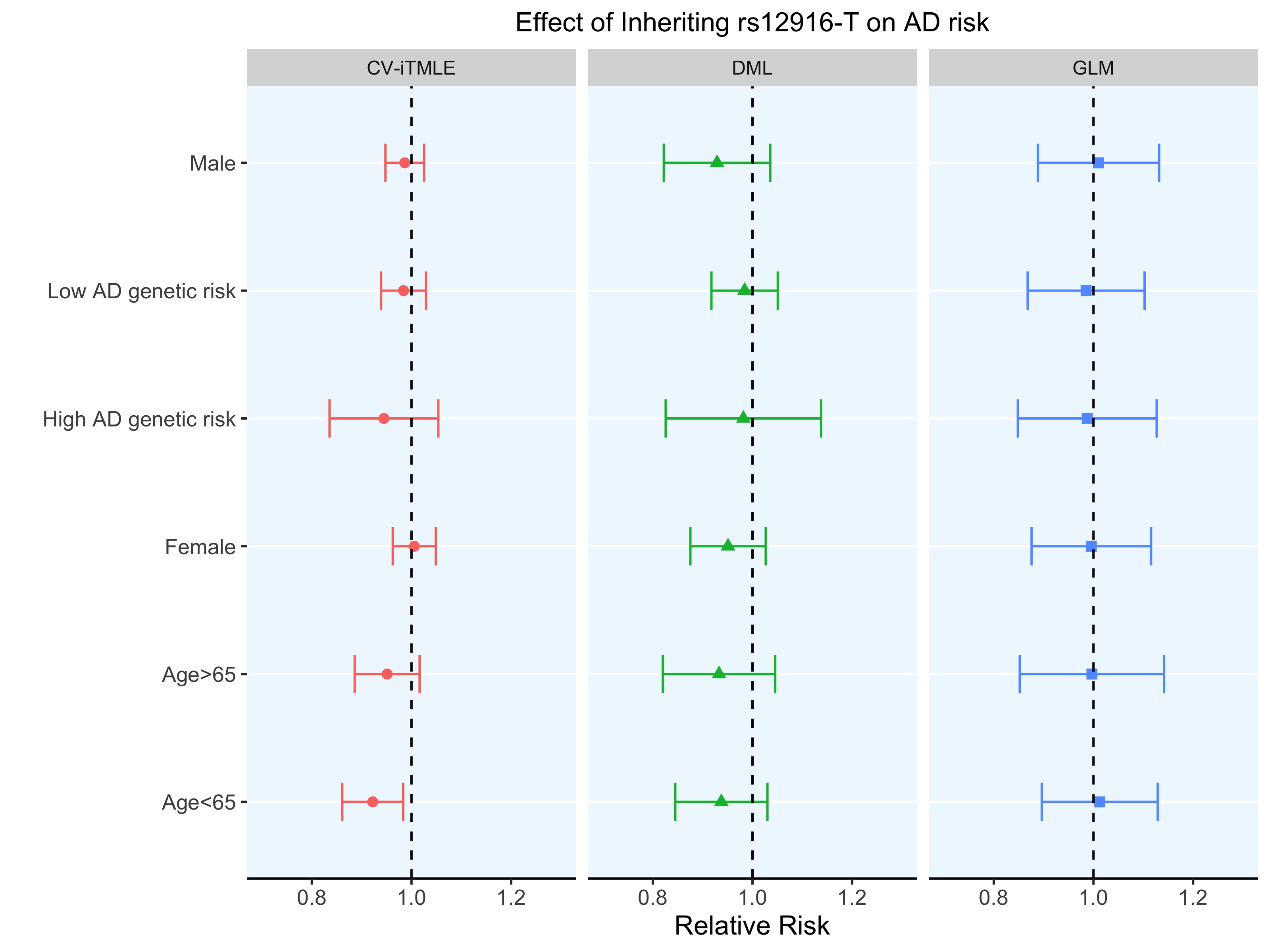}
\caption{The effect of inheriting rs12916-T allele (a proxy for statin usage) on the risk of developing Alzheimer's disease (AD) in the UK Biobank white British population ($n=293,929$). ``DML" denotes the double machine learning method. ``GLM" denotes the generalized linear models. GLM is used for association test and does not imply causal relationships. ``CV-iTMLE" denotes the cross-validated iTMLE method.}
\label{fig:real-data-all-high}
\end{figure}

Figure \ref{fig:real-data-all-high} summarizes the effect of inheriting rs12916-T (a proxy for statin usage) on AD risk in considered subgroups. As the GLM was applied to each subgroup separately and the sample size was much smaller, leading to non-significant associations for all the subgroups. The DML method also did not find any significant effects in all subgroups. This might be caused by small estimated propensity scores, leading to large variability in finite samples. In contrast, by targeting all subgroups simultaneously, the proposed method suggested that carrying rs12916-T allele is protective against AD in the subgroup younger than $65$ (RR: 0.92, 95\% CI: 0.86--0.98). In sum, our proposed method showed shortened confidence intervals with improved statistical power in detecting significant subgroups, while the GLM and DML methods tend to lose power.

Lastly, we acknowledge that the current study design has several potential limitations. First, our study only investigated the treatment effect of carrying rs12916-T allele or not. Although this genetic variant is a sensible proxy for statin usage, the findings from this study need to be interpreted cautiously.
Second, our study was based on UK Biobank data and only focused on White British population. UK Biobank participants were healthier (e.g., fewer self-reported health conditions) than the general population. Thus, our findings may not be generalizable to other populations.

\section{Concluding Remarks and Discussions}
In this manuscript, we propose a semiparametric efficient method for simultaneous heterogeneous treatment effect estimation across multiple subgroups. The proposed method allows us to construct a powerful multiple testing procedure leveraging the subgroup dependence structure. In our empirical studies, the proposed method demonstrates finite sample improvements compared to other conventional methods, including the doubly robust estimator, the classical one-step TMLE, and the double machine learning method. This manuscript opens a variety of possibilities for future research. From a methodological perspective, 
our current method can be extended to work with other types of outcomes. For example, if the outcome is continuous, one can either modify the updating step \citep{gruber2010targeted}, or dichotomize a continuous outcome into binary values (more details can be found in Web Appendix \ref{appendix:continuous-outcome}).  From an application perspective, the proposed method can be flexibly adapted to various clinical studies and assist the evaluation of other subpopulations of interest.

\bigskip

\noindent\textbf{Acknowledgements.} We thank the editor, the associate editor and the reviewer for their insightful comments and suggestions. We thank the individuals involved in the UK Biobank study for their participation and the research teams for their work on collecting, processing, and sharing these datasets. This research is supported in part by NIH awards R01AI074345, 1R03AG070669, 1R01CA263494 and R01MH125746, and NSF award DMS-2015325. 

\bigskip

\noindent\textbf{Data Availability Statement.}
The UK Biobank data can be requested from \url{https://www.ukbiobank.ac.uk/}. This research has been conducted using the UK Biobank Resource (application number 48240), subject to a data transfer agreement.

\bigskip

\noindent\textbf{Supporting Information.} Web Appendices, Tables, and Figures referenced in Sections 1,3,4,5,6, and 7 are available in Appendices. Software in the form of R code is available on the Github repository \url{https://github.com/WaverlyWei/iTMLE}.


\clearpage

\bibliographystyle{jasa} 
	\bibliography{reference}

	\clearpage

\appendix
\addcontentsline{toc}{section}{Appendix} 
\part{Appendix} 
\parttoc 
\onehalfspacing
\setcounter{page}{1}

\section{Proof of Theorems}\label{appendix:proof-theorem-1}
\subsection{Proof of Theorem 1}	

Suppose that $\{O_i\}_{i=1}^n: =\{ (Y_i, T_i, X_i)\}_{i=1}^n $ are i.i.d. random variables defined on the space $\mathcal{O}$ with probability measure $P$. For a real-valued function $f$ on $\mathcal{X}$, the empirical measure $\mathbb{P}_n$ is defined by $\mathbb{P}_n (f):=  \frac{1}{n}\sum_{i=1}^n f(O_i)$, and the empirical mean is defined as 
$\mathbb{E}_n [f(X)]:=  \frac{1}{n}\sum_{i=1}^n f(X_i)$. 
If $\mathcal{F}$ is a collection of real-valued functions defined on $\mathcal{X}$, then $\{ \mathbb{P}_n f: f \in \mathcal{F}\}$ is the empirical measure indexed by $\mathcal{F}$. We assume that $Pf = \int f \mbox{d} P$ exists for each $f\in\mathcal{F}$. Note that such a notation can be more helpful as we can treatment random functions. We use $E_X[f(X)]$ to denote expectation taken with respect to the random variable $X$ when it is more convenient to simplify notations.  Our procedure estimates the joint distribution of $(Y_i, T_i, X_i)$, and the estimated density evaluated at $(y,t, x)\in\mathcal{X}$ is  $\hat{p}_t(Y_i = y | X_i = x)\cdot \hat{e}_t(x)\cdot  \hat{p}(x)$. We denote the joint density of $(Y_i, T_i, X_i)$ as $p(y,t,x)$ which can be decomposed into the product $p(y,t,x) = p_t(x) e_t(x) p(x)$. Here $e_t(x) = te(x) + (1-t)(1-e(x))$. We denote the probability measure defined by such an estimated density as $\mathbb{P}_n^*$. Given the probability measure $P$, our target parameter $\bm{\alpha}_t$ can also be written as a statistical function of $P$, denoted as $\bm{\alpha}_t(P)$. Similarly, our proposed estimator can be written as $\bm{\alpha}_t(\mathbb{P}_n^*)$. The vector of efficient influence function of our target parameter $\bm{\alpha}_t(P)$ is denoted as 
\begin{align*}
\bm{\psi}\big(o; \bm{\alpha}_t(P), \bm{\eta}(P)\big) := \Big( \psi_1\big(o; \bm{\alpha}_t(P), \bm{\eta}(P) \big),\ldots, \bm{\psi}_J\big(o; \bm{\alpha}_t(P), \bm{\eta}(P)  \big)\Big)^{\top},
\end{align*}
where $\bm{\eta}$ contains the nuisance parameters. 
We decompose the efficient influence function into two parts: 
\begin{align*}
& \psi_j^{(1)}\big(o;  \bm{\eta}(P) \big) = \frac{ \mathds{1}\{x \in \mathcal{A}_j\}}{P(\mathcal{A}_j)}\Big[\frac{\mathds{1}\{t = 1\} }{e(x)} + \frac{\mathds{1}\{t = 0\} }{1 - e(x)}\Big]\big(y-p_t(x)\big), \\
& \psi_j^{(2)}\big(o; \bm{\alpha}_t(P), \bm{\eta}(P) \big) =  \frac{ \mathds{1}\{x \in \mathcal{A}_j\}}{P(\mathcal{A}_j)} \big(p_t(x)- \alpha_t \big), 
\end{align*}
where the first part does not depend on the target parameter, and $\bm{\alpha}_t = \int p_t \mbox{d}P = \mathbb{E}_X[p_t(X)]$ is a scalar. Similarly, we then decompose the vector of efficient influence into 
\begin{align*}
\bm{\psi}\big(o; \bm{\alpha}_t(P), \bm{\eta}(P)\big) = \bm{\psi}^{\text(1)}\big(o; \bm{\eta}(P)\big) + \bm{\psi}^{(2)}\big(o; \bm{\alpha}_t(P), \bm{\eta}(P)\big). 
\end{align*}

Our estimator satisfies the following expansion:
\begin{align}\label{eq:expansion-empirical-process}
\bm{\hat{\alpha}}_t(\mathbb{P}_n^*) - \bm{\alpha}_t(P) =  (\mathbb{P}_n-P) \bm{\psi}(o; \bm{\alpha}_t(P) , \bm{\eta}(P) ) + I_{n1} + I_{n2},
\end{align}
where the remainder terms are derived in Section \ref{Sec:empirical-process-remainder}
\begin{align*}
I_{n1} & = R_{n2} + R_{n3} = ( \mathbb{P}_n  - P)\big( \bm{\psi}(o; \bm{\alpha}_t(P), \bm{\eta}(\mathbb{P}_n^*) )  - \bm{\psi}(o;\bm{\alpha}_t(P), \bm{\eta}( P) )  \big) ,\\
I_{n2} & = R_{n4} = P \bm{\psi}\big(o; \bm{\alpha}_t(P), \bm{\eta}(\mathbb{P}_n^* )\big) = P \big[\bm{\psi}\big(o; \bm{\alpha}_t(P), \bm{\eta}(\mathbb{P}_n^* )\big)-\bm{\psi}\big(o; \bm{\alpha}_t(P), \bm{\eta}({P} )\big)\big].
\end{align*}
Let $\bm{\alpha}_t(P) =\bm{\alpha}_t $ and $\bm{\eta}(P) = \bm{\eta}_0$, and let the estimated parameters $\hat{\bm{\alpha}}_t(\mathbb{P}_n^*) = \hat{\bm{\alpha}}_t $ and $\bm{\eta}(\mathbb{P}_n^*) =\hat{\bm{\eta}}$. The above remainder terms simplify to 
\begin{align*}
   I_{n1} & =  \ ( \mathbb{P}_n  - P)\big( \bm{\psi}(o; \bm{\alpha}_t, \hat{\bm{\eta}} )  - \bm{\psi}(o;\bm{\alpha}_t, \bm{\eta}_0 )  \big),\\
   I_{n2} & = P \bm{\psi}\big(o; \bm{\alpha}_t, \hat{\bm{\eta}}\big) = P \big[\bm{\psi}\big(o; \bm{\alpha}_t,\hat{\bm{\eta}} \big)-\bm{\psi}\big(o; \bm{\alpha}_t, \bm{\eta}_0\big)\big].
\end{align*}
For any $\bm{\eta}$ in the nuisance estimator realization set $\mathcal{H}_n$, the term $I_{n1} = o_P(1)$ uniformly over $P\in\mathcal{P}_n$ under Assumption 4.2. To bound $I_{n2}$, we introduce the function 
\begin{align*}
f(r) = P\Big[ \bm{\psi}\big(o; \bm{\alpha}_t, \bm{\eta}_0 + r( \hat{\bm{\eta}}-\bm{\eta}_0)\big)\Big] .
\end{align*}
By Taylor expansion and the fact that $ f(0) = 0$, we have 
\begin{align*}
f(1) = f'(0) + f''(\tilde{r})/2, \quad \text{for some }\tilde{r}\in(0,1).
\end{align*}
We will next verify in Step 1 that $f'(0) = 0 $ and verify in Step 2 that $\sup_{r\in [0,1)} ||f''(r)|| = o_P(1/\sqrt{n})$. This finishes our proof. 

\noindent\textbf{Step 1.} For any $\bm{\eta}\in \mathcal{H}_n $, the first order derivative is equal to
\begin{align*}
f'(0)  = &\ \partial_{\eta} P \bm{\psi}\big(o; \bm{\alpha}_t, \bm{\eta}_0\big) [\bm{\eta}-\bm{\eta}_0] \\
=&\ \partial_r \{ P\big[  \bm{\psi}\big(o; \bm{\alpha}_t, \bm{\eta}_0 + r( \bm{\eta}-\bm{\eta}_0\big)\big] \}\big|_{r=0},
\end{align*}
First, we want to show $\bm{\psi}$ is Neyman orthogonal, such that $f'(0) = 0$. 

We denote $\bm{\eta} = (\check{e}(X), \check{p}_1(X), \check{p}_0(X), \check{\mathbb{P}}(\mathcal{A}_1), \ldots, \check{\mathbb{P}}(\mathcal{A}_J))$. The Gateaux derivative in the direction of $\bm{\eta} - \bm{\eta}_0 = \big(\check{e}(X)-e(X), \check{p}_1(X)-p_1(X), \check{p}_0(X)-p_0(X), \check{\mathbb{P}}(\mathcal{A}_1)-P(\mathcal{A}_1), \ldots, \check{\mathbb{P}}(\mathcal{A}_J) - P(\mathcal{A}_J)\big)$ is
\begin{align*}
   &\quad \partial_r \{ E\big[  \bm{\psi}\big(o; \bm{\alpha}_t, \bm{\eta}_0 + r( \bm{\eta}-\bm{\eta}_0\big)\big] \}\\
 &=\begin{pmatrix}
   & \partial_r E \Bigg[\frac{ \mathds{1}\{x \in \mathcal{A}_1\}}{P(\mathcal{A}_1) + r\big(\check{\mathbb{P}}(\mathcal{A}_1) -P(\mathcal{A}_1)\big)}\Big[\Big(\frac{\mathds{1}\{t = 1\} }{e(x)+r\big(\check{e}(x)-e(x)\big)} + \frac{\mathds{1}\{t = 0\} }{1 - e(x)-r\big(\check{e}(x)-e(x)\big)}\Big)\\
     &\quad\quad\quad\quad\cdot\Big(y- p_t(x) - r\big(\check{p}_t(x) - p_t(x)\big)\Big)+\big(p_t(x) + r(\check{p}_t(x)-p_t(x))- \alpha_t \big)\Big]\Bigg],\\
    &\vdots\\
    &\partial_r E \Bigg[\frac{ \mathds{1}\{x \in \mathcal{A}_J\}}{P(\mathcal{A}_J) + r\big(\check{\mathbb{P}}(\mathcal{A}_J) -P(\mathcal{A}_J)\big)}\Big[\Big(\frac{\mathds{1}\{t = 1\} }{e(x)+r\big(\check{e}(x)-e(x)\big)} + \frac{\mathds{1}\{t = 0\} }{1 - e(x)-r\big(\check{e}(x)-e(x)\big)}\Big)\\
   &\quad\quad\quad\quad\cdot\Big(y- p_t(x) - r\big(\check{p}_t(x) - p_t(x)\big)\Big)+\big(p_t(x) + r(\check{p}_t(x)-p_t(x))- \alpha_t \big)\Big]\Bigg]
 \end{pmatrix}
\end{align*}
The Gateaux derivative for each subgroup $j$ is
\begin{align*}
   &= \partial_r E \Bigg[\frac{ \mathds{1}\{x \in \mathcal{A}_j\}}{P(\mathcal{A}_j) + r\big(\check{\mathbb{P}}(\mathcal{A}_j) -P(\mathcal{A}_j)\big)}\Big[\Big(\frac{\mathds{1}\{t = 1\} }{e(x)+r\big(\check{e}(x)-e(x)\big)} + \frac{\mathds{1}\{t = 0\} }{1 - e(x)-r\big(\check{e}(x)-e(x)\big)}\Big)\\
   &\quad\quad\quad\quad\cdot\Big(y- p_t(x) - r\big(\check{p}_t(x) - p_t(x)\big)\Big)+\big(p_t(x) + r(\check{p}_t(x)-p_t(x))- \alpha_t \big)\Big]\Bigg],\\ 
   &= \partial_r E \Bigg[\frac{ \mathds{1}\{x \in \mathcal{A}_j\}}{P(\mathcal{A}_j) + r\big(\check{\mathbb{P}}(\mathcal{A}_j) - P(\mathcal{A}_j)\big)}\cdot (p_t(x)+r(\check{p}_t(x)-p_t(x))-\alpha_t)\\ &\quad\quad+\frac{ \mathds{1}\{x \in \mathcal{A}_j\}}{P(\mathcal{A}_j) + r\big(\check{\mathbb{P}}(\mathcal{A}_j) - P(\mathcal{A}_j)\big)}\cdot\Big( y-p_t(x)-r(\check{p}_t(x) - p_t(x))\Big)\\
   &\quad\quad\cdot\Big[\frac{\mathds{1}\{t = 1\} }{e(x) + r\big(\check{e}(x)-e(x)\big)} + \frac{\mathds{1}\{t = 0\} }{1 - e(x)-r\big(\check{e}(x)-e(x)\big)}\Big]\Bigg],\\ 
 &=E\Bigg[ \frac{ \mathds{1}\{x \in \mathcal{A}_j\}}{P(\mathcal{A}_j) + r\big(\check{\mathbb{P}}(\mathcal{A}_j) - P(\mathcal{A}_j)\big)} (\check{p}_t(x)-p_t(x)) \\
 &\quad\quad- \frac{ \mathds{1}\{x \in \mathcal{A}_j\}\big(\check{\mathbb{P}}(\mathcal{A}_j) - P(\mathcal{A}_j)\big)}{\big(P(\mathcal{A}_j) + r\big(\check{\mathbb{P}}(\mathcal{A}_j) - P(\mathcal{A}_j)\big)\big)^2} \big(p_t(x)+r(\check{p}_t(x)-p_t(x))-\alpha_t\big)\\
 &\quad\quad+\frac{ \mathds{1}\{x \in \mathcal{A}_j\}}{P(\mathcal{A}_j) + r\big(\check{\mathbb{P}}(\mathcal{A}_j) - P(\mathcal{A}_j)\big)}\cdot\Big( y-p_t(x)-r(\check{p}_t(x) - p_t(x))\Big)\cdot\\
 &\quad\quad\cdot\Big[-\frac{\mathds{1}\{t = 1\} \big(\check{e}(x)-e(x)\big)}{\big(e(x) + r\big(\check{e}(x)-e(x)\big)\big)^2} + \frac{\mathds{1}\{t = 0\} \big(\check{e}(x)-e(x)\big)}{\big(1 - e(x)-r\big(\check{e}(x)-e(x)\big)\big)^2}\Big]\\
 &\quad\quad+ \Big[\frac{ -\mathds{1}\{x \in \mathcal{A}_j\}(\check{p}_t(x) - p_t(x))}{P(\mathcal{A}_j) + r\big(\check{\mathbb{P}}(\mathcal{A}_j) - P(\mathcal{A}_j)\big)}- \frac{ \mathds{1}\{x \in \mathcal{A}_j\}\big(\check{\mathbb{P}}(\mathcal{A}_j) - P(\mathcal{A}_j)\big)}{\big(P(\mathcal{A}_j) + r\big(\check{\mathbb{P}}(\mathcal{A}_j) - P(\mathcal{A}_j)\big)\big)^2}\cdot\Big( y-p_t(x)-r(\check{p}_t(x) - p_t(x))\Big)\Big]\\
&\quad\quad\cdot\Big[\frac{\mathds{1}\{t = 1\} }{e(x) + r\big(\check{e}(x)-e(x)\big)} + \frac{\mathds{1}\{t = 0\} }{1 - e(x)-r\big(\check{e}(x)-e(x)\big)}\Big]\Bigg].
\end{align*}
Set $r=0$,
\begin{align*}
   &= E\Bigg[ \frac{ \mathds{1}\{x \in \mathcal{A}_j\}}{P(\mathcal{A}_j)} (\check{p}_t(x)-p_t(x)) - \frac{ \mathds{1}\{x \in \mathcal{A}_j\}\big(\check{\mathbb{P}}(\mathcal{A}_j) - P(\mathcal{A}_j)\big)}{\big(P(\mathcal{A}_j)\big)^2} \big(p_t(x)-\alpha_t\big)\\
 &\quad+\frac{ \mathds{1}\{x \in \mathcal{A}_j\}}{P(\mathcal{A}_j) }\Big( y-p_t(x)\Big)\Big[-\frac{\mathds{1}\{t = 1\} \big(\check{e}(x)-e(x)\big)}{\big(e(x)\big)^2} + \frac{\mathds{1}\{t = 0\} \big(\check{e}(x)-e(x)\big)}{\big(1 - e(x)\big)^2}\Big]\\
 &\quad+ \Big[\frac{ -\mathds{1}\{x \in \mathcal{A}_j\}(\check{p}_t(x) - p_t(x))}{P(\mathcal{A}_j)}- \frac{ \mathds{1}\{x \in \mathcal{A}_j\}\big(\check{\mathbb{P}}(\mathcal{A}_j) - P(\mathcal{A}_j)\big)}{\big(P(\mathcal{A}_j)\big)^2}\cdot\Big( y-p_t(x)\Big)\Big]\cdot\Big[\frac{\mathds{1}\{t = 1\} }{e(x)} + \frac{\mathds{1}\{t = 0\} }{1 - e(x)}\Big]\Bigg].
\end{align*}
Given that
\begin{align*}
    &E[p_t(x)-\alpha_t|x] = 0,\quad E[t=1|X] = e(x), \\
    &E[t(y-p_t(x))|x] = 0, \quad E[(1-t)(y-p_t(x))|x] = 0,
\end{align*}
the first-order Gateaux derivative for each subgroup $j$ is $0$. Thus $f'(0) = 0$ for all the subgroups.

\noindent\textbf{Step 2.} The second order remainder term satisfies 
\begin{align*}
||f''(\tilde{r})/2|| \leq \sup_{r\in (0,1)} || f''(r)/2 ||.
\end{align*}
For each subgroup $j$, 
\begin{align*}
    &\partial^2_r f_j(r) \\
    &=E\Bigg[ -\frac{ \mathds{1}\{x \in \mathcal{A}_j\}(\check{p}_t(x)-p_t(x))\big(\check{\mathbb{P}}(\mathcal{A}_j)- P(\mathcal{A}_j)\big)}{\big(P(\mathcal{A}_j) + r\big(\check{\mathbb{P}}(\mathcal{A}_j)- P(\mathcal{A}_j)\big)\big)^2}- \frac{ \mathds{1}\{x \in \mathcal{A}_j\}\big(\check{\mathbb{P}}(\mathcal{A}_j) - P(\mathcal{A}_j)\big)}{\big(P(\mathcal{A}_j) + r\big(\check{\mathbb{P}}(\mathcal{A}_j) - P(\mathcal{A}_j)\big)\big)^2} \big(\check{p}_t(x)-p_t(x)\big)\\
     &\quad\quad+ 2\frac{ \mathds{1}\{x \in \mathcal{A}_j\}\big(\check{\mathbb{P}}(\mathcal{A}_j) - P(\mathcal{A}_j)\big)^2}{\big(P(\mathcal{A}_j) + r\big(\check{\mathbb{P}}(\mathcal{A}_j) - P(\mathcal{A}_j)\big)\big)^3} \big(p_t(x)+r(\check{p}_t(x)-p_t(x))-\alpha_t\big)\\
 &\quad\quad+\frac{ \mathds{1}\{x \in \mathcal{A}_j\}}{P(\mathcal{A}_j) + r\big(\check{\mathbb{P}}(\mathcal{A}_j) - P(\mathcal{A}_j)\big)}\cdot\Big( y-p_t(x)-r(\check{p}_t(x) - p_t(x))\Big)\cdot\\
 &\quad\quad\quad\cdot2\Big[\frac{\mathds{1}\{t = 1\} \big(\check{e}(x)-e(x)\big)^2}{\big(e(x) + r\big(\check{e}(x)-e(x)\big)\big)^3} - \frac{\mathds{1}\{t = 0\} \big(\check{e}(x)-e(x)\big)^2}{\big(1 - e(x)-r\big(\check{e}(x)-e(x)\big)\big)^3}\Big]\\
&\quad\quad+\Big[\frac{ -\mathds{1}\{x \in \mathcal{A}_j\}(\check{p}_t(x) - p_t(x))}{P(\mathcal{A}_j) + r\big(\check{\mathbb{P}}(\mathcal{A}_j) - P(\mathcal{A}_j)\big)}-\frac{ \mathds{1}\{x \in \mathcal{A}_j\}\big(\check{\mathbb{P}}(\mathcal{A}_j) - P(\mathcal{A}_j)\big)}{\big(P(\mathcal{A}_j) + r\big(\check{\mathbb{P}}(\mathcal{A}_j) - P(\mathcal{A}_j)\big)\big)^2}\cdot\Big( y-p_t(x)-r(\check{p}_t(x) - p_t(x))\Big)\Big]\\
 &\quad\quad\quad\cdot\Big[-\frac{\mathds{1}\{t = 1\} \big(\check{e}(x)-e(x)\big)}{\big(e(x) + r\big(\check{e}(x)-e(x)\big)\big)^2} + \frac{\mathds{1}\{t = 0\} \big(\check{e}(x)-e(x)\big)}{\big(1 - e(x)-r\big(\check{e}(x)-e(x)\big)\big)^2}\Big]\\
 &\quad\quad+ \Big[\frac{ -\mathds{1}\{x \in \mathcal{A}_j\}(\check{p}_t(x) - p_t(x))}{P(\mathcal{A}_j) + r\big(\check{\mathbb{P}}(\mathcal{A}_j) - P(\mathcal{A}_j)\big)}- \frac{ \mathds{1}\{x \in \mathcal{A}_j\}\big(\check{\mathbb{P}}(\mathcal{A}_j) - P(\mathcal{A}_j)\big)}{\big(P(\mathcal{A}_j) + r\big(\check{\mathbb{P}}(\mathcal{A}_j) - P(\mathcal{A}_j)\big)\big)^2}\cdot\Big( y-p_t(x)-r(\check{p}_t(x) - p_t(x))\Big)\Big]\\
&\quad\quad\quad\cdot\Big[\frac{-\mathds{1}\{t = 1\}\big(\check{e}(x)-e(x)\big) }{\big(e(x) + r\big(\check{e}(x)-e(x)\big)\big)^2} + \frac{\mathds{1}\{t = 0\} \big(\check{e}(x)-e(x)\big)}{\big(1 - e(x)-r\big(\check{e}(x)-e(x)\big)\big)^2}\Big].\\
 &\quad\quad+ \Big[\frac{ \mathds{1}\{x \in \mathcal{A}_j\}(\check{p}_t(x) - p_t(x))\big(\check{\mathbb{P}}(\mathcal{A}_j) - P(\mathcal{A}_j)\big)}{\big(P(\mathcal{A}_j) + r\big(\check{\mathbb{P}}(\mathcal{A}_j) - P(\mathcal{A}_j)\big)\big)^2}+ \frac{ \mathds{1}\{x \in \mathcal{A}_j\}\big(\check{\mathbb{P}}(\mathcal{A}_j) - P(\mathcal{A}_j)\big)\cdot(\check{p}_t(x) - p_t(x))}{\big(P(\mathcal{A}_j) + r\big(\check{\mathbb{P}}(\mathcal{A}_j) - P(\mathcal{A}_j)\big)\big)^2}\\
 &\quad\quad+ 2\frac{ \mathds{1}\{x \in \mathcal{A}_j\}\big(\check{\mathbb{P}}(\mathcal{A}_j) - P(\mathcal{A}_j)\big)^2}{\big(P(\mathcal{A}_j) + r\big(\check{\mathbb{P}}(\mathcal{A}_j) - P(\mathcal{A}_j)\big)\big)^3}\cdot\Big( y-p_t(x)-r(\check{p}_t(x) - p_t(x))\Big)\Big]\\
 &\quad\quad\quad\cdot\Big[\frac{\mathds{1}\{t = 1\} }{e(x) + r\big(\check{e}(x)-e(x)\big)} + \frac{\mathds{1}\{t = 0\} }{1 - e(x)-r\big(\check{e}(x)-e(x)\big)}\Big]\Bigg].
\end{align*}
Simplify the second-order Gateaux derivative we have
\begin{align*}
    \partial^2_r f(r) &=E\Bigg[ -2\frac{ \mathds{1}\{x \in \mathcal{A}_j\}\big(\check{\mathbb{P}}(\mathcal{A}_j)- P(\mathcal{A}_j)\big)}{\big(P(\mathcal{A}_j) + r\big(\check{\mathbb{P}}(\mathcal{A}_j)- P(\mathcal{A}_j)\big)\big)^2} (\check{p}_t(x)-p_t(x))\Bigg]\\
     &+ E\Bigg[2\frac{ \mathds{1}\{x \in \mathcal{A}_j\}\big(\check{\mathbb{P}}(\mathcal{A}_j) - P(\mathcal{A}_j)\big)^2}{\big(P(\mathcal{A}_j) + r\big(\check{\mathbb{P}}(\mathcal{A}_j) - P(\mathcal{A}_j)\big)\big)^3} \big(p_t(x)+r(\check{p}_t(x)-p_t(x))-\alpha_t\big)\Bigg]\\
 &+E\Bigg[\frac{ \mathds{1}\{x \in \mathcal{A}_j\}}{P(\mathcal{A}_j) + r\big(\check{\mathbb{P}}(\mathcal{A}_j) - P(\mathcal{A}_j)\big)}\cdot\Big( y-p_t(x)-r(\check{p}_t(x) - p_t(x))\Big)\cdot\\
 &\quad\cdot2\Big[\frac{\mathds{1}\{t = 1\} \big(\check{e}(x)-e(x)\big)^2}{\big(e(x) + r\big(\check{e}(x)-e(x)\big)\big)^3} - \frac{\mathds{1}\{t = 0\} \big(\check{e}(x)-e(x)\big)^2}{\big(1 - e(x)-r\big(\check{e}(x)-e(x)\big)\big)^3}\Big]\Bigg]\\
&+E\Bigg[2\Big[\frac{ -\mathds{1}\{x \in \mathcal{A}_j\}(\check{p}_t(x) - p_t(x))}{P(\mathcal{A}_j) + r\big(\check{\mathbb{P}}(\mathcal{A}_j) - P(\mathcal{A}_j)\big)}-\frac{ \mathds{1}\{x \in \mathcal{A}_j\}\big(\check{\mathbb{P}}(\mathcal{A}_j) - P(\mathcal{A}_j)\big)}{\big(P(\mathcal{A}_j) + r\big(\check{\mathbb{P}}(\mathcal{A}_j) - P(\mathcal{A}_j)\big)\big)^2}\\
&\cdot\Big( y-p_t(x)-r(\check{p}_t(x) - p_t(x))\Big)\Big]\\
 &\quad\cdot\Big[-\frac{\mathds{1}\{t = 1\} \big(\check{e}(x)-e(x)\big)}{\big(e(x) + r\big(\check{e}(x)-e(x)\big)\big)^2} + \frac{\mathds{1}\{t = 0\} \big(\check{e}(x)-e(x)\big)}{\big(1 - e(x)-r\big(\check{e}(x)-e(x)\big)\big)^2}\Big]\Bigg]\\
 &+E\Bigg[\Big[ \frac{ 2\mathds{1}\{x \in \mathcal{A}_j\}\big(\check{\mathbb{P}}(\mathcal{A}_j) - P(\mathcal{A}_j)\big)}{\big(P(\mathcal{A}_j) + r\big(\check{\mathbb{P}}(\mathcal{A}_j) - P(\mathcal{A}_j)\big)\big)^2}(\check{p}_t(x) - p_t(x))\\
 &\quad\quad+  \frac{ 2\mathds{1}\{x \in \mathcal{A}_j\}\big(\check{\mathbb{P}}(\mathcal{A}_j) - P(\mathcal{A}_j)\big)^2}{\big(P(\mathcal{A}_j) + r\big(\check{\mathbb{P}}(\mathcal{A}_j) - P(\mathcal{A}_j)\big)\big)^3}\cdot\Big( y-p_t(x)-r(\check{p}_t(x) - p_t(x))\Big)\Big]\\
&\quad\cdot\Big[\frac{\mathds{1}\{t = 1\} }{e(x) + r\big(\check{e}(x)-e(x)\big)} + \frac{\mathds{1}\{t = 0\} }{1 - e(x)-r\big(\check{e}(x)-e(x)\big)}\Big]\Bigg]. 
\end{align*}

Set $r = \tilde{r}$, 
by Assumption 4.3, for some constants $C$ and $\xi_n$,
\begin{align*}
&|f''_j(\tilde{r})| \leq C ||\check{e}(x)-e(x)||_2  ||\check{p}_t(x)-p_t(x)||_2  \leq \delta_n n^{-1/2},\quad j = 1,\ldots, J,\\
&||f''(\tilde{r})||_{\infty} \leq \sup_{r\in[0,1)}|| f''(r)||_{\infty} = o_P(1/\sqrt{n}).
\end{align*}

In sum, 
\begin{align*}
f(1) &= f'(0) + f''(\tilde{r})/2 = o_P(1/\sqrt{n}), \quad r\in[0,1).
\end{align*}

\subsection{Derivation of the expansion (\ref{eq:expansion-empirical-process})}\label{Sec:empirical-process-remainder}

It holds trivially true that our estimator satisfies 
\begin{align*}
 \bm{\hat{\alpha}}_t - \bm{\alpha}_t =  &\ - \int \bm{\psi}^{\text(1)}\big(o; \bm{\hat{\eta}} \big) \mbox{d}P + \Big(   \bm{\hat{\alpha}}_t - \bm{\alpha}_t + \int \bm{\psi}^{\text(1)}\big(o; \bm{\hat{\eta}} \big) \mbox{d}P \Big)\\
 := &\ \int \bm{\psi}^{\text(1)}\big(o;  \bm{\hat{\eta}} \big) \mbox{d}\mathbb{P}_n - \int \bm{\psi}^{\text(1)}\big(o; \bm{\hat{\eta}} \big) \mbox{d}P + R_{n1}\\
 =&\ \mathbb{P}_n \bm{\psi}^{\text(1)}(o; \bm{\hat{\eta}} ) - P\bm{\psi}^{\text(1)}(o; \bm{\hat{\eta}} ) + R_{n1}\\
 =&\  (\mathbb{P}_n-P) \bm{\psi}^{\text(1)}(o; \bm{\eta}_0 ) + ( \mathbb{P}_n  - P) \big( \bm{\psi}^{\text(1)}(o; \bm{\hat{\eta}} )  - \bm{\psi}^{\text(1)}(o;\bm{\eta}_0 )  \big) + R_{n1}\\
 =& \ (\mathbb{P}_n-P) \bm{\psi}^{\text{(1)}}(o; \bm{\eta}_0 ) + R_{n2}+ R_{n1},
\end{align*}
where the second equality is guaranteed by our proposed procedure, and 
\begin{align*}
R_{n1} = &\  \bm{\hat{\alpha}}_t - \bm{\alpha}_t + \int \bm{\psi}^{\text(1)}\big(o;  \bm{\hat{\eta}} \big) \mbox{d}P\\
=& \ \mathbb{E}_n \big[ \hat{p}_t(X) \big] - E_X\big[p_t(X)\big] + E_{Y,T,X}\left[ \frac{ \mathds{1}\{X \in \mathcal{A}_j\}}{\hat{\mathbb{P}}(\mathcal{A}_j)}\Big(\frac{\mathds{1}\{T = 1\} }{\hat{e}(X)} + \frac{\mathds{1}\{T = 0\} }{1 - \hat{e}(X)}\Big)\big(Y-\hat{p}_T(X)\big) \right]\\
=& \mathbb{E}_n \big[ \hat{p}_t(X) \big] - E_X\big[\hat{p}_t(X)\big] + E_X\big[\hat{p}_t(X)\big] - E_X\big[p_t(X)\big] \\
& +E_{Y,T,X}\left[ \frac{ \mathds{1}\{X \in \mathcal{A}_j\}}{\hat{\mathbb{P}}(\mathcal{A}_j)}\Big(\frac{\mathds{1}\{T = 1\} }{\hat{e}(X)} + \frac{\mathds{1}\{T = 0\} }{1 - \hat{e}(X)}\Big)\big(Y-\hat{p}_T(X)\big) \right]\\
=& (\mathbb{P}_n - P) \hat{p}_t + P \bm{\psi}\big(o; \bm{\alpha}_t, \bm{\hat{\eta}}\big)\\
=& (\mathbb{P}_n - P) (\hat{p}_t  - p_t) + (\mathbb{P}_n - P) p_t +P \bm{\psi}\big(o; \bm{\alpha}_t, \bm{\hat{\eta}}\big)\\
=& (\mathbb{P}_n-P) \bm{\psi}^{\text{(2)}}(o;\bm{\alpha},\bm{\eta}_0 )  + (\mathbb{P}_n - P) (\hat{p}_t  - p_t) + (\mathbb{P}_n - P) p_t +P \bm{\psi}\big(o; \bm{\alpha}_t, \bm{\hat{\eta}}\big).
\end{align*}
To this end, we have the final expansion: 
\begin{align*}
\bm{\hat{\alpha}}_t - \bm{\alpha}_t =  (\mathbb{P}_n-P) \bm{\psi}(o;\bm{\alpha} ,\bm{\eta}_0 ) + R_{n2} + R_{n3} + R_{n4}, 
\end{align*}
where 
\begin{align*}
R_{n2} & = ( \mathbb{P}_n  - P)\big( \bm{\psi}^{\text(1)}(o; \bm{\hat{\eta}} )  - \bm{\psi}^{\text(1)}(o;\bm{\eta}_0 )  \big) ,\\
R_{n3} & =  (\mathbb{P}_n - P) (\hat{p}_t  - p_t) \\
R_{n4} &= P \bm{\psi}\big(o; \bm{\alpha}_t, \bm{\hat{\eta}}\big).
\end{align*}



\subsection{Proof of semiparametric efficiency results with delta method}\label{appendix:rr-or-semiparametric}

\cite{vaart2000asymptotic} shows that the asymptotic variance obtained by the delta method indeed achieves the semiparametric efficiency bound (Chapter 25.7). Since the results in \cite{vaart2000asymptotic} are formulated using notation different from our manuscript, in what follows, 
we provide a justification to show why the asymptotic variances of  $\hat{\alpha}_{\text{RR}}$ and $\hat{\alpha}_{\text{OR}}$ attain the semiparametric efficiency
 In what follows, 
we justify why the asymptotic variances of  $\hat{\alpha}_{ARD}$, $\hat{\alpha}_{RR}$, and $\hat{\alpha}_{OR}$ attain the semiparametric efficiency
bound. For simplicity, we shall not discuss the subgroup case, but the justifications below can be easily extended to subgroup $\hat{\alpha}_{ARD}$, $\hat{\alpha}_{RR}$, and $\hat{\alpha}_{OR}$. 

First, denote the joint density of (Y,T,X) as
\begin{align*}
    f(y,t,x) = \big(f_1(y|x)e(x)\big)^t \big(f_0(y|x)(1-e(x))\big)^{1-t} f(x),
\end{align*}
where $e(x):= f(t|x)$ denotes the propensity score,  $f_1(y|x)$ denotes the conditional density of $y$ given $x$ under treatment, and $f_0(y|x)$ denotes the conditional density of $y$ given $x$ under control.
Assume the parametric submodel indexed by parameter $\theta$ is
\begin{align*}
  f(y,t,x;\theta) = \big(f_1(y|x;\theta)e(x;\theta)\big)^t \big(f_0(y|x;\theta)(1-e(x;\theta))\big)^{1-t} f(x;\theta), 
\end{align*}
where $\theta$ is indexed a finite-dimensional parameter $\beta$ and an infinite-dimensional parameter $\eta$, i.e. $\theta = (\beta, \eta)$. The score function of the above parametric submodel is
\begin{align}\label{eq:score}
    s(y,t,x|\theta) = t s_1(y|x;\theta) + (1-t)s_0(y|x;\theta) + \frac{t-e(x;\theta)}{e(x;\theta)(1-e(x;\theta))}e'(x;\theta) + p(x;\theta),
\end{align}
where $s_1(y|x;\theta)= \frac{d}{d\theta}\log f_1(Y|X;\theta)$, $s_0(y|x;\theta)= \frac{d}{d\theta}\log f_0(Y|X;\theta)$, $p(x;\theta) = \frac{d}{d\theta}\log f(X;\theta)$, and $e'(x;\theta) = \frac{d}{d\theta}e(x;\theta)$. From the score function Eq \eqref{eq:score} we obtain the tangent space $\mathcal{T}$ spanned by the score function $t s_1(y|x) + (1-t)s_0(y|x) + \frac{t-e(x)}{e(x)(1-e(x))}e'(x) + p(x)$.
Now assume $\theta$ can be parametrized by $\beta_1$ and $\beta_0$, that is
\begin{align*}
    \beta_1(\theta) &= \int\int yf_1(y|x;\theta)f(x;\theta)dy dx,\\
    \beta_0(\theta) &= \int\int yf_0(y|x;\theta)f(x;\theta)dy dx.
\end{align*}
$\beta_1(\theta)$ and $\beta_0(\theta)$ represent the conditional mean of outcome under treatment and control, respectively. The pathwise derivatives of $\beta_1(\theta)$ and $\beta_0(\theta)$ are
\begin{align*}
     \frac{\partial\beta_1(\theta_0)}{\partial\theta} &= \int\int ys_1(y|x;\theta_0)f_1(y|x)f(x)dy dx + \int\int\beta_1(x)p(x;\theta_0)f(x)dx,\\
     \frac{\partial\beta_0(\theta_0)}{\partial\theta} &= \int\int ys_0(y|x;\theta_0)f_0(y|x)f(x)dy dx + \int\int\beta_0(x)p(x;\theta_0)f(x)dx,
\end{align*}
where $\theta_0$ denotes the true parameter. If the target parameter of interest is the absolute risk difference (ARD),
\begin{align*}
    \frac{\partial \beta^{\text{ARD}}(\theta_0)}{\partial\theta} &= \frac{\partial \beta^{\text{ARD}}(\theta_0)}{\partial\big(\beta_1(\theta_0),\beta_0(\theta_0)\big)} \frac{\partial\big(\beta_1(\theta_0),\beta_0(\theta_0)\big)}{\partial\theta}\\
    &= (1 \ -1)
    \begin{pmatrix}
    \frac{\partial\beta_1(\theta_0)}{\partial\theta}\\
    \frac{\partial\beta_0(\theta_0)}{\partial\theta}
    \end{pmatrix} = \frac{\partial\beta_1(\theta_0)}{\partial\theta}- \frac{\partial\beta_0(\theta_0)}{\partial\theta},\\
    &= \int\int ys_1(y|x;\theta_0)f_1(y|x)f(x)dy dx + \int\int\beta_1(x)p(x;\theta_0)f(x)dx\\
    &\quad-\int\int ys_0(y|x;\theta_0)f_0(y|x)f(x)dy dx - \int\int\beta_0(x)p(x;\theta_0)f(x)dx.
\end{align*}
Similarly, if the target parameter of interest is the relative risk,
\begin{align*}
    \frac{\partial \beta^{\text{RR}}(\theta_0)}{\partial\theta} &= \frac{\partial \beta^{\text{RR}}(\theta_0)}{\partial\big(\beta_1(\theta_0),\beta_0(\theta_0)\big)} \frac{\partial\big(\beta_1(\theta_0),\beta_0(\theta_0)\big)}{\partial\theta}\\
    &=\begin{pmatrix}
    \frac{1}{\beta_0(\theta_0)} & \frac{-\beta_1(\theta_0)}{\beta^2_0(\theta_0)}
    \end{pmatrix}
    \begin{pmatrix}
    \frac{\partial\beta_1(\theta_0)}{\partial\theta}\\
    \frac{\partial\beta_0(\theta_0)}{\partial\theta}
    \end{pmatrix}= \frac{1}{\beta_0(\theta_0)}\frac{\partial\beta_1(\theta_0)}{\partial\theta}-\frac{\beta_1(\theta_0)}{\beta_0^2(\theta_0)} \frac{\partial\beta_0(\theta_0)}{\partial\theta},\\
    &=\frac{1}{\beta_0(\theta_0)} \Big(\int\int ys_1(y|x;\theta_0)f_1(y|x)f(x)dy dx + \int\int\beta_1(x)p(x;\theta_0)f(x)dx\Big)\\
    &\quad-\frac{\beta_1(\theta_0)}{\beta_0^2(\theta_0)}\Big(\int\int ys_0(y|x;\theta_0)f_0(y|x)f(x)dy dx + \int\int\beta_0(x)p(x;\theta_0)f(x)dx\Big).
\end{align*}

If the target parameter of interest is the odds ratio,
\begin{align*}
    \frac{\partial \beta^{\text{OR}}(\theta_0)}{\partial\theta} &= \frac{\partial \beta^{\text{OR}}(\theta_0)}{\partial\big(\beta_1(\theta_0),\beta_0(\theta_0)\big)} \frac{\partial\big(\beta_1(\theta_0),\beta_0(\theta_0)\big)}{\partial\theta}\\
    &=\begin{pmatrix}
    \frac{1-\beta_0(\theta_0)}{\beta_0(\theta_0)(1-\beta_1(\theta_0))^2} & \frac{-\beta_1(\theta_0)}{\beta^2_0(\theta_0)(1-\beta_1(\theta_0))}
    \end{pmatrix}
    \begin{pmatrix}
    \frac{\partial\beta_1(\theta_0)}{\partial\theta}\\
    \frac{\partial\beta_0(\theta_0)}{\partial\theta}
    \end{pmatrix},\\
    &= \frac{1-\beta_0(\theta_0)}{\beta_0(\theta_0)(1-\beta_1(\theta_0))^2}\frac{\partial\beta_1(\theta_0)}{\partial\theta}-\frac{\beta_1(\theta_0)}{\beta_0^2(\theta_0)(1-\beta_1(\theta_0))} \frac{\partial\beta_0(\theta_0)}{\partial\theta},\\
    &=\frac{1-\beta_0(\theta_0)}{\beta_0(\theta_0)(1-\beta_1(\theta_0))^2} \Big(\int\int ys_1(y|x;\theta_0)f_1(y|x)f(x)dy dx + \int\int\beta_1(x)p(x;\theta_0)f(x)dx\Big)\\
    &\quad-\frac{\beta_1(\theta_0)}{\beta_0^2(\theta_0)(1-\beta_1(\theta_0))}\Big(\int\int ys_0(y|x;\theta_0)f_0(y|x)f(x)dy dx + \int\int\beta_0(x)p(x;\theta_0)f(x)dx\Big).
\end{align*}

Now let
\begin{align*}
    \varphi^{\text{ARD}}(Y,T,X) &= \frac{T}{e(X)}\big(Y-\beta_1(X)\big)-\frac{1-T}{1-e(X)}\big(Y-\beta_0(X)\big)+\beta_1(X)-\beta_0(X)-(\beta_1-\beta_0),\\
      \varphi^{\text{RR}}(Y,T,X)&= \frac{1}{\beta_0}\Big(\frac{T}{e(X)}\big(Y-\beta_1(X)\big)+\beta_1(X)-\beta_1\Big) - \frac{\beta_1}{\beta_0^2}\Big(\frac{1-T}{1-e(X)}\big(Y-\beta_0(X)\big)+\beta_0(X)-\beta_0\Big),\\
      \varphi^{\text{OR}}(Y,T,X)&= \frac{1-\beta_0}{\beta_0(1-\beta_1)^2}\Big(\frac{T}{e(X)}\big(Y-\beta_1(X)\big)+\beta_1(X)-\beta_1\Big)\\
      &\quad - \frac{\beta_1}{\beta_0^2(1-\beta_1)}\Big(\frac{1-T}{1-e(X)}\big(Y-\beta_0(X)\big)+\beta_0(X)-\beta_0\Big).
\end{align*}
Taking the product of $\varphi(Y,T,X)$ and the score function in Eq \eqref{eq:score}, we observe that
\begin{align*}
    \frac{\partial\beta^{\text{ARD}}(\theta_0)}{\partial\theta_0} &= \mathbb{E}[\varphi^{\text{ARD}}(Y,T,X)\cdot s(Y,T,X|\theta_0)],\\
    \frac{\partial\beta^{\text{RR}}(\theta_0)}{\partial\theta_0} &= \mathbb{E}[\varphi^{\text{RR}}(Y,T,X)\cdot s(Y,T,X|\theta_0)],\\
    \frac{\partial\beta^{\text{OR}}(\theta_0)}{\partial\theta_0} &= \mathbb{E}[\varphi^{\text{OR}}(Y,T,X)\cdot s(Y,T,X|\theta_0)].
\end{align*}
The above derivations suggest that $\varphi^{\text{ARD}}(Y,T,X)$,  $\varphi^{\text{RR}}(Y,T,X)$, and $\varphi^{\text{OR}}(Y,T,X) \in \mathcal{T}$, and thus the semiparametric efficiency bounds of ARD, RR, and OR can be computed as $\mathbb{E}\big[\big(\varphi^{\text{ARD}}(Y,T,X)\big)^2\big]$, $\mathbb{E}\big[\big(\varphi^{\text{RR}}(Y,T,X)\big)^2\big]$, and $\mathbb{E}\big[\big(\varphi^{\text{OR}}(Y,T,X)\big)^2\big]$, respectively.

The above derivations suggest that the semiparametric efficiency results in Theorem 1 can be generalized to various estimators of interest.

\section{Proof of Lemma \ref{lemma:lagrangian}}\label{appendix:duality}

To simplify notation, suppose we work with a single update and denote the initial estimate as $\hat{p}_t^{\text{Init}}(\cdot)$.
The primal optimization problem is:
\begin{align*}
\underset{ \bm{\varepsilon}\in\mathbb{R}^J }{\min} \ &\ -\frac{1}{n}\sum_{i=1}^n   \Big[ Y_i \Big( \text{logit}\big(\hat{p}^{\text{Init}}_t(X_i)\big)+ \sum_{j=1}^J \varepsilon_j \cdot  \frac{\mathds{1}(X_i\in\mathcal{A}_j)}{\mathds{P}(\mathcal{A}_j)} \frac{\mathds{1}(T_i=t)}{\hat{e}(X_i)}\Big)\\
	 & \qquad \qquad - \log \Big( 1 + \text{exp}^{ \text{logit}\big(\hat{p}^{\text{Init}}_t(X_i)\big)+ \sum_{j=1}^J\varepsilon_j \cdot \frac{\mathds{1}(X_i\in\mathcal{A}_j)}{\mathds{P}(\mathcal{A}_j)} \frac{\mathds{1}(T_i=t)}{\hat{e}(X_i)} } \Big) \Big],\\
\text{s.t.} \ &\ ||\bm{\varepsilon}||_2 - \delta \leq 0 . 
\end{align*}
The Lagrangian associated with the primal problem is 
\begin{align*}
 L(\bm{\varepsilon}, \lambda)= -\frac{1}{n}\sum_{i=1}^n  & \Big[ Y_i \Big( \text{logit}\big(\hat{p}^{\text{Init}}_t(X_i)\big)+ \sum_{j=1}^J \varepsilon_j \cdot  \frac{\mathds{1}(X_i\in\mathcal{A}_j)}{\mathds{P}(\mathcal{A}_j)} \frac{\mathds{1}(T_i=t)}{\hat{e}(X_i)}\Big)\\
	 &- \log \Big( 1 + \text{exp}^{ \text{logit}\big(\hat{p}^{\text{Init}}_t(X_i)\big)+ \sum_{j=1}^J\varepsilon_j \cdot \frac{\mathds{1}(X_i\in\mathcal{A}_j)}{\mathds{P}(\mathcal{A}_j)} \frac{\mathds{1}(T_i=t)}{\hat{e}(X_i)} } \Big) \Big] + \lambda (||\bm{\varepsilon}||_2 - \delta),
\end{align*}
where $\lambda\geq 0$ is the Lagrange multiplier. The Lagrangian primal problem is defined as 
\begin{align*}
\min_{\bm{\varepsilon}}  L_{\texttt{primal}}(\bm{\varepsilon}, \lambda) = \min_{\bm{\varepsilon}} \max_{\lambda \geq 0}  L(\bm{\varepsilon}, \lambda).
\end{align*}
The Lagrangian dual function is thus defined as
\begin{align*}
 L_{\texttt{dual}}(\lambda) = \min_{\bm{\varepsilon}} L(\bm{\varepsilon}, \lambda), 
\end{align*}
and the Lagrangian dual problem is 
\begin{align*}
\max_{\lambda \geq 0} L_{\texttt{dual}}( \lambda) = \max_{\lambda \geq 0} \min_{\bm{\varepsilon}} L(\bm{\varepsilon}, \lambda).
\end{align*}
Given our optimization problem is a convex problem, there is no duality gap between the primal and dual problems. 
Next, we solve for $\hat{\varepsilon}_j$ by taking derivative of the dual function.
\begin{align*}
  \frac{\partial L_{\texttt{dual}}(\hat{\bm{\varepsilon}},\lambda)}{\partial \varepsilon_j}&=
  -\frac{1}{n}\sum_{i:T_i=t}\Big\{ Y_i \frac{\mathds{1}(X_i\in\mathcal{A}_j)}{\mathds{P}(\mathcal{A}_j)} \frac{\mathds{1}(T_i=t)}{\hat{e}(X_i)} - \frac{\text{exp}\Big(\text{logit}\big(\hat{p}^{\text{Init}}_t(X_i)\big) + \sum_{j=1}^J\hat{\varepsilon}_j\cdot \frac{\mathds{1}(X_i\in\mathcal{A}_j)}{\mathds{P}(\mathcal{A}_j)} \frac{\mathds{1}(T_i=t)}{\hat{e}(X_i)}\Big)}{1+\text{exp}\Big(\text{logit}\big(\hat{p}^{\text{Init}}_t(X_i)\big) + \sum_{j=1}^J\hat{\varepsilon}_j\cdot \frac{\mathds{1}(X_i\in\mathcal{A}_j)}{\mathds{P}(\mathcal{A}_j)} \frac{\mathds{1}(T_i=t)}{\hat{e}(X_i)}\Big)}\\
  &\cdot \frac{\mathds{1}(X_i\in\mathcal{A}_j)}{\mathds{P}(\mathcal{A}_j)} \frac{\mathds{1}(T_i=t)}{\hat{e}(X_i)}\Big\} + \lambda\varepsilon_j ,\\
&= -\frac{1}{n}\sum_{i:T_i=t}\Big\{ \Big[Y_i -\frac{\text{exp}\Big(\text{logit}\big(\hat{p}^{\text{Init}}_t(X_i)\big) + \sum_{j=1}^J\hat{\varepsilon}_j \cdot \frac{\mathds{1}(X_i\in\mathcal{A}_j)}{\mathds{P}(\mathcal{A}_j)} \frac{\mathds{1}(T_i=t)}{\hat{e}(X_i)}\Big)}{1+\text{exp}\Big(\text{logit}\big(\hat{p}^{\text{Init}}_t(X_i)\big)+ \sum_{j=1}^J\hat{\varepsilon}_j\cdot \frac{\mathds{1}(X_i\in\mathcal{A}_j)}{\mathds{P}(\mathcal{A}_j)} \frac{\mathds{1}(T_i=t) }{\hat{e}(X_i)}\Big)}\Big]\\
&\cdot \frac{\mathds{1}(X_i\in\mathcal{A}_j)}{\mathds{P}(\mathcal{A}_j)} \frac{\mathds{1}(T_i=t)}{\hat{e}(X_i)}\Big\}+ \lambda\varepsilon_j,\\
  &= -\frac{1}{n}\sum_{i:T_i=t} \Big\{\Big(Y_i - \hat{p}_t(X_i)\Big)\cdot  \frac{\mathds{1}(X_i\in\mathcal{A}_j)}{\mathds{P}(\mathcal{A}_j)} \frac{1}{\hat{e}(X_i)}\Big\} + \lambda \varepsilon_j,\\
  \frac{\partial L_{\texttt{dual}}(\hat{\bm{\varepsilon}},\lambda)}{\partial \lambda}  &=  -\hat{\phi}_j(X_i) + \lambda\varepsilon_j= 0,\\
  &\implies \hat{\phi}_j(X_i) = \lambda\varepsilon_j, \quad \varepsilon_j = \frac{\hat{\phi}_j(X_i)}{\lambda}.
\end{align*}

By plugging in $ \varepsilon_j = \frac{\hat{\phi}_j(X_i)}{\lambda}$, we have the  Lagrangian dual function equals to 
\begin{align*}
 L_{\texttt{dual}}( \lambda) = -\frac{1}{n}\sum_{i:T_i=t}  & \Big[ Y_i \Big( \text{logit}\big(\hat{p}^{\text{Init}}_t(X_i)\big)+ \sum_{j=1}^J \frac{\hat{\phi}_j(X_i)\hat{S}_j(X_i)}{\lambda}\Big)\\
	 &- \log \Big( 1 + \text{exp}^{ \text{logit}\big(\hat{p}^{\text{Init}}_t(X_i)\big)+  \sum_{j=1}^J \frac{\hat{\phi}_j(X_i)\hat{S}_j(X_i)}{\lambda} } \Big) \Big] + \lambda (\frac{||\bm{\hat{\phi}(X_i)}||_2}{\lambda} - \delta).
\end{align*}
Hence, the dual problem reduces to 
\begin{align*}
 & \max_{\lambda\geq 0}\  -\frac{1}{n}\sum_{i:T_i=t}   \Big[ Y_i \Big( \text{logit}\big(\hat{p}^{\text{Init}}_t(X_i)\big)+ \sum_{j=1}^J \frac{\hat{\phi}_j(X_i)\hat{S}_j(X_i)}{\lambda}\Big)\\
	 & \qquad \qquad- \log \Big( 1 + \text{exp}^{ \text{logit}\big(\hat{p}^{\text{Init}}_t(X_i)\big)+  \sum_{j=1}^J \frac{\hat{\phi}_j(X_i)\hat{S}_j(X_i)}{\lambda} } \Big) \Big] + \lambda (\frac{||\bm{\hat{\phi}(X_i)}||_2}{\lambda} - \delta).
\end{align*}
With the following substitution,
\begin{align*}
\sum_{j=1}^J \frac{\hat{\phi}_j(X_i)\hat{S}_j(X_i)}{\lambda} = \frac{\sum_{j=1}^J \hat{\phi}_j \hat{S}_j(X_i)}{||\bm{\hat{\phi}}(X_i)||_2}\cdot \frac{||\bm{\hat{\phi}}(X_i)||_2}{\lambda} = \tilde{S}_t(X_i) \cdot \frac{||\bm{\hat{\phi}}(X_i)||_2}{\lambda}, 
\end{align*}
the dual problem satisfies
\begin{align*}
&\arg\max_{\lambda\geq 0}\  -\frac{1}{n}\sum_{i:T_i=t}   \Big[ Y_i \Big( \text{logit}\big(\hat{p}^{\text{Init}}_t(X_i)\big)+ \tilde{S}_t(X_i) \cdot \frac{||\bm{\hat{\phi}}(X_i)||_2}{\lambda} \Big)\\
	 &\qquad \qquad - \log \Big( 1 + \text{exp}^{ \text{logit}\big(\hat{p}^{\text{Init}}_t(X_i)\big)+  \tilde{S}_t(X_i) \cdot \frac{||\bm{\hat{\phi}}(X_i)||_2}{\lambda} } \Big) \Big] + \lambda (\frac{||\bm{\hat{\phi}}(X_i)||_2}{\lambda}- \delta),\\
= &\arg \max_{\lambda\geq 0}\  -\frac{1}{n}\sum_{i:T_i=t}   \Big[ Y_i \Big( \text{logit}\big(\hat{p}^{\text{Init}}_t(X_i)\big)+ \tilde{S}_t(X_i) \cdot \frac{||\bm{\hat{\phi}}(X_i)||_2}{\lambda} \Big)\\
	 &\qquad \qquad - \log \Big( 1 + \text{exp}^{ \text{logit}\big(\hat{p}^{\text{Init}}_t(X_i)\big)+  \tilde{S}_t(X_i) \cdot \frac{||\bm{\hat{\phi}}(X_i)||_2}{\lambda} } \Big) \Big] - \lambda\delta.
\end{align*}

Through the reparametrization $\gamma = \frac{||\hat{\phi}(X_i)||_2}{\delta}$, the dual problem can be reformulated as
\begin{align*}
  &\arg\min_{\gamma> 0} -\frac{1}{n}\sum_{i:T_i=t} \Big[ Y_i \Big( \text{logit}\big(\hat{p}^{(k-1)}_t(X_i)\big)+ \tilde{S}_t(X_i) \cdot \gamma \Big) \\
  &\quad\quad\quad\quad- \log \Big( 1 + \text{exp}^{ \text{logit}\big(\hat{p}^{(k-1)}_t(X_i)\big)+  \tilde{S}_t(X_i) \cdot \gamma } \Big) \Big] - \frac{||\bm{\hat{\phi}}(X_i)||_2\cdot\delta}{\gamma},
 \end{align*}
 where $\frac{||\bm{\hat{\phi}}(X_i)||_2\cdot\delta}{\gamma}$ is sufficiently close to 0.



\section{Additional Remarks}

\begin{remark}\label{remark:uni-local}
The multivariate local least favorable submodel defined in Eq (2) implies a one-dimensional universal least favorable submodel. The resulting one-step TMLE along this universal least favorable submodel corresponds with iteratively maximizing log likelihood for the multivariate least favorable submodel in $\bm{\varepsilon}$ under a constraint that $||\bm{\varepsilon}||=\gamma$ for some small $\gamma$, and noting that this MLE $\bm{\varepsilon}$ is known in closed form and equals $\gamma \cdot \frac{P_n \varphi} {||P_n \varphi||}$. One stops the iteration when the log-likelihood reaches its maximum. Our proposed method in Eq (3) can be viewed as a variation of this in the targeted learning literature.
\end{remark}

\begin{remark}\label{remark:continuous-measure}
The above theoretical results can be readily extended to infinitely many subgroups cases. Concretely, we denote $\bm{\alpha}_t(\bm{\nu})$ as a set of subgroup parameters indexed by a continuous vector $\bm{\nu}$ defined on a compact parameter space, where $\bm{\alpha}_{t}(\bm{\nu}) =  {P}\big( Y(t)=1|X\in\mathcal{A}(\bm{\nu})\big)$, and we define the vector of efficient influence functions as $\bm{\varphi}_t(\bm{\nu})$. As long as $\bm{\alpha}_t(\bm{\nu})$ is a smooth function of $\bm{\nu}$ the function class $\{ \bm{\varphi}(o;\bm{\alpha}_t(\bm{\nu}), \bm{\eta}), \bm{\eta}\in \mathcal{H}\}$  is a Donsker class, our theoretical results suggest that $\sqrt{n}(\hat{\bm{\alpha}}_t(\bm{\nu}) - \bm{\alpha}_t(\bm{\nu}))$ converges to a Gaussian process. 
\end{remark}

\section{Extension of the Proposed Method}

\subsection{Extension to continuous outcomes}\label{appendix:continuous-outcome}

There are two options to adapt our proposed method to continuous outcomes.  The first option is to use a different updating procedure for continuous outcomes \citep{gruber2010targeted}. For example, $\hat{p}^{(k)}_t$ can be obtained through a linear update:
\begin{align*}
\hat{p}^{(k)}_t(X_i) &= \hat{p}^{(k-1)}_t(X_i) + \hat{\gamma}^{(k)}\cdot \tilde{S}^{(k-1)}_t(X_i), \quad i\in\{i:T_i=t\}, 
\end{align*}
where $$\tilde{S}^{(k-1)}_t(X_i) =
 \frac{ \sum_{j=1}^d \frac{\mathds{1}(X_i\in\mathcal{A}_j)}{\hat{P}(\mathcal{A}_j)} \frac{\mathds{1}( T_i=t)}{\hat{e}_t(X_i)}\cdot\Big(\sum_{l=1}^n \hat{\phi}^{(k-1)}_j(Y_l, T_l,X_l)\Big)}{\sqrt{ \sum_{j=1}^d \left(  \sum_{l=1}^n \hat{\phi}^{(k-1)}_j(Y_l,T_l,X_l) \right)^2}},$$ and $\hat{\phi}^{(k-1)}_j(Y_i, T_i,X_i) = \frac{\mathds{1}(X_i\in\mathcal{A}_j)}{\hat{P}(\mathcal{A}_j)}\frac{\mathds{1}( T_i=t)}{\hat{e}_t(X_i)}(Y_i - \hat{p}_t^{(k-1)}(X_i)).$ $\hat{{\gamma}}^{(k)}$ can then be obtained by minimizing a user specified loss function $L(\cdot)$,
\begin{align*}
\hat{{\gamma}}^{(k)}  &= \underset{\gamma\geq 0}{\arg\min}\  \frac{1}{n}\sum_{i: T_i=t}   L \Big( Y_i,\ \hat{p}^{(k-1)}_t(X_i), \gamma, \ \tilde{S}^{(k-1)}_t(X_i) \Big).
\end{align*}
For example, one can consider the $l_2$-loss \citep{gruber2010targeted}, that is
\begin{align*}
    \hat{{\gamma}}^{(k)}  &= \underset{\gamma\geq 0}{\arg\min}\  \frac{1}{n}\sum_{i: T_i=t}    || Y_i -  \hat{p}^{(k)}_t(X_i)||_2^2.
\end{align*}

The second option is to dichotomize a continuous outcome into a binary outcome \citep{lewis2004defence,van2011targeted}. Then the proposed methodology in the main manuscript can be directly applied.

\subsection{Simultaneous confidence intervals of absolute risk difference, relative risk, and odds ratio}\label{supp:simul}
Let $\hat{\kappa}(q, J)$ be a consistent estimate of the $(1-q)$-th quantile of  $\max_{j\in 1, \ldots, J} |Z_j|$, where $Z_1, \ldots, Z_J$ are i.i.d. standard normal random variables. Then,
\begin{align*}
\hat{\alpha}_{\text{ARD},j} \pm \hat{\kappa}_{\text{ARD},q/2}\Big(\frac{\hat{\Sigma}_{\text{ARD},jj}}{n}\Big)^{1/2},\quad\hat{\alpha}_{\text{RR},j} \pm \hat{\kappa}_{\text{RR},q/2}\Big(\frac{\hat{\Sigma}_{\text{RR},jj}}{n}\Big)^{1/2},\quad\hat{\alpha}_{\text{OR},j}  \pm \hat{\kappa}_{\text{OR},q/2}\Big(\frac{\hat{\Sigma}_{\text{OR},jj}}{n}\Big)^{1/2}
\end{align*}
\begin{align*}
\lim_{n\rightarrow \infty}\mathbb{P}\left(\hat{\alpha}_{\text{ARD},j} \pm \hat{\kappa}(q, J) \cdot\Big(\frac{\hat{\Sigma}_{\text{ATE},jj}}{n}\Big)^{1/2}, j =1, \ldots, J\right) = 1-q,\\
\lim_{n\rightarrow \infty}\mathbb{P}\left(\hat{\alpha}_{\text{RR},j} \pm \hat{\kappa}(q, J) \cdot\Big(\frac{\hat{\Sigma}_{\text{RR},jj}}{n}\Big)^{1/2}, j =1, \ldots, J\right) = 1-q,\\
\lim_{n\rightarrow \infty}\mathbb{P}\left(\hat{\alpha}_{\text{OR},j} \pm \hat{\kappa}(q, J) \cdot\Big(\frac{\hat{\Sigma}_{\text{OR},jj}}{n}\Big)^{1/2}, j =1, \ldots, J\right) = 1-q,
\end{align*}
where $\hat{\Sigma}_{\text{ARD}} = \big( \hat{\Sigma}_{{\text{ARD}}, jk}\big)_{j,k=1}^J = \frac{1}{n}\sum_{i=1}^n \bm{\hat{\varphi}}_{{\text{ARD}},i}\bm{\hat{\varphi}}_{\text{ARD},i}' $, $$\bm{\hat{\varphi}}_{{\text{ARD}},i} = \big(\hat{\varphi}_{{\text{ARD}},1}(Y_i,T_i,X_i), \ldots, \hat{\varphi}_{{\text{ARD}},J}(Y_i,T_i,X_i)\big)' $$. Similarly, we can construct the covariance matrix $\hat{\Sigma}_{\text{RR}}$ and $\hat{\Sigma}_{\text{OR}}$. The plug-in estimates of the efficient influence functions are,
\begin{align*}
\hat{\varphi}_{\text{ARD},j}(O_i) &=  \frac{ \mathds{1}(X_i \in \mathcal{A}_j)}{\hat{\mathds{P}}(\mathcal{A}_j)}\Big[\Big(\frac{T_i}{\hat{e}_1(X_i)}\big(Y_i-\hat{p}_1(X_i)\big)+\hat{p}_1(X_i)-\hat{\alpha}_1\Big)\\
&- \Big(\frac{1-T_i}{\hat{e}_0(X_i)}\big(Y_i-\hat{p}_0(X_i)\big)+\hat{p}_0(X_i)-\hat{\alpha}_0\Big)\Big],\\
\hat{\varphi}_{\text{RR},j}(O_i) &=  \frac{ \mathds{1}(X_i \in \mathcal{A}_j)}{\hat{\mathds{P}}(\mathcal{A}_j)}\Big[\frac{1}{\hat{\alpha}_0}\Big(\frac{T_i}{\hat{e}_1(X_i)}\big(Y_i-\hat{p}_1(X_i)\big)+\hat{p}_1(X_i)-\hat{\alpha}_1\Big)\\
&- \frac{\hat{\alpha}_1}{\hat{\alpha}_0^2}\Big(\frac{1-T_i}{\hat{e}_0(X_i)}\big(Y_i-\hat{p}_0(X_i)\big)+\hat{p}_0(X_i)-\hat{\alpha}_0\Big)\Big],\\
\hat{\varphi}_{\text{OR},j}(O_i) &=  \frac{ \mathds{1}(X_i \in \mathcal{A}_j)}{\hat{\mathds{P}}(\mathcal{A}_j)}\Big[\frac{1-\hat{\alpha}_0}{\hat{\alpha}_0(1-\hat{\alpha}_1)^2}\Big(\frac{T_i}{\hat{e}_1(X_i)}\big(Y_i-\hat{p}_1(X_i)\big)+\hat{p}_1(X_i)-\hat{\alpha}_1\Big)\\
&- \frac{\hat{\alpha}_1}{\hat{\alpha}_0^2(1-\hat{\alpha}_1)}\Big(\frac{1-T_i}{\hat{e}_0(X_i)}\big(Y_i-\hat{p}_0(X_i)\big)+\hat{p}_0(X_i)-\hat{\alpha}_0\Big)\Big].
\end{align*}

We obtain the efficient influence functions for $\bm{\alpha}_{\textbf{ARD}}$, $\bm{\alpha}_{\textbf{RR}}$, and $\bm{\alpha}_{\textbf{OR}}$ by applying multivariate delta method on $(\bm{\alpha_1}, \bm{\alpha_0})$.

\section{Proof of Eq (7) in Section \ref{subsec:proposal}}\label{appendix:multi-score-proof}

In this part, we aim to derive the score function under the iterative procedure. For simplicity, denote the final update as $\hat{p}_1(X)$, $\hat{\varepsilon}$. The conditional likelihood function of $Y$ given $(T,X)$ is:
\begin{align*}
L(Y|T,X) = p(T, X)^Y \cdot \big(1- p(T,X)\big)^{1-Y}.
\end{align*}

\begin{proof}
\begin{align*}
  \frac{\partial \text{log} L(Y|T,X;\varepsilon)}{\partial \varepsilon} &= \frac{1}{n}\sum_{i=1}^n\sum_{j=1}^J\Big\{ Y_i \tilde{S}_1(X_i) - \frac{\text{expit}\Big(\text{logit}\big(\hat{p}^{(K-1)}_1(X_i)\big) + \hat{\varepsilon}\cdot \tilde{S}_1(X_i)\Big)}{1+\text{expit}\Big(\text{logit}\big(\hat{p}^{(K-1)}_1(X_i)\big) + \hat{\varepsilon}\cdot \tilde{S}_1(X_i)\Big)}\cdot \tilde{S}_1(X_i)\Big\} = 0,\\
  &= \frac{1}{n}\sum_{i=1}^n\sum_{j=1}^J\Big\{ \Big[Y_i -\frac{\text{expit}\Big(\text{logit}\big(\hat{p}^{(K-1)}_1(X_i)\big) + \hat{\varepsilon}\cdot \tilde{S}_1(X_i)\Big)}{1+\text{expit}\Big(\text{logit}\big(\hat{p}^{(K-1)}_1(X_i)\big) + \hat{\varepsilon}\cdot \tilde{S}_1(X_i)\Big)}\Big]\cdot \tilde{S}_1(X_i)\Big\}=0,\\
  &= \frac{1}{n}\sum_{i=1}^n\sum_{j=1}^J \Big\{\Big(Y_i - \hat{p}_1(X_i)\Big)\cdot \tilde{S}_1(X_i)\Big\} = 0,\\
  &= \frac{1}{n}\sum_{i=1}^n\sum_{j=1}^J \Big\{\Big(Y_i - \hat{p}_1(X_i)\Big)\cdot \frac{\mathds{1}(X_i\in\mathcal{A}_j)}{\hat{\mathds{P}}(\mathcal{A}_j)}\frac{T_i}{\hat{e}_1(X_i)}\cdot\frac{\frac{1}{n}\sum_{i=1}^n \hat{\phi}_j(Y_i, T_i,X_i)}{\sqrt{ \sum_{m=1}^J \left(  \frac{1}{n}\sum_{i=1}^n \hat{\phi}_m(Y_i,T_i,X_i) \right)^2}}\Big\},\\
  &= \frac{\sum_{j=1}^J\Big\{ \frac{1}{n}\sum_{i=1}^n  \big\{ \hat{\phi}_j(Y_i,T_i,X_i) \big\} \frac{1}{n}\sum_{i=1}^n  \big\{\hat{\phi}_j(Y_i,T_i,X_i)\big\}\Big\}}{\sqrt{ \sum_{m=1}^J \left(  \frac{1}{n}\sum_{i=1}^n  \hat{\phi}_m(Y_i,T_i,X_i) \right)^2}},\\
  &= \frac{\sum_{j=1}^J 
 \Big(\frac{1}{n}\sum_{i=1}^n \hat{\phi}_j(Y_i,T_i,X_i)\Big)^2}{\sqrt{ \sum_{m=1}^J \left(  \frac{1}{n}\sum_{i=1}^n \hat{\phi}_m(Y_i,T_i,X_i) \right)^2}}=\sqrt{ \sum_{j=1}^J \left(  \frac{1}{n}\sum_{i=1}^n \hat{\phi}_j(Y_i,T_i,X_i) \right)^2} = 0.
\end{align*}
\end{proof}

\section{Implementation Details}

\subsection{Implementation details of cross-fitted TMLE }\label{appendix:cv-TMLE}
As mentioned in the main manuscript, the Donsker class condition on the efficient influence function can be relaxed by cross-fitting. Here, we briefly discuss the implementation details of the cross-fitted iterative version of the one-step TMLE of the multivariate dimensional parameters. The non-iterative version can be carried out similarly. 

\begin{description}
\item[Step 1.] Randomly split the sample into $V$ equal-sized subsamples.  
\item[Step 2.]For $v \leftarrow 1$ to $V$:
\item[\quad\quad\quad\quad (a)] Use subsample $v$ as the validation data and the rest as training data. Generate initial estimates of $p_t(X)$ and $e_t(X)$ by fitting the model on the training set, and predict on the validation set, denoted as $\hat{p}_{t,v}^{(0)}(X)$ and $\hat{e}_{t,v}(X)$. 
\item[\quad\quad\quad\quad (b)] For $k \leftarrow 1, \ldots, K$ (or until converge):
\begin{align*} 
	   \hat{\varepsilon}_v^{(k)} = \underset{\varepsilon \in\mathds{R}}{\arg\max}  &\Bigg\{\frac{1}{n_v}\sum_{i:T_i=t, i\in v}  \Big[ Y_i \Big( \text{logit}\big(\hat{p}^{(k-1)}_{t,v}(X_i)\big)+ \varepsilon  \tilde{S}^{(k-1)}_{t,v}(X_i) \Big)\\
	  &-\log \Big( 1 + \text{exp}^{ \text{logit}\big(\hat{p}^{(k-1)}_{t,v}(X_i)\big)+ \varepsilon  \tilde{S}^{(k-1)}_{t,v}(X_i) } \Big) \Big]\Bigg\},
\end{align*}
\quad\quad\quad\quad where
\begin{align*}
    \tilde{S}^{(k-1)}_{t,v}(X_i)=\frac{\sum_{j=1}^d\frac{1}{n_v}\sum_{l\in v}\hat{\phi}^{(k-1)}_j(Y_l, T_l,X_l)}{\sqrt{ \sum_{m=1}^d \left(  \frac{1}{n_v}\sum_{i\in v} \hat{\phi}^{(k-1)}_m(Y_i,T_i,X_i) \right)^2}} \cdot  \frac{\mathds{1}(X_i\in\mathcal{A}_j)}{\hat{\mathds{P}}(\mathcal{A}_j)} \frac{T_i}{\hat{e}(X_i)},
\end{align*}
\quad\quad\quad\quad and $\hat{\phi}^{(k-1)}_{j,v} = \frac{\mathds{1}(X_i\in\mathcal{A}_j)}{\hat{\mathds{P}}(\mathcal{A}_j)} \frac{T_i}{\hat{e}_t(X_i)}(Y_i - \hat{p}_{t,v}^{(k-1)}(X_i))$.
\item[\quad\quad\quad\quad (c)] Update the conditional risk via: 
\begin{align*}
	    \hat{p}_{t,v}^{(k)}( X_i) &= \text{expit}\left( \text{logit}\big(\hat{p}^{(k-1)}_{t,v}( X_i) \big)
	          + \varepsilon_v^{(k)} \cdot \tilde{S}^{(k-1)}_{t,v}(X_i) \right).
			\end{align*}
\item[\quad\quad\quad\quad (d)] Estimate $\alpha_{j,v}$ by: 
\begin{align*}
\hat{\alpha}_{j,v} = \frac{
\sum_{i\in v} \mathds{1}(X_i\in \mathcal{A}_j)\cdot \hat{p}_{t,v}^{(K)}(X_i)}{\sum_{i\in v}\mathds{1}(X_i\in\mathcal{A}_j)}.
\end{align*}
\item[Step 3.] Aggregate estimates from the validation sets by:
\begin{align*}
\hat{\alpha}_j = \frac{1}{V}
\sum_{v=1}^V \hat{\alpha}_{j,v}.
\end{align*}
	\end{description}

\subsection{Implementation details of the relative risk and odds ratio estimators}\label{supp:implementation}

\begin{description}
\item[Step 1.] Randomly split the sample into $V$ equal-sized subsamples.  
\item[Step 2.]For $v \leftarrow 1$ to $V$:
\item[\quad\quad\quad\quad (a)] Use subsample $v$ as the validation data and the rest as training data. Generate initial estimates of $p_1(X)$, $p_0(X)$, $e_1(X)$ and $e_0(X)$ by fitting the model on the training set, and predict on the validation set, denoted as $\hat{p}_{1,v}^{(0)}(X)$, $\hat{p}_{0,v}^{(0)}(X)$, $\hat{e}_{1,v}(X)$, and $\hat{e}_{0,v}(X)$. 
\item[\quad\quad\quad\quad (b)] For $k \leftarrow 1, \ldots, K$ (or until converge):
\begin{align*} 
	   \hat{\varepsilon}_{t,v}^{(k)} = \underset{\varepsilon \in\mathds{R}}{\arg\max}  &\Bigg\{\frac{1}{n_v}\sum_{i\in v}  \Big[ Y_i \Big( \text{logit}\big(\hat{p}^{(k-1)}_{t,v}(X_i)\big)+ \varepsilon  \tilde{S}^{(k-1)}_{t,v}(X_i) \Big)\\
	  &-\log \Big( 1 + \text{exp}^{ \text{logit}\big(\hat{p}^{(k-1)}_{t,v}(X_i)\big)+ \varepsilon  \tilde{S}^{(k-1)}_{t,v}(X_i) } \Big) \Big]\Bigg\},
\end{align*}
\quad\quad\quad\quad where
\begin{align*}
    \tilde{S}^{(k-1)}_{t,v}(X_i)=\frac{\sum_{j=1}^d\frac{1}{n_v}\sum_{l\in v}\hat{\phi}^{(k-1)}_j(Y_l, T_l,X_l)}{\sqrt{ \sum_{m=1}^d \left(  \frac{1}{n_v}\sum_{i\in v} \hat{\phi}^{(k-1)}_m(Y_i,T_i,X_i) \right)^2}} \cdot  \frac{\mathds{1}(X_i\in\mathcal{A}_j)}{\hat{\mathds{P}}(\mathcal{A}_j)} \frac{\mathds{1}(T_i=t)}{\hat{e}(X_i)},
\end{align*}
\quad\quad\quad\quad and $\hat{\phi}^{(k-1)}_{j,v} = \frac{\mathds{1}(X_i\in\mathcal{A}_j)}{\hat{\mathds{P}}(\mathcal{A}_j)} \frac{T_i}{\hat{e}_t(X_i)}(Y_i - \hat{p}_{t,v}^{(k-1)}(X_i))$.
\item[\quad\quad\quad\quad (c)] Update the conditional risk via: 
\begin{align*}
	    \hat{p}_{t,v}^{(k)}(T_i,X_i) &= \text{expit}\left( \text{logit}\big(\hat{p}^{(k-1)}_{t,v}(X_i) \big)
	          + \varepsilon_{t,v}^{(k)} \cdot \tilde{S}^{(k-1)}_{1,v}(X_i) \right).
			\end{align*}
\item[\quad\quad\quad\quad (d)] Estimate $\alpha_{j,v}$ by: 
\begin{align*}
\hat{\alpha}_{tj,v} = \frac{
\sum_{i\in v} \mathds{1}(X_i\in \mathcal{A}_j)\cdot \hat{p}_{t,v}^{(K)}(X_i)}{\sum_{i\in v}\mathds{1}(X_i\in\mathcal{A}_j)}.
\end{align*}
\item[Step 3.] Aggregate estimates from the validation sets by:
\begin{align*}
\hat{\alpha}_{tj} = \frac{1}{V}
\sum_{v=1}^V \hat{\alpha}_{tj,v}.
\end{align*}
\item[Step 4.] Estimate $\bm{\alpha}_{\text{RR}}$ and $\bm{\alpha}_{\text{OR}}$ as
\begin{align*}
\bm{\hat{\alpha}_{\textbf{RR}}} &= \Big(\frac{\hat{\alpha}_{1,1}}{\hat{\alpha}_{0,1}},\dots, \frac{\hat{\alpha}_{1,d}}{\hat{\alpha}_{0,d}}\Big),\quad\bm{\hat{\alpha}_{\textbf{OR}}} = \Big(\frac{\hat{\alpha}_{1,1}}{1-\hat{\alpha}_{1,1}}\Big/\frac{\hat{\alpha}_{0,1}}{1-\hat{\alpha}_{0,1}},\dots, \frac{\hat{\alpha}_{1,d}}{1-\hat{\alpha}_{1,d}}\Big/\frac{\hat{\alpha}_{0,d}}{1-\hat{\alpha}_{0,d}}\Big).
\end{align*}
	\end{description}

\section{Additional Simulation Results}\label{appendix:simulation}

\subsection{Comparison of the targeted learning approaches discussed in Section \ref{sec:method}}\label{appendix:sim-comparison}

In this simulation study, we compare the proposed iTMLE method with the conventional targeted learning approach without targeting multiple subgroups. That is, we compare the performances for four estimators: (1) ``TMLE-single," which is the TMLE without targeting multiple subgroups as discussed in Section \ref{sec:limitation-of-one-step-TMLE}, (2) ``TMLE-multiple," which is the one-step TMLE estimator that targets multiple subgroups discussed in Eq (2), (3) ``TMLE-multiple-ulfm," which is the same method as (2) but operates under the universal least favorable submodel, and (4) ``iTMLE," which is an iterative version of the one-step TMLE estimator discussed in Eq (3).

\begin{figure}[!p]
\centering
\includegraphics[width=\textwidth]{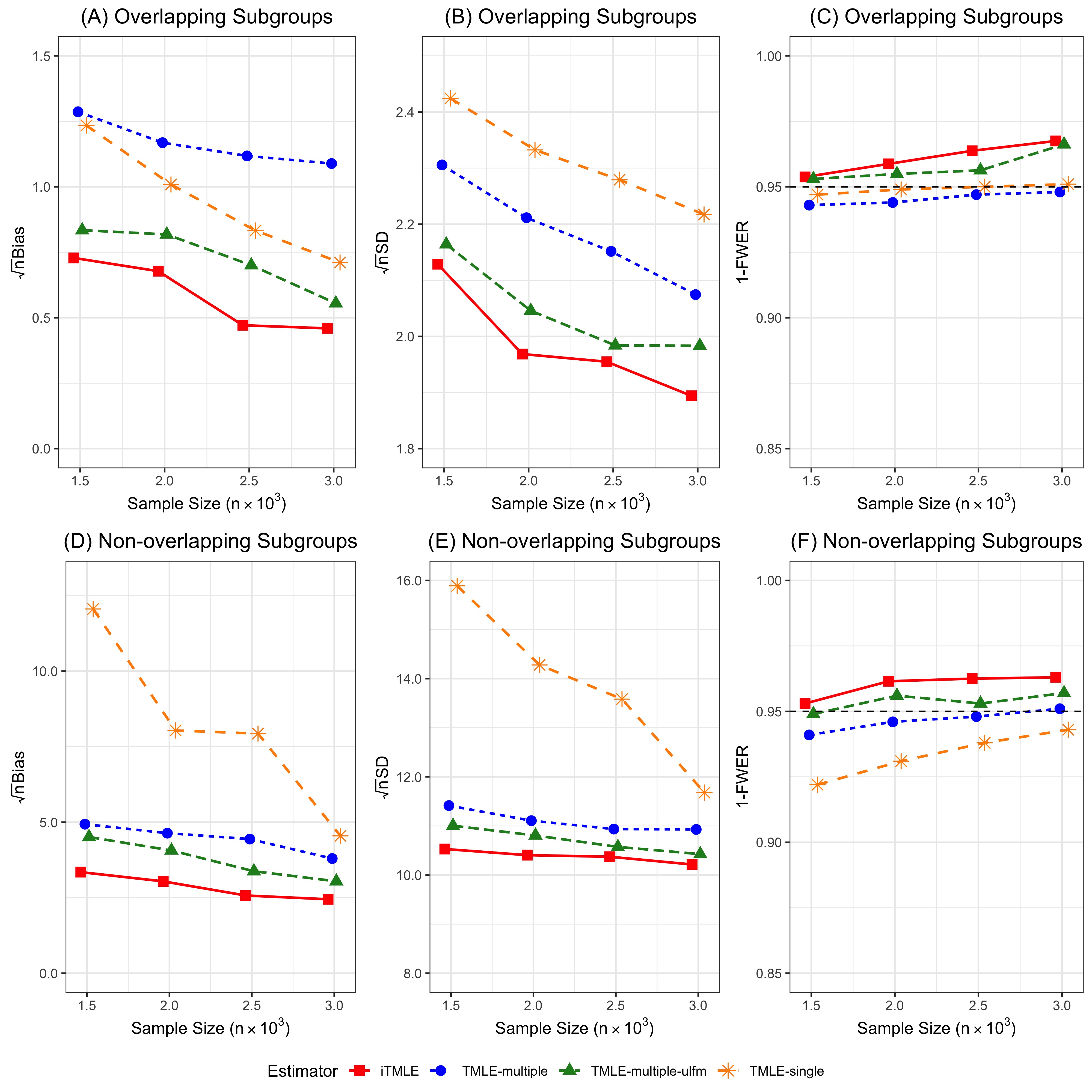}
\caption{Comparison of four TMLE estimators under overlapping ($d=4$) and non-overlapping subgroups ($d=10$). The considered methods include (1) iTMLE; (2) TMLE-single; (3) TMLE-multiple; (4) TMLE-multiple-ulfm.  All four methods use the random forest for the initial estimations. The maximum Monte Carlo standard error of (1-FWER) across the four estimators is 0.027. ``The maximum Monte Carlo standard error of (1-FWER)" refers to the largest standard error of (1-FWER) (out of all three considered estimators for the propensity score and the conditional expectation of the outcome based on logistic regression, random forest, and gradient boosting) computed from Monte Carlo samples.}
\label{fig:sim-tmle}
\end{figure}

Figure \ref{fig:sim-tmle} demonstrates that when $d=4$, iTMLE has smaller bias, smaller variance, and lower FWER. When the number of subgroups is not too large ($d=4$), targeting one subgroup at a time (``TMLE-single") yields smaller bias than targeting multiple subgroups (``TMLE-multiple"), while ``TMLE-single" loses control over bias, variance, and FWER when the number of subgroups is large. ``TMLE-multiple-ulfm" shows similar performance to iTMLE, but with slightly larger bias and variance.

\subsection{Misspecified propensity score model }

In this section, we compare the performance of the discussed method with other conventional estimators under the mis-specified propensity score model. The results from Figure \ref{supp:fig:tmle-dr} and Figure \ref{supp:fig:tmle-dml} are in-line with simulation studies in Section 6 of the main manuscript. 

\begin{figure}[ht!]
\centering
\includegraphics[width=\textwidth]{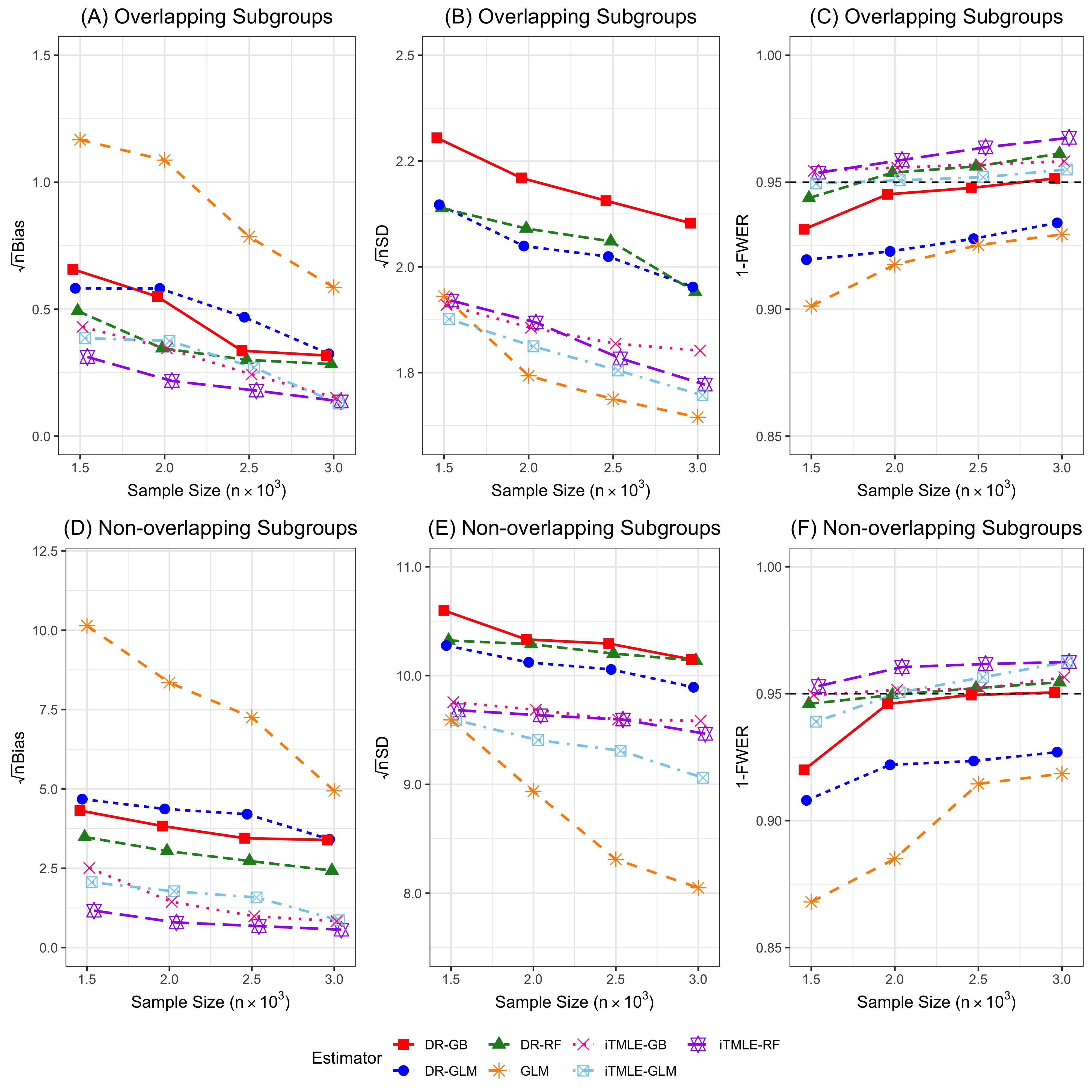}
\caption{Comparison of iTMLE with the conventional methods. ``iTMLE" denotes the discussed method. ``DR" denotes the doubly robust estimator. ``GLM" denotes the generalized linear models. The maximum Monte Carlo standard error of (1-FWER) across the four estimators is 0.030. ``The maximum Monte Carlo standard error of (1-FWER)" refers to the largest standard error of (1-FWER) (out of all three considered estimators for the propensity score and the conditional expectation of the outcome based on logistic regression, random forest, and gradient boosting) computed from Monte Carlo samples.}
\label{supp:fig:tmle-dr}
\end{figure}

\begin{figure}[ht!]
\centering
\includegraphics[width=\textwidth]{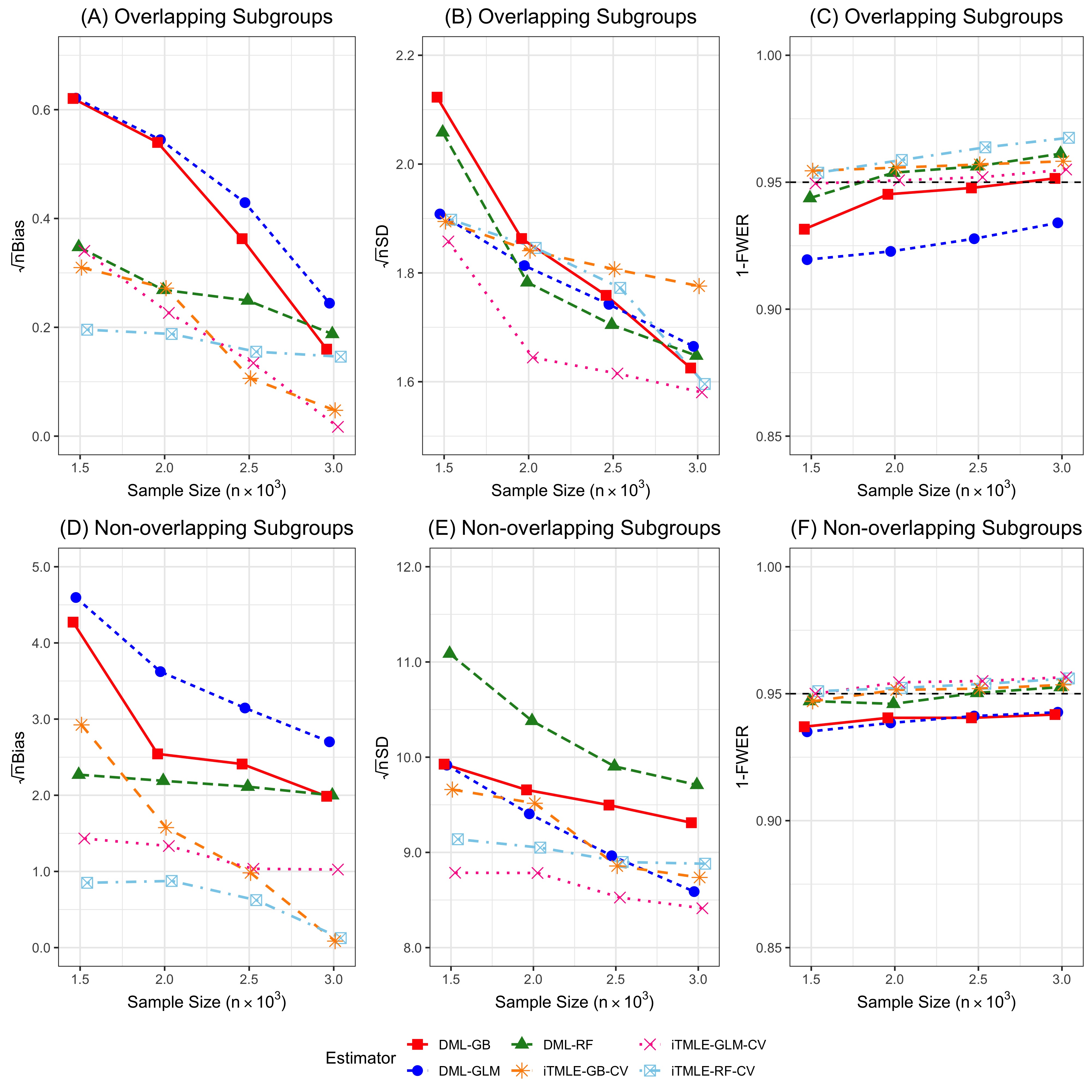}
\caption{Comparison of CV-iTMLE with the double machine learning method. ``CV-iTMLE" denotes  the discussed method with cross-fitting. ``DML" denotes the double machine learning method. The maximum Monte Carlo standard error of (1-FWER) across the four estimators is 0.028. ``The maximum Monte Carlo standard error of (1-FWER)" refers to the largest standard error of (1-FWER) (out of all three considered estimators for the propensity score and the conditional expectation of the outcome based on logistic regression, random forest, and gradient boosting) computed from Monte Carlo samples.}
\label{supp:fig:tmle-dml}
\end{figure}

\subsection{An alternative simulation design}\label{appendix:alternative-sim}

Kindly pointed out by an anonymous reviewer, the simulation design adopted in our main manuscript produces rather deterministic outcomes. Therefore, we provide additional simulation results under an alternative simulation design:
\begin{align*}
    X = (X_1, \ldots, X_5)^{\intercal} &\sim N(0, \Sigma),\quad \Sigma_{ij} = 0.5^{|i-j|}, \\
    T &\sim \text{Bernoulli}\Big(\text{expit}(X_{1} - 0.5\cdot X_{2} + 0.25 \cdot X_{3} + 0.1 \cdot X_{4})\Big),\\
    Y | T, X &\sim \text{Bernoulli}\Big( \text{expit} (T + \cdot X_{1} + \cdot X_{2} + \cdot X_{3} + \cdot X_{4}) \Big).
\end{align*}

We provide the distribution of $p_t(X) = P(Y=1|T = t, X)$ under the alternative simulation design in Figured \ref{fig:pt-dist} (A) and (B), and under the simulation design adopted in the main manuscript in Figured \ref{fig:pt-dist} (C) and (D).

\begin{figure}[h]
\centering
\includegraphics[width=0.85\textwidth]{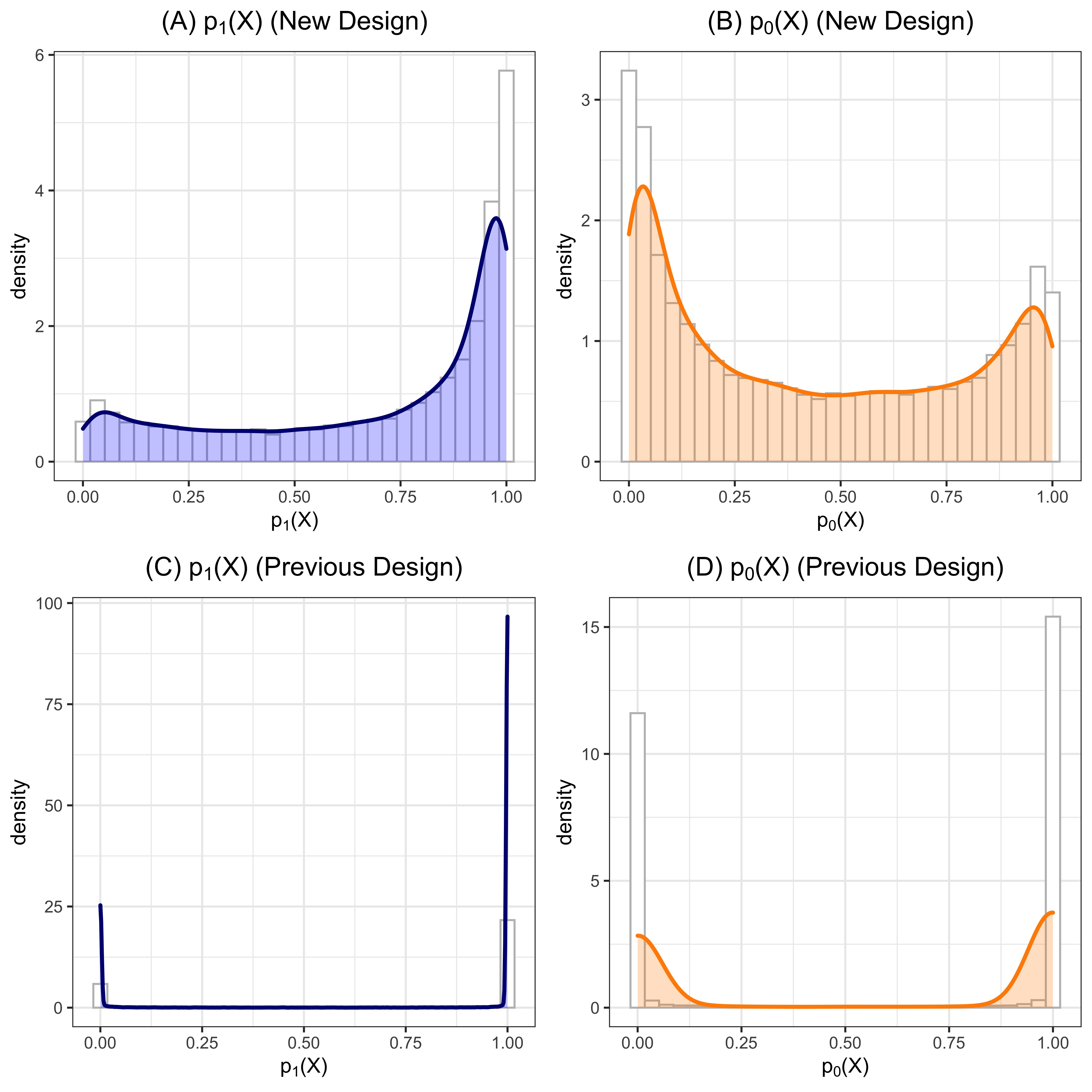}
\caption{Panels (A) and (B) provide the distribution of $p_t(X)$ under the new simulation design, and Panels (C) and (B) provide the distribution of $p_t(X)$ under the original simulation design.}
\label{fig:pt-dist}
\end{figure}

The results are summarized in Figure \ref{fig:sim1-new}
and \ref{fig:sim2-new}. From these new simulation results, we observe that the conclusions overall do not substantially  differ  from our previous simulation results. The $\sqrt{n}$-scaled standard deviation of all the estimators increase compared to the previous simulation design. We conjecture that the increased variances are due to the higher variability of the outcome variables under the new simulation setup. 

\begin{figure}[!p]
\centering
\includegraphics[width=\textwidth]{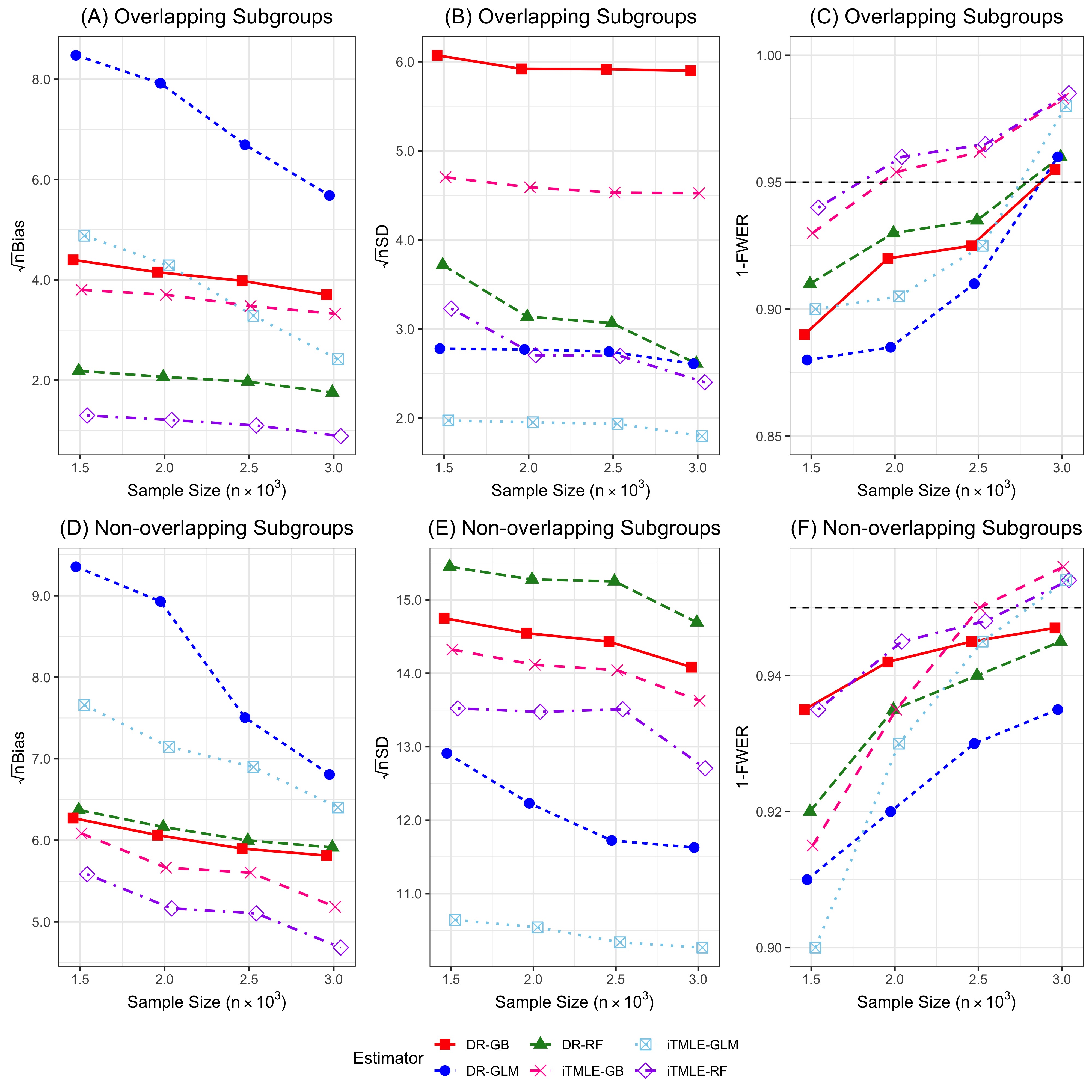}
\caption{Comparison of bias, standard deviation (scaled by root-$n$), and (1-FWER) in overlapping and non-overlapping subgroups. ``iTMLE" denotes the proposed estimator. ``DR" denotes the doubly robust estimator. The maximum Monte Carlo standard error of (1-FWER) across the four estimators is 0.030. ``The maximum Monte Carlo standard error of (1-FWER)" refers to the largest standard error of (1-FWER) (out of all three considered estimators for the propensity score and the conditional expectation of the outcome based on logistic regression, random forest, and gradient boosting) computed from Monte Carlo samples.}
\label{fig:sim1-new}
\end{figure}

\begin{figure}[!p]
\centering
\includegraphics[width=\textwidth]{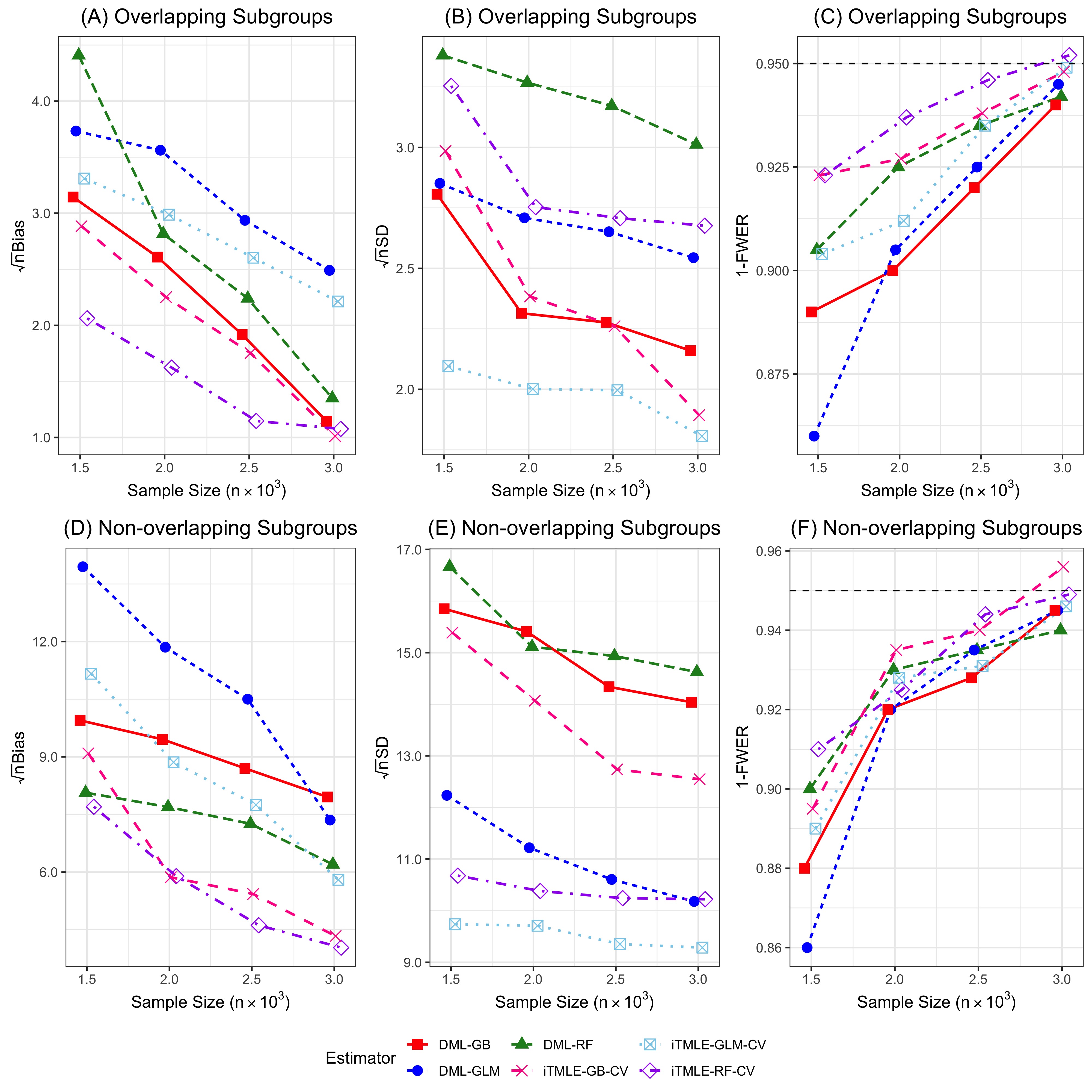}
\caption{Comparison of bias, standard deviation (scaled by root-$n$), and (1-FWER) in overlapping and non-overlapping subgroups. ``iTMLE" denotes the proposed estimator. ``DML" denotes the double machine learning method. The maximum Monte Carlo standard error of (1-FWER) across the four estimators is 0.031. ``The maximum Monte Carlo standard error of (1-FWER)" refers to the largest standard error of (1-FWER) (out of all three considered estimators for the propensity score and the conditional expectation of the outcome based on logistic regression, random forest, and gradient boosting) computed from Monte Carlo samples. }
\label{fig:sim2-new}
\end{figure}

\section{Details on the Case Studies}

\subsection{UK Biobank data preprocessing details}\label{appendix:data-details}

In UK Biobank study, participants provided lifestyle, medical history, and other health-related information through electronic questionnaires and physical measurements at one of the 22 assessment centers. Blood samples were also collected for genotyping. The UK Biobank study gained approval from the National Health Service’s National Research Ethics Service North West (11/NW/0382).

The individuals we investigated are unrelated and had passed standard quality control steps, including removal of outliers for heterozygosity or genotype missing rate, withdrawal of informed consent, and mismatch between reported and inferred sex by genotypes.

We obtain phenotype and genotype data from UK Biobank with the following steps. For phenotype data, first, we download encoded data in \textit{.enc} format from UK Biobank's Access Management System (AMS). To decrypt the encoded data, we download three helper programs: \texttt{ukb\_md5}, \texttt{ukb\_unpack}, and \texttt{ukb\_conv}. Note that the helper programs are only supported by Windows and Linux systems.  Second, we verify the integrity 
of the encoded data via \texttt{ukb\_md5} and unpack them into \textit{.enc\_ukb} format with \texttt{ukb\_unpack}. To convert the data into readible format, we use \texttt{ukb\_conv} to convert the \textit{.enc\_ukb} data into \textit{.csv} format (other options include txt, r, sas stata or bulk format). The data dictionary can be obtained using \texttt{ukb\_conv} with \textit{docs} option. After decrypting and converting the encoded data, we obtain a dataset with sample size n = $502,481$ and $20,502$ variables. In our study, since we only work with the phenotypes at the baseline to avoid confounding issues, we extract baseline variables with the suffix ``-0.0". The phenotype data we extract include gender, age at recruitment, AD family history, International Statistical Classification of Disease 9th revision (ICD 9) and 10th revision (ICD 10) codes and  self-report for T2D and AD.

For genotype data, first, we download imputed genotypes and associated sample information for 23 chromosomes. \textit{Imputation BGEN} and \textit{Imputation sample} can be obtained via \texttt{ukb\_gene} program.   \textit{Imputation BGI} and \textit{Imputation MAF+info} can be downloaded directly from UK Biobank resources $1965$ and $1967$. Second, we use the \textit{Imputation sample} file to remove individuals without genotype information, which yields sample size $n=407,057$.  Finally, we read in BGEN files with \texttt{snp\_readBGEN} function in \texttt{R} package \texttt{bigsnpr} using \textit{HapMap3} as the reference genomes. \texttt{snp\_readBGEN} converts the BGEN files into an \texttt{R} object comprising of two elements: \textit{genotype} and \textit{map}, where \textit{genotype} represents the imputed genotypes in a matrix format and \textit{map} contains the features of SNPs (chromosome, rsid, physical position, major and minor alleles and allele frequency). We only extract the genotype matrix from \textit{genotype} and \textit{rsid} from \textit{map} as our genotype data. We restrict our sample to subjects used for the principle components (PCs) computation, since those individuals are unrelated. From the extracted genotype data, we obtain the treatment variable: rs12916 (on chromosome 5), a functionally equivalent SNP of statins. 

Since rs12916 and T2D are associated with some other SNPs due to pleiotropy, we adjust for low-density lipoprotein (LDL) and T2D related SNPs in our study. To find disease-associated SNPs, we rely on the published GWAS studies from the GWAS catalogue. In our study, we define the disease-associated SNPs as SNPs associated with LDL or T2D with $p$-values less than $5\times 10^{-8}$. To determine the $p-$values for multiple correlated SNPs in the same locus, we use the linkage disequilibrium clumping procedure with $R^2 < 0.01$.  Our filtration criteria yield $385$ disease-associated SNPs.

\subsection*{Disease status definition} 
We identify T2D cases from three sources: doctor diagnosis, and self-reports. If one's self-reported T2D status is missing, we define the self-reported T2D as the following: self-reported diabetes =1 and self-reported gestational diabetes = 0 and self-reported type 1 diabetes = 0. To identify AD cases, we reply on self-reports , family-reports, and ICD codes. We use ICD-10 codes: G309, G301, F002, F000, G308, G300, F009, F001, and ICD-9 codes: 3310.

\subsection{Case study: T2D as the outcome}
\label{appendix:case-study-t2d}

Because statin usage may increase the risk of T2D (\citealp{swerdlow2015hmg}), as a secondary analysis, we further investigated the  effect  of  rs12916-T allele  on  T2D under the same considered subgroups to evaluate the potential heterogeneous side effects. We still considered the ``high AD genetic risk" subgroup and the ``low AD genetic risk" subgroup in this setting because existing studies suggest that insulin resistance links T2D and AD \citep{chatterjee2018alzheimer}. As some genetic variants are shared between T2D and AD \citep{gao2016shared}, people with high AD genetic risk may be more vulnerable to T2D risk and thus is more sensitive to the side effects of statins use. Therefore, we hypothesized that the effect of carrying rs12916-T allele on T2D risks could be heterogeneous in subgroups with different AD genetic risks and evaluated our hypothesis through subgroup analysis. We compared the performance of the proposed method (CV-iTMLE) with the double machine learning (DML) method and the widely used generalized linear models (GLM). We used the random forest as our first stage estimator as it provides the most robust results in our simulation studies.

Figure \ref{fig:real-data-ad-2} demonstrated that the treatment effect of carrying rs12916-T allele on T2D risk is heterogeneous across considered subgroups. Both the proposed method and the double machine learning (DML) method suggested that carrying rs12916-T allele increases T2D risk in females and in individuals under 65. For the significant subgroup, the confidence interval of the proposed method was much shorter than that of DML, indicating that the proposed method is more efficient. Furthermore, our results showed that the effect of inheriting rs12916-T allele on T2D risk is heterogeneous in subgroups with different AD genetic risks, which potentially implies that subjects with higher AD genetic risk can be more vulnerable to statin usage. Such findings could be partially explained by the similar pathological pathways shared between T2D and AD \citep{li2007common}. Some of our above findings are in-line with current believes in the existing literature.
For example, the results from one randomized clinical trial (JUPITER trial \citealp{mora2010statins}) indicate that statins may increase T2D risk more significantly in females than in males. The significant adverse effect of carrying rs12916-T allele (a proxy for statin usage) in individuals with high AD genetic risk is a rather novel finding. We conjecture that this is because insulin resistance links T2D and AD \citep{chatterjee2018alzheimer}, and some genetic variants are also shared between T2D and AD \citep{gao2016shared}. These findings warrant further clinical studies.

In this secondary analysis, our proposed method again showed shortened confidence intervals and hence improved power in detecting significant subgroups, while the GLM and the double machine learning method tend to lose power. We conjectured that the double machine learning method failed to detect the adverse effect of rs12916-T allele on the high AD genetic risk subpopulation because of the large variance caused by small estimated propensity scores. In contrast, the proposed method is rather robust to the small estimated propensity scores.

\begin{figure}[!p]
\centering
\includegraphics[width=\textwidth]{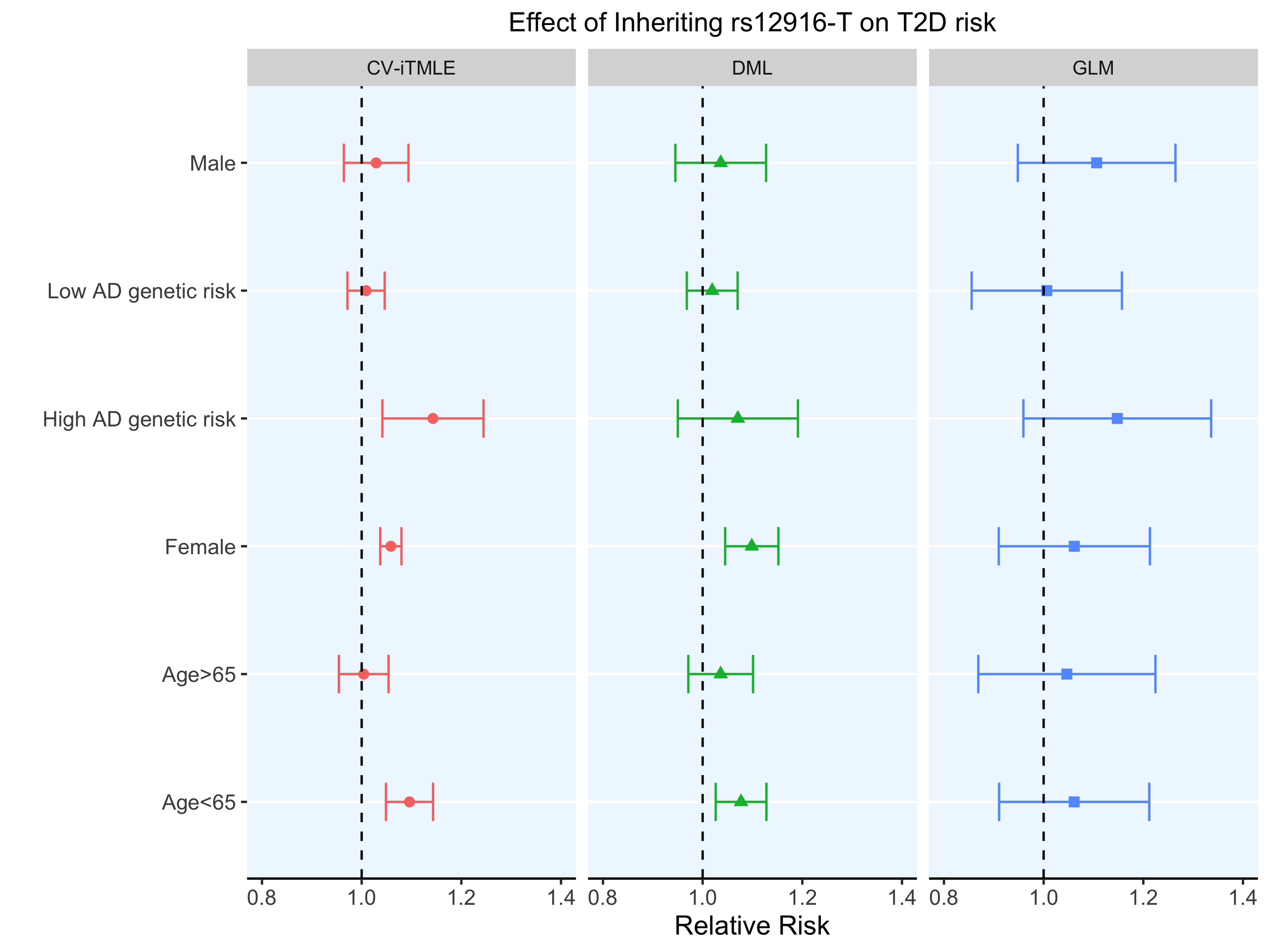}
\caption{The effect of inheriting rs12916-T allele (a proxy for statin usage) on the type 2 diabetes (T2D) risk in the UK Biobank white British population ($n=293,929$). ``DML" denotes the double machine learning method. ``GLM" denotes the generalized linear models. GLM is used for association test and does not imply causal relationships. ``CV-iTMLE" denotes the cross-validated iTMLE method.}
\label{fig:real-data-ad-2}
\end{figure}

\section{Efficient Influence Function and Delta Method}\label{supp:eif}
\subsection{Efficient influence function derivation for subgroup conditional risk}
Suppose our parameter of interest is $\alpha = E[Y|T=1, X\in \mathcal{A}_j]$. 
\begin{proof}
Given $\alpha(P_0) = \mathds{E}[ Y(1) | X\in \mathcal{A}_j ]$ and the $\varepsilon$-perturbed distribution is $P_{\varepsilon}^h$,
\begin{align*}
 \frac{\partial}{\partial \varepsilon} \alpha(P_{\varepsilon}^h) \big|_{\varepsilon=0} = & \lim_{\varepsilon\rightarrow 0} \frac{1}{\varepsilon}\Big\{  \int_{y\in\mathds{R}} t\cdot y  \cdot \mbox{d}
  \alpha(P_{\varepsilon}^h) - \alpha(P)\Big\},\\
  =&\lim_{\varepsilon\rightarrow 0} \frac{1}{\varepsilon} \Big\{ \int_{y\in\mathds{R}} \int_{x\in\mathds{R}} t\cdot y \cdot \mathds{1}(x\in \mathcal{A}_j)\cdot f(x,y,t; \alpha)(1 + \varepsilon s(x,y,t; \alpha))\mbox{d} x \mbox{d} y,
  \\
  &\quad\quad - \int_{y\in\mathds{R}}\int_{x\in\mathds{R}} t\cdot y \cdot \mathds{1}(x\in \mathcal{A}_j) \cdot f(x,y,t; \alpha) \mbox{d}x \mbox{d} y \Big\},\\
=& \int_{y\in\mathds{R}}\int_{x\in\mathds{R}} t\cdot y\cdot  \mathds{1}(x\in \mathcal{A}_j) \cdot f(x,y,t; \alpha) s(x,y,t; \alpha)\mbox{d} x \mbox{d} y,\\
=& \int_{y\in\mathds{R}}\int_{x\in\mathds{R}} t \cdot y \cdot \mathds{1}(x\in \mathcal{A}_j) \cdot \Big(s(y|x\cdot \mathds{1}(x\in \mathcal{A}_j),t ;\alpha)\\
&\cdot f(y|x\cdot \mathds{1}(x\in \mathcal{A}_j),t;\alpha) \cdot f(t|x \cdot \mathds{1}(x\in \mathcal{A}_j);\alpha) \cdot f(x\cdot \mathds{1}(x\in \mathcal{A});\alpha\Big)\mbox{d} x \mbox{d} y, \\
=& \int_{y\in\mathds{R}}\int_{x\in\mathds{R}} t \cdot y \cdot \mathds{1}(x\in \mathcal{A}_j) \cdot \Big(s(y|x\cdot \mathds{1}(x\in \mathcal{A}_j),t ;\alpha)\\
&\cdot \frac{f(y,x\cdot \mathds{1}(x\in \mathcal{A}_j),t;\alpha)} {f(t|x\cdot \mathds{1}(x\in \mathcal{A}_j);\alpha) f(x\cdot \mathds{1}(x\in \mathcal{A}_j);\alpha)}\Big)dx dy,\\
=& \int_{y\in\mathds{R}} \int_{x\in\mathds{R}} t\cdot y\cdot \mathds{1}(x\in \mathcal{A}_j)\cdot 
 \Big(s(y|x\cdot \mathds{1}(x\in \mathcal{A}_j),t;\alpha)\cdot \frac{f(y,x\cdot \mathds{1}(x\in \mathcal{A}_j),t;\alpha)}{ e(x;\alpha)\cdot  f( x\cdot\mathds{1}(x\in \mathcal{A}_j);\alpha)}dxdy,
\end{align*}
given $\int_{x\in\mathds{R}} f(x\cdot \mathds{1}(x\in \mathcal{A}_j;\alpha)dx = P(\mathcal{A}_j;\alpha)$,
\begin{align*}
=&\frac{\mathds{1}(x\in \mathcal{A}_j)}{P(\mathcal{A}_j;\alpha)} \int_{y\in\mathds{R}}\int_{x\in\mathds{R}} \frac{t\cdot y}{e(x;\alpha)}\Big(s(y|x\cdot \mathds{1}(x\in \mathcal{A}_j);\alpha)\cdot
f(y,x\cdot \mathds{1}(x\in \mathcal{A}_j;\alpha)dxdy,\\
=& E\Big[\frac{\mathds{1}(x\in \mathcal{A}_j)}{P(\mathcal{A}_j;\alpha)}\frac{T}{e(X;\alpha)}Y\cdot s(x,y;\alpha)\Big].
\end{align*}
Next, we follow the standard technique and obtain the efficient influence function as 
\begin{align*}
\varphi_{1,j}(Y, D, X) &= \frac{ \mathds{1}(X\in\mathcal{A}_j)  }{\mathds{P}(\mathcal{A}_j)} \Big(\frac{T}{e(X)}\big(Y - p_1(X)\big) + p_1(x)-\alpha_{1,j}\Big),\\
\varphi_{0,j}(Y, D, X) &= \frac{ \mathds{1}(X\in\mathcal{A}_j)  }{\mathds{P}(\mathcal{A}_j)} \Big(\frac{1-T}{1-e(X)}\big(Y - p_0(X)\big)+p_0(X)-\alpha_{0,j}\Big). 
\end{align*}
\end{proof}

\subsection{Delta method}
The efficient influence functions of ATE, the relative risk and the odds ratio can be derived by applying the delta method on the efficient influence functions of $(\bm{\alpha_1}, \bm{\alpha_0})$. 
\begin{align*}
    \bm{\varphi}_{\textbf{ATE}} &= (\bm{\varphi_1}, \bm{\varphi_0})\cdot \frac{\partial f_\text{ATE}(\bm{\alpha_1}, \bm{\alpha_0})}{\partial (\bm{\alpha_1}, \bm{\alpha_0})}\\
    &=
    \begin{pmatrix}
    \varphi_{1,1} &\varphi_{0,1}\\
    \vdots &\vdots\\
    \varphi_{1,J}  & \varphi_{0,J}
    \end{pmatrix}
    \begin{pmatrix}
    1\\
    -1
    \end{pmatrix}
    =
    \begin{pmatrix}
    \varphi_{1,1} - \varphi_{0,1}\\
    \vdots\\
    \varphi_{1,J} - \varphi_{0,J}
    \end{pmatrix},\\
    \bm{\varphi}_{\textbf{OR}} &= (\bm{\varphi_1}, \bm{\varphi_0})\cdot \frac{\partial f_\text{OR}(\bm{\alpha_1}, \bm{\alpha_0})}{\partial (\bm{\alpha_1}, \bm{\alpha_0})}\\
    &=
    \begin{pmatrix}
    \varphi_{1,1} & \varphi_{0,1}\\
    \vdots &\vdots\\
    \varphi_{1,J} & \varphi_{0,J}
    \end{pmatrix}
    \begin{pmatrix}
 \frac{1-\alpha_0}{\alpha_0(1-\alpha_1)^2} \\ -\frac{\alpha_1}{\alpha^2_0(1-\alpha_1)} 
 \end{pmatrix}=
    \begin{pmatrix}
    \frac{1-\alpha_0}{\alpha_0(1-\alpha_1)^2}\cdot\varphi_{1,1} - \frac{\alpha_1}{\alpha^2_0(1-\alpha_1)} \cdot \varphi_{0,1}\\
    \vdots\\
    \frac{1-\alpha_0}{\alpha_0(1-\alpha_1)^2}\cdot\varphi_{1,J} -  \frac{\alpha_1}{\alpha^2_0(1-\alpha_1)}\varphi_{0,J}
    \end{pmatrix},\\
    \bm{\varphi}_{\textbf{RR}} &= (\bm{\varphi_1}, \bm{\varphi_0})\cdot \frac{\partial f_\text{RR}(\bm{\alpha_1}, \bm{\alpha_0})}{\partial (\bm{\alpha_1}, \bm{\alpha_0})}\\
    &=
    \begin{pmatrix}
    \varphi_{1,1} & \varphi_{0,1}\\
    \vdots &\vdots\\
    \varphi_{1,J} & \varphi_{0,J}
    \end{pmatrix}
    \begin{pmatrix}
 \frac{1}{\alpha_0} \\ 
 -\frac{\alpha_1}{\alpha^2_0} 
 \end{pmatrix}=
    \begin{pmatrix}
    \frac{1}{\alpha_0}\cdot\varphi_{1,1} - \frac{\alpha_1}{\alpha^2_0} \cdot \varphi_{0,1}\\
    \vdots\\
    \frac{1}{\alpha_0}\cdot\varphi_{1,J} -  \frac{\alpha_1}{\alpha^2_0}\varphi_{0,J}
    \end{pmatrix}.
\end{align*}

\section{Additional Theoretical Investigations}\label{supp:another-proof-theorem-1}

\subsection{Proposition 1: sufficient conditions on Assumption 4.2}
In this subsection, we provide sufficient conditions on Assumption 4.2, which directly imposes conditions on the updated initial estimator $\hat{p}_t(\cdot)$. Our sufficient conditions are imposed on the initial estimators of the propensity score as well as ${p}_t(\cdot)$.

\begin{assumption}[Initial estimator restrictions]\label{supp:assumption:regularity}
The estimators $\hat{e}_t(X)$ and $\hat{p}^{\text{Init}}_t(X)$ obey the following for a sequence of data generating processes $\{P_n\}$.

\begin{enumerate}[label={(\alph*)}, ref={(\alph*)}]

    \item\label{assumption:regularity:a} $\frac{1}{n}
    \sum_{i=1}^n[(\hat{e}_t(X_i)-e_t(X_i))^2] = o_{P}(1)$ \text{and} $\frac{1}{n}
    \sum_{i=1}^n[(\hat{p}^{\text{Init}}_t(X_i)-p_t(X_i))^2] = o_{P}(1)$;
    \item \label{assumption:regularity:b}
    $
    [\sum_{i=1}^n(\hat{p}^{\text{Init}}_t(X_i)-p_t(X_i))^2/n]^{1/2}[
    \sum_{i=1}^n(\hat{e}_t(X_i)-e_t(X_i))^2/n]^{1/2} =  o_{P}(n^{-1/2})$;
     \item\label{assumption:regularity:c} $\frac{1}{n}
    \sum_{i=1}^n[(\hat{p}^{\text{Init}}_t(X_i)-p_t(X_i))(1-\frac{\mathds{1}(T_i=t)}{e_t(X_i)})] =  o_{P}(n^{-1/2})$.
\end{enumerate}
\end{assumption}

\begin{assumption}[Number of subgroups growth rate]\label{supp:assumption:subgroup} Suppose the number of subgroups under consideration satisfies that  $\text{log}J/n = o(1)$. 
\end{assumption}

\begin{prop}[Iterative update]\label{lemma:replace-pk}
Under Lemma \ref{lemma:epsilon}, Assumption \ref{supp:assumption:regularity} and Assumption \ref{supp:assumption:subgroup}, we conclude $
\hat{p}_1^{(K)}(X) =  \hat{p}_1^{(0)}(X) + K\cdot o_P(1)$,
where $K$ denotes the total number of iterations.
\end{prop}

\begin{coro}\label{lemma:pt-consistency}
Following Proposition \ref{lemma:replace-pk}, 
$\sqrt{n}(\hat{p}_t - p_t) = o_P(1)$. (Proof in Section \ref{appendix:subsec:consistency}). 
\end{coro}



\subsection{Proof of Proposition
\ref{lemma:replace-pk}}\label{appendix:replace-pk}
\begin{proof}
We can rewrite the iterative procedure as 
			\begin{align*}
	   \text{logit} \hat{p}_1^{(1)}( X_i) &=  \text{logit}\big(\hat{p}^{(0)}_1( X_i) \big)
	          + \hat{\varepsilon}^{(1)} \cdot \tilde{S}^{(0)}_1(X_i) , \\
	            \text{logit} \hat{p}_1^{(2)}( X_i) &=  \text{logit}\big(\hat{p}^{(1)}_1( X_i) \big)
	          + \hat{\varepsilon}^{(2)} \cdot \tilde{S}^{(1)}_1(X_i) , \\
	          &= \text{logit}\big(\hat{p}^{(0)}_1( X_i) \big)
	          + \hat{\varepsilon}^{(1)} \cdot \tilde{S}^{(0)}_1(X_i) 
	          + \hat{\varepsilon}^{(2)} \cdot \tilde{S}^{(1)}_1(X_i) , \\
	          \dots\\
	             \text{logit} \hat{p}_1^{(K)}( X_i) 
	          &= \text{logit}\big(\hat{p}^{(0)}_1( X_i) \big)
	          + K\cdot(
	          \hat{\varepsilon}^{(1)} \cdot \tilde{S}^{(0)}_1(X_i)),
			\end{align*}
where $\hat{p}_1^{(0)} = \hat{p}_1^{\text{Init}}$. First, we want to show $\hat{p}^{(1)}_1(X) - \hat{p}^{(0)}_1(X) = \hat{\varepsilon}\tilde{S}_1(X) + R_n =  o_P(1)$.
Applying Taylor expansion we have
\begin{align*}
    \hat{p}^{(1)}_1(X) &= \text{expit}\Big(\text{logit}\big(\hat{p}_1^{(0)}(X)\big) +  \hat{\varepsilon}\cdot\tilde{S}_1(X)\Big),\\
    &= \hat{p}_1^{(0)}(X) + \hat{\varepsilon}\tilde{S}_1(X)\cdot \hat{p}_1^{(0)}(X)(1-\hat{p}_1^{(0)}(X))+R_n,\\
    &= \hat{p}_1^{(0)}(X) + \hat{\varepsilon}\Big( \frac{ \frac{\mathds{1}(X\in\mathcal{A}_j)}{\hat{\mathds{P}}(\mathcal{A}_j)}\cdot \frac{T}{\hat{e}(X)}\cdot \frac{1}{n}\sum_{i=1}^n\Big\{\hat{\phi}_j(Y_i, T_i,X_i)\Big\}}{\sqrt{ \sum_{m=1}^J \left(  \frac{1}{n}\sum_{i=1}^n \hat{\phi}_m(Y_i,T_i,X_i) \right)^2}}\Big)\cdot  \hat{p}^{(0)}_1(X)(1-\hat{p}^{(0)}_1(X)) + R_n.
\end{align*}

Assume $C_1= \frac{1}{\sqrt{ \frac{1}{J}\sum_{m=1}^J \left(  \frac{1}{n}\sum_{i=1}^n \hat{\phi}_m(Y,T,X) \right)^2}}$ is a constant.
\begin{align*}
  \frac{1}{n}\sum_{i=1}^n \hat{p}_1^{(1)}(X_i) &=\frac{1}{n}\sum_{i=1}^n \Big\{\hat{p}_1^{(0)}(X_i) + \hat{\varepsilon}C_1\Big( \frac{\mathds{1}(X_i\in\mathcal{A}_j)}{\hat{\mathds{P}}(\mathcal{A}_j)} \frac{T_i}{\hat{e}(X_i)}\mathds{P}_n\hat{\phi}^{(0)}_j(Y_i, T_i,X_i)\Big)\cdot  \hat{p}^{(0)}_1(X)(1-\hat{p}^{(0)}_1(X_i))\Big\} + R_n,\\
  &= \frac{1}{n}\sum_{i=1}^n \Big\{\hat{p}_1^{(0)}(X_i) + \hat{\varepsilon}C_1\Big( \Big(\frac{\mathds{1}(X_i\in\mathcal{A}_j)}{\hat{\mathds{P}}(\mathcal{A}_j)} \frac{T_i}{\hat{e}(X_i)}\Big)(\mathds{P}_n\hat{\phi}^{(0)}_j(Y_i, T_i,X_i))\Big)\cdot  \hat{p}^{(0)}_1(X_i)(1-\hat{p}^{(0)}_1(X_i)),\\
  &+\hat{\varepsilon}^2C_1^2\Big[\Big(\frac{\mathds{1}(X_i\in\mathcal{A}_j)}{\hat{\mathds{P}}(\mathcal{A}_j)}\cdot \frac{T_i}{\hat{e}(X_i)}\Big)^2(\mathds{P}_n\hat{\phi}^{(0)}_j(Y_i, T_i,X_i))^2\Big]^2\\
  &\quad\quad\cdot  \hat{p}^{(0)}_1(X_i)(1-\hat{p}^{(0)}_1(X_i))(1-2\hat{p}^{(0)}_1(X_i)) \Big\}+ R_n.
\end{align*}
For simplicity, assume we know the true $P(\mathcal{A}_j)$. Under Mean-Value Theorem, assume there exists $\tilde{p}(X)$, such that $\tilde{p}(X) \in [\hat{p}^{(0)}_1(X), \hat{p}^{(1)}_1(X)]$. we know $R_n = \frac{1}{n}\sum_{i=1}^n\Big(\text{expit}\big(\text{logit}(\tilde{p}_1(X_i)\big) \Big)^{'''}\Big(\hat{\varepsilon}\tilde{S}_1(X_i)\Big)^3.$ We want to show $R_n = o_P(1)$, and the first-order and second-order terms in the expansion are also $o_P(1)$. First, we bound the remainder term $R_n$, assuming $\tilde{p}_1(X)$ follows the Assumptions listed in Assumption \ref{supp:assumption:regularity}.
\begin{align*}
    R_n &= \frac{1}{n}\sum_{i=1}^n \Big\{\Big(\text{expit}\big(\text{logit}(\tilde{p}_1(X_i)\big) \Big)^{'''}\Big(\hat{\varepsilon}\tilde{S}_1(X_i)\Big)^3\Big\},\\
    &=\frac{1}{n}\sum_{i=1}^n \Big\{ \Big[\hat{\varepsilon}C_1 \sum_{i=1}^n\Big( \frac{\mathds{1}(X\in\mathcal{A}_j)}{\mathds{P}(\mathcal{A}_j)}\Big)^2\cdot\Big(\frac{T_i}{\hat{e}(X_i)}\Big)^2 \big(Y_i -\tilde{p}(X_i)\big)\Big]^3  \tilde{p}_1(X_i)\big(1-\tilde{p}_1(X_i)\big)\\
    &\cdot \Big[ \Big(1-6\tilde{p}_1(X_i)\big(1-\tilde{p}_1(X_i)\big)\Big)\Big]\Big\},\\
    &= \frac{1}{n}\sum_{i=1}^n \Big\{ \hat{\varepsilon}^3 C_1^3 \Big[ \sum_{i=1}^n\Big( \frac{\mathds{1}(X\in\mathcal{A}_j)}{\mathds{P}(\mathcal{A}_j)}\Big)^2\Big]^3\cdot\Big[\Big(\frac{T_i}{e(X_i)}\Big)^2 p_1(X_i)(1-p_1(X_i)) + A_{1,n} + A_{2,n} 
    \Big]\\
    &\cdot \Big[(1-6p_1(X_i))(1-p_1(X_i))\frac{T_i}{e(X_i)} + B_{1,n} + B_{2,n}\Big]\cdot \Big[\frac{T_i}{e(X_i)}(Y_i - p_1(X_i))+ C_{1,n} + C_{2,n}\Big]^3\Big\},
\end{align*}
where,
\begin{align*}
A_{1,n}&= \frac{1}{n}
    \sum_{i=1}^n \Big(\frac{T_i}{\hat{e}(X_i)}\Big)^2 \tilde{p}_1(X_i)(1-\tilde{p}_1(X_i)) - \Big(\frac{T_i}{\hat{e}(X_i)}\Big)^2 p_1(X_i)(1-p_1(X_i)),\\
A_{2,n}&= \frac{1}{n}
    \sum_{i=1}^n \Big(\frac{T_i}{\hat{e}(X_i)}\Big)^2 p_1(X_i)(1-p_1(X_i)) - \Big(\frac{T_i}{e(X_i)}\Big)^2 p_1(X_i)(1-p_1(X_i)),\\
    B_{1,n}&= \frac{1}{n}
    \sum_{i=1}^n \Big(\frac{T_i}{\hat{e}(X_i)}\Big) (1-6\tilde{p}_1(X_i))(1-\tilde{p}_1(X_i)) - \Big(\frac{T_i}{\hat{e}(X_i)}\Big) (1-6p_1(X_i))(1-p_1(X_i)),\\
    B_{2,n}&= \frac{1}{n}
    \sum_{i=1}^n \Big(\frac{T_i}{\hat{e}(X_i)}\Big) (1-6p_1(X_i))\big(1-p_1(X_i)\big) - \Big(\frac{T_i}{e(X_i)}\Big) (1-6p_1(X_i))(1-p_1(X_i)),\\
    C_{1,n}&=\frac{1}{n}\sum_{i=1}^n \frac{T_i}{\hat{e}(X_i)}(Y_i - \tilde{p}_1(X_i))-\frac{T_i}{\hat{e}(X_i)}(Y_i - p_1(X_i)),\\
     C_{2,n}&=\frac{1}{n}\sum_{i=1}^n\frac{T_i}{\hat{e}(X_i)}(Y_i - p_1(X_i))-\frac{T_i}{e(X_i)}(Y_i - p_1(X_i)).
\end{align*}
We want to show $A_{1,n}$, $A_{2,n}$, $B_{1,n}$, $B_{2,n}$ $C_{1,n}$ and $C_{2,n}$ are all $o_P(n^{-1/2})$. We will use this equality frequently in the following proof:
\begin{align*}
\tag{1}\label{eqn:equality}
    \frac{T}{\hat{e}(X)} = \frac{T}{e(X)} + \frac{T\cdot e(X)-T\cdot \hat{e}(X)}{e(X)\hat{e}(X)}.
\end{align*}
\begin{align*}
A_{1,n}&= \frac{1}{n}
    \sum_{i=1}^n \Big\{ \Big(\frac{T_i}{\hat{e}(X_i)}\Big)^2 \tilde{p}_1(X_i)\big(1-\tilde{p}_1(X_i)\big) - \Big(\frac{T_i}{\hat{e}(X_i)}\Big)^2 p_1(X_i)(1-p_1(X_i))\Big\},\\
   &= \frac{1}{n}
\sum_{i=1}^n\Big\{\Big[\Big(\tilde{p}_1(X_i)-p_1(X_i)\Big) + \Big(p^2_1(X_i)-(\tilde{p}_1(X_i))^2 \Big)\Big] \Big(\frac{T_ie(X_i)-T_i\hat{e}(X_i)}{e(X_i)\hat{e}(X_i)} \Big)^2\Big\},\\
&= \frac{1}{n}
\sum_{i=1}^n\Big\{\Big(\tilde{p}_1(X_i)-p_1(X_i)\Big)\Big(\frac{T_ie(X_i)-T_i\hat{e}(X_i)}{e(X_i)\hat{e}(X_i)} \Big)^2\Big\}\\
&\quad\quad+ \frac{1}{n}
\sum_{i=1}^n \Big\{\Big(p^2_1(X_i)-(\tilde{p}_1(X_i))^2 \Big) \Big(\frac{T_ie(X_i)-T_i\hat{e}(X_i)}{e(X_i)\hat{e}(X_i)} \Big)^2\Big\},\\
&\leq  \Big(\max_{i\leq n}\Big(\frac{T_i}{e(X_i)\hat{e}(X_i)}\Big)^2 \Big(e(X_i)+\hat{e}(X_i)\Big)\Big) \sqrt{\frac{1}{n}\sum_{i=1}^n\Big(\tilde{p}_1(X_i)-p_1(X_i)\Big)^2\frac{1}{n}\sum_{i=1}^n\Big(e(X_i)-\hat{e}(X_i)\Big)^2}\\
&+ \Big(\max_{i\leq n}\Big(\frac{T_i}{e(X_i)\hat{e}(X_i)}\Big)^2 \Big(e(X_i)+\hat{e}(X_i)\Big)\Big) \sqrt{\frac{1}{n}\sum_{i=1}^n\Big(\big(\tilde{p}_1(X_i)\big)^2-p_1^2(X_i)\Big)^2\frac{1}{n}\sum_{i=1}^n\Big(e(X_i)-\hat{e}(X_i)\Big)^2},\\
&= o_P(n^{-1/2}),\\
A_{2,n}&=\frac{1}{n}
\sum_{i=1}^n \Big\{ (1-p_1(X_i))p_1(X_i)\Big[\frac{e(X_i)^2-\hat{e}(X_i)^2}{e(X_i)^2\hat{e}(X_i)^2}\Big]\Big\},\\\
&\leq \Big(\max_{i\leq n}\Big(\frac{T_i}{e(X_i)\hat{e}(X_i)}\Big)^2 \Big(e(X_i)+\hat{e}(X_i)\Big)(1-p_1(X_i)p_1(X_i))\Big)\sqrt{\frac{1}{n}\sum_{i=1}^n \Big(e(X_i)-\hat{e}(X_i)\Big)^2},\\
&= o_P(n^{-1/2}),
\end{align*}
We follow the similar proof for $B_{1,n}$ and $B_{2,n}$:
\begin{align*}
    B_{1,n} 
&\leq  \Big(\max_{i\leq n}\Big(\frac{T_i}{e(X_i)\hat{e}(X_i)}\Big)\Big) \sqrt{\frac{1}{n}\sum_{i=1}^n\Big(\tilde{p}_1(X_i)-p_1(X_i)\Big)^2\frac{1}{n}\sum_{i=1}^n\Big(e(X_i)-\hat{e}(X_i)\Big)^2}\\
&+ \Big(\max_{i\leq n}\Big(\frac{T_i}{e(X_i)\hat{e}(X_i)}\Big)\Big) \sqrt{\frac{1}{n}\sum_{i=1}^n\Big(\big(\tilde{p}_1(X_i)\big)^2-p_1^2(X_i)\Big)^2\frac{1}{n}\sum_{i=1}^n\Big(e(X_i)-\hat{e}(X_i)\Big)^2},\\
&= o_P(n^{-1/2}),\\
B_{2,n}
&\leq \Big(\max_{i\leq n}\Big(\frac{T_i}{e(X_i)\hat{e}(X_i)}\Big) (1-p_1(X_i)p_1(X_i))\Big)\sqrt{\frac{1}{n}\sum_{i=1}^n \Big(e(X_i)-\hat{e}(X_i)\Big)^2},\\
&= o_P(n^{-1/2}),\\
C_{1,n} &= \frac{1}{n}\sum_{i=1}^n \Big(\frac{T_i}{e(X_i)} + \frac{T_i(e(X_i)-\hat{e}(X_i))}{e(X_i)\hat{e}(X_i)}
\Big)(p_1(X_i) - \tilde{p}_1(X_i)),\\
&\leq o_P(n^{-1/2}) +   \Big(\max_{i\leq n}\Big(\frac{T_i}{e(X_i)\hat{e}(X_i)}\Big) \sqrt{\frac{1}{n}\sum_{i=1}^n\Big(\tilde{p}_1(X_i)-p_1(X_i)\Big)^2\frac{1}{n}\sum_{i=1}^n\Big(e(X_i)-\hat{e}(X_i)\Big)^2}\Big),\\
&= o_P(n^{-1/2}),\\
C_{2,n} &= \frac{1}{n}\sum_{i=1}^n\Big\{\frac{T_i(e(X_i)-\hat{e}(X_i))}{e(X_i)\hat{e}(X_i)}(Y_i - p_1(X_i)\Big\},\\
&\leq  \Big(\max_{i\leq n}\Big(\frac{T_i(Y_i-p_1(X_i)}{e(X_i)\hat{e}(X_i)}\Big) \sqrt{\frac{1}{n}\sum_{i=1}^n\Big(e(X_i)-\hat{e}(X_i)\Big)^2}\Big),\\
&= o_P(n^{-1/2}).
\end{align*}
Therefore, $R_n = o_P(1)$ by H$\ddot{\text{o}}$lder's inequality, Assumption \ref{supp:assumption:regularity} and Assumption 4.1 (in the main manuscript). Next, we work on the first-order and the second-order term. Rearrange the first-order term we have:
\begin{align*}
 &\quad \hat{\varepsilon}C_1\sum_{j=1}^J \Big(\frac{\mathds{1}(X_i\in\mathcal{A}_j)}{\mathds{P}(\mathcal{A}_j)}\Big)^2\cdot \frac{T_i}{\hat{e}(X_i)}(Y_i - \hat{p}^{(0)}(X_i))\cdot  \frac{T_i}{\hat{e}(X_i)}\hat{p}^{(0)}(X_i)(1-\hat{p}^{(0)}(X_i)),\\
  &= \hat{\varepsilon}C_1\sum_{j=1}^J \Big(\frac{\mathds{1}(X_i\in\mathcal{A}_j)}{\mathds{P}(\mathcal{A}_j)}\Big)^2\Big[ \frac{T_i}{e(X_i)}(Y_i - p_1(X_i)) + A_{1,n} + A_{2,n}\Big]\cdot \Big[ \frac{T_i}{e(X_i)}p_1(X_i)(1-p_1(X_i)) + B_{1,n} + B_{2,n}\Big],
\end{align*}
where,
\begin{align*}
  A_{1,n} &= \frac{T_i}{\hat{e}(X_i)}(Y_i - \hat{p}_1^{(0)}(X_i))  - \frac{T_i}{\hat{e}(X_i)}(Y_i - p_1(X_i)),\\
  A_{2,n} &= \frac{T_i}{\hat{e}(X_i)}(Y_i - p_1(X_i)) - \frac{T_i}{e(X_i)}(Y_i - p_1(X_i)),\\
  B_{1,n} &= \frac{T_i}{\hat{e}(X_i)}\hat{p}_1^{(0)}(X_i)(1-\hat{p}_1^{(0)}(X_i)) -  \frac{T_i}{\hat{e}(X_i)}p_1(X_i)(1-p_1(X_i)),\\
  B_{2,n}&=   \frac{T_i}{\hat{e}(X_i)}p_1(X_i)(1-p_1(X_i))-\frac{T_i}{e(X_i)}p_1(X_i)(1-p_1(X_i)).
\end{align*}
Following the similar proof as before, $A_{1,n}$,$A_{2,n}$,$B_{1,n}$,$B_{2,n}$ are $o_P(1)$, so the first-order term is $o_P(1)$. Now we want to bound the second-order term.
\begin{align*}
   &\quad\hat{\varepsilon}^2C_1^2\Big(\sum_{j=1}^J \frac{\mathds{1}(X_i\in\mathcal{A}_j)}{\mathds{P}(\mathcal{A}_j)}\Big)^2\Big(\frac{T_i}{\hat{e}(X_i)}\Big)^2\Big(Y_i - \hat{p}_1^{(0)}(X_i)\Big)^2\cdot \Big(\frac{T_i}{\hat{e}(X_i)}\Big)^2 \hat{p}^{(0)}_1(X_i)(1-\hat{p}^{(0)}_1(X_i))(1-2\hat{p}_1^{(0)}(X_i)),\\
   &= \hat{\varepsilon}^2C_1^2\Big(\sum_{j=1}^J \frac{\mathds{1}(X_i\in\mathcal{A}_j)}{\mathds{P}(\mathcal{A}_j)}\Big)^2\Big[\frac{T_i}{e(X_i)}\big(Y_i - p_1(X_i)\big)+A_{1,n} + A_{2,n}\Big]^2\\
   &\cdot \Big[\frac{T_i}{e(X_i)} p_1(X_i)(1-p_1(X_i)) + V_{1,n} + V_{2,n}\Big]\cdot\Big[\frac{T_i}{\hat{e}(X_i)}(1-2p_1(X_i)) + V_{3,n} + V_{4,n}\Big],
\end{align*}
where $A_{1,n}$ and $A_{2,n}$ have shown to be $o_P(n^{-1/2})$ and
\begin{align*}
    V_{1,n} &= \frac{1}{n}\sum_{i=1}^n \Big\{ \frac{T_i}{\hat{e}(X_i)} \cdot \hat{p}^{(0)}_1(X_i)(1-\hat{p}^{(0)}_1(X_i))- \frac{T_i}{\hat{e}(X_i)} \cdot p_1(X_i)(1-p_1(X_i))\Big\},\\
    V_{2,n} &= \frac{1}{n}\sum_{i=1}^n \Big\{ \frac{T_i}{\hat{e}(X_i)} \cdot p_1(X_i)(1-p_1(X_i))- \frac{T_i}{e(X_i)} \cdot p_1(X_i)(1-p_1(X_i))\Big\},\\
    V_{3,n}&= \frac{1}{n}\sum_{i=1}^n \Big\{2\frac{T_i}{\hat{e}(X_i)}(p_1(X_i)-\hat{p}^{(0)}_1(X_i))\Big\},\\
   V_{4,n} &=\frac{1}{n}\sum_{i=1}^n  \Big\{\Big(\frac{T_i}{\hat{e}(X_i)}-\frac{T_i}{e(X_i)}\Big)\Big(1-2p_1(X_i)\Big)\Big\}.
\end{align*}
Now we want to show $V_{1,n}$, $V_{2,n}$, $V_{3,n}$, $V_{4,n}$  are $o_P(n^{-1/2})$.  First, we show $V_{1,n} = o_P(n^{-1/2})$ by proving the sub-components $V_{11,n}$ and $V_{12,n}$ are $o_P(n^{-1/2})$. We will also use Equation (\ref{eqn:equality}) in the following proof.
\begin{align*}
   V_{1,n} &= \frac{1}{n}\sum_{i=1}^n \Big\{ \frac{T_i}{\hat{e}(X_i)} \cdot \Big[\hat{p}^{(0)}_1(X_i)-p_1(X_i)+\Big(p_1^2(X_i)-\big(\hat{p}^{(0)}_1(X_i)\big)^2\Big)\Big]\Big\},\\
   &= \frac{1}{n}\sum_{i=1}^n \Big\{ \Big(\frac{T_i}{e(X_i)}+ \frac{T_i(e(X_i)-\hat{e}(X_i))}{e(X_i)\hat{e}(X_i)}\Big) \cdot \Big[\hat{p}^{(0)}_1(X_i)-p_1(X_i)+\Big(p_1^2(X_i)-\big(\hat{p}^{(0)}_1(X_i)\big)^2\Big)\Big]\Big\},\\
   &= \underbrace{\frac{1}{n}\sum_{i=1}^n \Big\{ \frac{T_i}{e(X_i)} \cdot \Big[\hat{p}^{(0)}_1(X_i)-p_1(X_i)+\Big(p_1^2(X_i)-\big(\hat{p}^{(0)}_1(X_i)\big)^2\Big)\Big]}_{V_{11,n}}\Big\},\\
   &+ \underbrace{\frac{1}{n}\sum_{i=1}^n  \Big\{ \frac{T_i(e(X_i)-\hat{e}(X_i))}{e(X_i)\hat{e}(X_i)} \cdot \Big[p_1(X_i)- \hat{p}^{(0)}_1(X_i)+\Big(p_1^2(X_i)-\big(\hat{p}^{(0)}_1(X_i)\big)^2\Big)\Big]}_{V_{12,n}}\Big\}.
\end{align*}
where
\begin{align*}
    V_{11,n}&= \frac{1}{n}\sum_{i=1}^n \Big\{ \frac{T_i}{e(X_i)} \cdot \Big[\hat{p}^{(0)}_1(X_i)-p_1(X_i)\Big]+
    \frac{1}{n}\sum_{i=1}^n  \frac{T_i}{e(X_i)} \cdot \Big[ p^2_1(X_i)-(\hat{p}^{(0)}_1(X_i))^2\Big]\Big\}.
\end{align*}
Given 
\begin{align}\label{eq:taylor-expansion}
    \frac{1}{n}\sum_{i=1}^n\Big[\Big(p^2_1(X_i)-\big(\hat{p}^{(0)}_1(X_i)\big)^2\Big)\cdot \frac{T_i}{e(X_i)}\Big] =  o_{P}(\frac{1}{\sqrt{n}}),
\end{align}
and Assumption \ref{supp:assumption:regularity}\ref{assumption:regularity:a}, $V_{11,n} = o_P(n^{-1/2})$.
\begin{align*}
    V_{12,n} &= \frac{1}{n}\sum_{i=1}^n  \Big\{ \frac{T_i(e(X_i)-\hat{e}(X_i))}{e(X_i)\hat{e}(X_i)} \cdot \Big[p_1(X_i)-\hat{p}^{(0)}_1(X_i)\Big]\Big\}+
    \frac{1}{n}\sum_{i=1}^n  \Big\{ \frac{T_i(e(X_i)-\hat{e}(X_i))}{e(X_i)\hat{e}(X_i)} \Big[\Big(p_1^2(X_i)-\big(\hat{p}^{(0)}_1(X_i)\big)^2\Big)\Big]\Big\},\\
    &\leq \max_{i\leq n}(\frac{T_i}{e(X_i)\hat{e}(X_i)})\sqrt{\mathds{P}_n \big[\big(e(X_i)-\hat{e}(X_i)\big)^2\big]\mathds{P}_n \big[ \big(\hat{p}^{(0)}_1(X_i)-   p_1(X_i)\big)^2\big]}\\
    &+ \max_{i\leq n}(\frac{T_i}{e(X_i)\hat{e}(X_i)})\sqrt{\mathds{P}_n\big[ \big(e(X_i)-\hat{e}(X_i)\big)^2\big]\cdot \mathds{P}_n\big[\big(\hat{p}^{(0)}_1(X_i)-   p_1(X_i)\big)^2\big]}\sqrt{\mathds{P}_n \big[ \big(\hat{p}^{(0)}_1(X_i)-   p_1(X_i)\big)^2\big]},\\
    &= o_P(n^{-1/2}).
\end{align*}
by H$\ddot{\text{o}}$lder's inequality, Assumption \ref{supp:assumption:regularity} \ref{assumption:regularity:a}, \ref{supp:assumption:regularity} \ref{assumption:regularity:b} and Assumption 4.1 (in the main manuscript). We have shown $V_{1,n} = V_{11,n} + V_{12,n} = o_P(n^{-1/2})$. Next, we want to show $V_{2,n} = o_P(n^{-1/2})$ by bounding the second moment.
 \begin{align*}
   \sqrt{n}V_{2,n} &= \frac{1}{\sqrt{n}}\sum_{i=1}^n \Big(\frac{T_i}{\hat{e}(X_i)} - \frac{T_i}{e(X_i)}\Big) \cdot p_1(X_i)(1-p_1(X_i)),\\
   &= \frac{1}{\sqrt{n}}\sum_{i=1}^n  \Big(\frac{T_i(e(X_i)-\hat{e}(X_i))}{\hat{e}(X_i)e(X_i)} \Big) \cdot p_1(X_i)(1-p_1(X_i)),\\
   \mathds{E}[V_{2,n}^2] &= \frac{1}{n}\sum_{i=1}^n  \Big(\frac{T_i^2 p_1^2(X_i)(1-p_1(X_i))^2}{\hat{e}^2(X_i)e^2(X_i)} \Big) \cdot (e(X_i)-\hat{e}(X_i))^2,\\
   &\leq C\cdot \frac{1}{n}\sum_{i=1}^n \Big\{ (e(X_i)-\hat{e}(X_i))^2\Big\} = o_P(1).
\end{align*}
The proof of $V_{3,n}$ is similar to $V_{1,n}$, and the proof of $V_{4,n}$ is similar to $V_{2,n}$. By H$\ddot{\text{o}}$lder's inequality, Assumption \ref{supp:assumption:regularity} \ref{assumption:regularity:a}, \ref{supp:assumption:regularity} \ref{assumption:regularity:b} and Assumption 4.1 (in the main manuscript).
\begin{align*}
    V_{3,n} &=\frac{1}{n}\sum_{i=1}^n  \frac{T_i}{e(X_i)}(p_1(X_i)-\hat{p}_1(X_i)) + \frac{1}{e(X_i)\hat{e}(X_i)}(e(X_i)-\hat{e}(X_i))(p_1(X_i)-\hat{p}_1(X_i)),\\
     &\leq \max_{i\leq n}(\frac{T_i}{e(X_i)\hat{e}(X_i)})\sqrt{\mathds{P}_n \big[\big(e(X_i)-\hat{e}(X_i)\big)^2\big]\mathds{P}_n \big[ \big(\hat{p}^{\text{Init}}_1(X_i)-   p_1(X_i)\big)^2\big]},\\
    &= o_P(n^{-1/2}),\\
    V_{4,n}&= \frac{1}{n}\sum_{i=1}^n \frac{T_i(e(X_i)-\hat{e}(X_i))}{e(X_i)\hat{e}(X_i)}(1-2p_1(X_i)) = o_P(n^{-1/2}),\\
    \mathds{E}[V_{4,n}^2] &= \frac{1}{n}\sum_{i=1}^n  \Big(\frac{T_i^2 p_1^2(X_i)(1-2p_1(X_i))^2}{\hat{e}^2(X_i)e^2(X_i)} \Big) \cdot (e(X_i)-\hat{e}(X_i))^2,\\
   &\leq C\cdot \frac{1}{n}\sum_{i=1}^n \Big\{ (e(X_i)-\hat{e}(X_i))^2\Big\} = o_P(1).
\end{align*}
In conclusion,
\begin{align*}
    \hat{p}^{(1)}_1(X) &=  \hat{p}^{(0)}_1(X) + \varepsilon \tilde{S}_1(X) + o_P(1),\\
    &=  \hat{p}^{(0)}_1(X) + o_P(1),\\
     \hat{p}^{(K)}_1(X) &= \hat{p}^{(0)}_1(X) + K\cdot o_P(1),
\end{align*}
\end{proof}

\subsection{Lemma \ref{lemma:epsilon}}\label{appendix:epsilon}

\begin{lem}[$\hat{\varepsilon}$ convergence rate]\label{lemma:epsilon}
$\hat{\varepsilon} - \varepsilon = H_{n}^{-1}\cdot C\cdot \sum_{m=1}^J \left(  \frac{1}{n}\sum_{i=1}^n \phi_m(Y_i,T_i,X_i) \right)^2 + o_P(n^{-1/2})$, where $$H_n^{-1}=\Big\{\mathds{E} \Big[\Big(1-p_1(X)\Big)\cdot p_1(X)\cdot \frac{\Big(\frac{\mathds{1}(X\in\mathcal{A}_j)}{\mathds{P}(\mathcal{A}_j)}\frac{T}{e(X)} \big(\frac{1}{n}\sum_{i=1}^n\phi_j(Y_i, T_i,X_i)\big)^2\Big)^2}{ \sum_{m=1}^d \left(  \frac{1}{n}\sum_{i=1}^n \phi_m(Y_i,T_i,X_i) \right)^2}\Big]\Big\}^{-1},$$ and
$C =\Big(\sqrt{ \sum_{j=1}^J \left(  \frac{1}{n}\sum_{i=1}^n \hat{\phi}_j(Y_i,T_i,X_i) \right)^2}\Big)^{-1}$.
\end{lem}

\begin{proof}
\begin{align*}
   \hat{\varepsilon} - \varepsilon &= \Big\{\mathds{E} \Big[\Big(1-p_1(X)\Big)\cdot p_1(X)\cdot \tilde{S}_1^2(X)\Big]\Big\}^{-1}\cdot \frac{1}{n}\sum_{i=1}^n \tilde{S}_{1}(X_i)\Big(Y_i - \hat{p}^{\text{Init}}_1(X_i)\Big),\\
    &= \Big\{\mathds{E} \Big[\Big(1-p_1(X)\Big)\cdot p_1(X)\cdot \frac{\Big(\frac{\mathds{1}(X\in\mathcal{A}_j)}{\mathds{P}(\mathcal{A}_j)}\frac{T}{e(X)} \big(\frac{1}{n}\sum_{i=1}^n\phi_j(Y_i, T_i,X_i)\big)^2\Big)^2}{ \sum_{m=1}^J \left(  \frac{1}{n}\sum_{i=1}^n \phi_m(Y_i,T_i,X_i) \right)^2}\Big]\Big\}^{-1}\\
    &\cdot \frac{\sum_{j=1}^J \big(\frac{1}{n} \sum_{i=1}^n \hat{\phi}^{\text{Init}}_j(Y_i, T_i,X_i)\big)^2}{\sqrt{ \sum_{m=1}^J \left(  \frac{1}{n}\sum_{i=1}^n \hat{\phi}^{\text{Init}}_m(Y_i,T_i,X_i) \right)^2}} + R_n.
\end{align*}
Assume $C = \frac{1}{\sqrt{ \sum_{m=1}^J \left(  \frac{1}{n}\sum_{i=1}^n \hat{\phi}^{\text{Init}}_m(Y_i,T_i,X_i) \right)^2}}$, we lay out the proof as the following: First, we work with the first-order term using linearization. Second, we derive the convergence rate of the second-order term by bounding the sample Hessian matrix. For simplicity, we assume $\mathds{P}(\mathcal{A}_j)$ is known.
\begin{align*}
   \hat{\varepsilon} - \varepsilon
    &= H_n^{-1}\cdot \Big\{C\cdot \sum_{j=1}^J \Big(\frac{1}{n}\sum_{i=1}^n \hat{\phi}_j^{\text{Init}}(Y_i,T_i,X_i)\Big)^2\Big)\Big\} + (\hat{H}_n^{-1}-H_n^{-1})\cdot \Big\{C\cdot \sum_{j=1}^J \Big(\frac{1}{n}\sum_{i=1}^n \hat{\phi}_j^{\text{Init}}(Y_i,T_i,X_i)\Big)^2\Big)\Big\},\\
    & = H_n^{-1}\Big\{C \sum_{j=1}^J \Big(\frac{1}{n}\sum_{i=1}^n \phi_j(Y_i,T_i,X_i)\Big)^2\Big) +U_n\Big\}\\ &+\hat{H}_n^{-1}(\hat{H}_n-H_n)H_n^{-1}\Big\{C \sum_{j=1}^J \Big(\frac{1}{n}\sum_{i=1}^n \phi_j(Y_i,T_i,X_i)\Big)^2\Big) +U_n\Big\},
    \end{align*}
where $H_n^{-1}$ is a constant, $\varepsilon = 0$, and
\begin{align}\label{eq:hessian}
    \hat{H}_n - H_n = o_P(n^{-1/2}).
\end{align}
 To complete the proof, we only need to show $U_n = o_P(n^{-1/2})$.
\begin{align*}
    U_n &=  \Big(\frac{1}{n}\sum_{i=1}^n \sum_{j=1}^J\Big\{\hat{\phi}^{\text{Init}}_j(Y_i,T_i,X_i) - \phi_j(Y_i,T_i,X_i)\Big\}\Big) \Big(\frac{1}{n}\sum_{i=1}^n \sum_{j=1}^J\Big\{ \hat{\phi}^{\text{Init}}_j(Y_i,T_i,X_i) - \phi_j(Y_i,T_i,X_i)\Big\}\Big)\\
    &\quad+ 2\Big(\frac{1}{n}\sum_{i=1}^n \hat{\phi}^{\text{Init}}_j(Y_i,T_i,X_i) \phi_j(Y_i,T_i,X_i)\Big),
\end{align*}
where,
\begin{align*}
    2\Big(\frac{1}{n}\sum_{i=1}^n\sum_{j=1}^J \hat{\phi}^{\text{Init}}_j(Y_i,T_i,X_i) \phi_j(Y_i,T_i,X_i)\Big) &= 2\Big(\frac{1}{n}\sum_{i=1}^n \sum_{j=1}^J\Big\{ \frac{\mathds{1}(X_i\in\mathcal{A}_j)}{\mathds{P}(\mathcal{A}_j)}\frac{T_i}{\hat{e}(X_i)}(Y_i-\hat{p}_1^{\text{Init}}(X_i))\\
    &\cdot \frac{\mathds{1}(X_i\in\mathcal{A}_j)}{\mathds{P}(\mathcal{A}_j)}\frac{T_i}{e(X_i)}(Y_i-p_1(X_i))\Big\}\Big),\\
    &= o_P(1),
\end{align*}
by Assumption \ref{supp:assumption:regularity}\ref{assumption:regularity:a} and Assumption 4.1 (in the main manuscript). Next, we only need to show $\Big(\frac{1}{n}\sum_{i=1}^n \hat{\phi}^{\text{Init}}_j(Y_i,T_i,X_i) - \phi_j(Y_i,T_i,X_i)\Big) = o_P(n^{-1/2})$.
\begin{align*}
    \frac{1}{n}\sum_{i=1}^n \sum_{j=1}^J \Big\{\hat{\phi}^{\text{Init}}_j(Y_i,T_i,X_i) - &\phi_j(Y_i,T_i,X_i)\Big\} = \frac{1}{n}\sum_{i=1}^n \sum_{j=1}^J \frac{\mathds{1}(X_i)\in\mathcal{A}_j}{\mathds{P}(\mathcal{A}_j)} \Big(\frac{T_i}{\hat{e}(X_i)} - \frac{T_i}{e(X_i)}\Big) \cdot \big(\hat{p}_1^{\text{Init}}(X_i) - p_1(X_i)\big),\\
    &= \frac{1}{n}\sum_{i=1}^n \sum_{j=1}^J \frac{\mathds{1}(X_i)\in\mathcal{A}_j)}{\mathds{P}(\mathcal{A}_j)} \Big(\frac{T_i(\hat{e}(X_i) - e(X_i))}{\hat{e}(X_i)e(X_i)}\Big) \cdot \big(\hat{p}_1^{\text{Init}}(X_i) - p_1(X_i)\big),\\
&\leq \Big(\max_{i\leq n}\frac{T_i}{\hat{e}(X_i)e(X_i)}\Big)\sqrt{\frac{1}{n}\sum_{i=1}^n\big(\hat{e}(X_i) - e(X_i)\big)^2 \frac{1}{n}\sum_{i=1}^n\big(\hat{p}_1^{\text{Init}}(X_i)- p_1(X_i)\big)^2},\\
    &= o_P(n^{-1/2}),
\end{align*}
by H$\ddot{\text{o}}$lder's inequality, Assumption \ref{supp:assumption:regularity}\ref{assumption:regularity:b} and Assumption 4.1 (in the main manuscript). Therefore, $U_n = o_P(n^{-1/2})$.
\end{proof}

\subsection{Proof of Corollary 1}\label{appendix:subsec:consistency}
Assume our parameter of interest is the single subgroup treatment effect under the treated, $p_1(\mathcal{A}) = \mathds{P}(Y(1)|X\in\mathcal{A})$, where $\mathcal{A}$ denotes the single subgroup of interest. 


\begin{proof}
By taking derivative of the loss function with respect to $\varepsilon$, we have: 
\begin{align*}
 &\frac{1}{n}\sum_{i=1}^n \frac{\mathds{1}(X_i\in\mathcal{A})}{\mathds{P}(\mathcal{A})}\Big\{\Big(Y_i - \hat{p}_1(X_i)\Big)\frac{T_i}{\hat{e}(X_i)} + \hat{p}_1(X_i) - \hat{p}_1\Big\} = 0,\\
\hat{p}_1(\mathcal{A}) &= \frac{1}{n}\sum_{i=1}^n \frac{\mathds{1}(X_i\in\mathcal{A})}{\mathds{P}(\mathcal{A})}\Big\{\Big(Y_i - \hat{p}_1(X_i)\Big)\frac{T_i}{\hat{e}(X_i)} + \hat{p}_1(X_i)\Big\},\\
    \hat{p}_1(\mathcal{A}) - p_1(\mathcal{A}) &=  \frac{1}{n}\sum_{i=1}^n \frac{\mathds{1}(X_i\in\mathcal{A})}{\hat{\mathds{P}}(\mathcal{A})}\Big\{ \Big(Y_i - \hat{p}_1(X_i)\Big)\frac{T_i}{\hat{e}(X_i)} +\hat{p}_1(X_i)- p_1\Big\},\\
    &= \frac{1}{n}\sum_{i=1}^n \frac{\mathds{1}(X_i\in\mathcal{A})}{\mathds{P}(\mathcal{A})}\Big\{ \Big(Y_i - p_1(X_i)\Big)\frac{T_i}{e(X_i)} + p_1(X_i)-p_1 + R_{1,n} + R_{2,n} + R_{3,n}\Big\},
\end{align*}
where the three remainder terms are:
\begin{align*}
    R_{1,n} &=   \frac{1}{n}\sum_{i=1}^n \Big\{ \Big(Y_i - \hat{p}_1(X_i)\Big)\Big(\frac{T_i}{\hat{e}(X_i)}-1\Big)\Big\} - \frac{1}{n}\sum_{i=1}^n \Big\{ \Big(Y_i - p_1(X_i)\Big)\Big(\frac{T_i}{\hat{e}(X_i)}-1\Big)\Big\},\\
    R_{2,n}&= \frac{1}{n}\sum_{i=1}^n \Big\{ \Big(Y_i - p_1(X_i)\Big)\Big(\frac{T_i}{\hat{e}(X_i)}-1\Big)\Big\} - \frac{1}{n}\sum_{i=1}^n \Big\{ \Big(Y_i - p_1(X_i)\Big)\Big(\frac{T_i}{e(X_i)}-1\Big)\Big\},\\
    R_{3,n}&= \Big(\frac{1}{\hat{\mathds{P}}(\mathcal{A})} - \frac{1}{\mathds{P}(\mathcal{A})}\Big)\cdot \Big(\frac{1}{n}\sum_{i=1}^n \mathds{1}(X_i\in\mathcal{A})  \cdot \Big\{ \Big(Y_i - p_1(X_i)\Big)\frac{T_i}{e(X_i)} + p_1(X_i)-p_1 + R_{1,n} + R_{2,n} \Big\}\Big).
\end{align*}
First, we want to show $R_{1,n}$ and $R_{2,n}$  are $o_P(n^{-1/2})$. Second, we aim to show $R_{3,n} = o_P(1)$. 

\begin{align*}
  R_{3,n}&=  \Big(\frac{1}{\hat{\mathds{P}}(\mathcal{A})} - \frac{1}{\mathds{P}(\mathcal{A})}\Big)\cdot o_P(n^{-1/2}), \\
  \frac{1}{\hat{\mathds{P}}(\mathcal{A})} - \frac{1}{\mathds{P}(\mathcal{A})}&= \frac{\mathds{P}(\mathcal{A})-\hat{\mathds{P}}(\mathcal{A}) }{\hat{\mathds{P}}(\mathcal{A})\mathds{P}(\mathcal{A})},\\
    &\leq\max\frac{1}{\hat{\mathds{P}}(\mathcal{A})\mathds{P}(\mathcal{A})} \cdot \Big(\hat{\mathds{P}}(\mathcal{A}) - \mathds{P}(\mathcal{A})\Big),\\
    &= o_P(1) \cdot \Big(\hat{\mathds{P}}(\mathcal{A}) - \mathds{P}(\mathcal{A})\Big),\\
    &= o_P(1).
\end{align*}
by H$\ddot{\text{o}}$lder's inequality and Assumption \ref{supp:assumption:regularity}. Since $R_{1,n}$, $R_{2,n} = o_P(n^{-1/2})$ and $R_{3,n} = o_P(1)$, we conclude the remainder term is $o_P(1)$.
\end{proof}

	
\end{document}